\documentclass[11pt]{article}
\usepackage[T1]{fontenc}
\usepackage{geometry}
\usepackage{natbib}
\usepackage{float}
\usepackage{tcolorbox}
\usepackage{hyperref}
\usepackage{arydshln}
\usepackage{bm}
\usepackage{setspace}
\usepackage{tikz}
\usepackage{bbm}
\usepackage{adjustbox}
\usepackage{graphics}
\usepackage{xcolor}
\usepackage{arydshln}
\usepackage{stackengine}
\stackMath
\newcommand\tsup[2][2]{%
	\def\useanchorwidth{T}%
	\ifnum#1>1%
	\stackon[-.5pt]{\tsup[\numexpr#1-1\relax]{#2}}{\scriptscriptstyle\sim}%
	\else%
	\stackon[.5pt]{#2}{\scriptscriptstyle\sim}%
	\fi%
}

\usepackage{amsmath}
\usepackage{amsthm}
\usepackage{amssymb}
\usepackage{booktabs,caption}
\usepackage{threeparttable}
\usepackage{graphicx}
\usepackage{color}
\providecommand{\keywords}[1]{{\textit{Keywords---}} #1}
\providecommand{\JELclassification}[1]{{\textit{JEL---}} #1}
\usepackage{bbm}
\makeatletter
\theoremstyle{plain}
\newtheorem{thm}{\protect\theoremname}
  \theoremstyle{plain}
  \newtheorem{assumption}{\protect\assumptionname}[section]
  \theoremstyle{plain}
  \newtheorem{prop}{\protect\propositionname}[section]
  \theoremstyle{plain}
  \newtheorem{lem}{\protect\lemmaname}[section]
  \theoremstyle{plain}
  \newtheorem{property}{Property}
  \theoremstyle{plain}

  \newtheorem{Corollary}{Corollary}[section]

\newtheorem{remark}{Remark}[section]

\newtheorem{definition}{Definition}[section]
\makeatother

\usepackage[english]{babel}
  \providecommand{\assumptionname}{Assumption}
  \providecommand{\lemmaname}{Lemma}
  \providecommand{\propositionname}{Proposition}
\providecommand{\theoremname}{Theorem}
\usepackage{thmtools}
\renewcommand\thmcontinues[1]{Continued}
\usepackage{authblk}
\numberwithin{equation}{section}
\declaretheorem[style=plain]{example}
\begin{document}
	\title{Extending Economic Models with Testable Assumptions: Theory and Applications}
	\author[1]{Moyu Liao \thanks{ I would like to thank Marc Henry for his invaluable advice and encouragement. I also thank Andres Aradillas-Lopez and Keisuke Hirano, Michael Gechter, Patrik Guggenberger, Sun Jae Jun, and Joris Pinkse for their useful comments.}}
	\affil[1]{Nanjing University}
	\maketitle

	\begin{abstract}
		\begin{spacing}{1.2}
		{This paper studies the identification and hypothesis testing in complete and incomplete economic models with testable assumptions. A testable assumption ($A$) gives interpretable empirical content to the economic model but it also carry the possibility that some distributions of observed outcomes may reject these assumptions. A way to avoid the data rejection problem is to find a relaxed assumptions ($\tilde{A}$) that cannot be rejected by any distribution of observed outcomes. We also want the identified set for the parameter of interest under $\tilde{A}$ is not changed when the original assumption $A$ is not rejected by the observed data distribution. I characterize the properties of such a relaxed assumption $\tilde{A}$ using a generalized notion of refutability and confirmability. I also propose a general method to construct such $\tilde{A}$. I apply my methodology to the instrument monotonicity assumption in Local Average Treatment Effect (LATE) estimation and to the sector selection assumption in a binary outcome Roy model of employment sector choice.  In the LATE application, I use my general method to construct a relaxed assumption $\tilde{A}$ that can never be rejected, and the identified set for LATE is unchanged when $A$ holds. LATE is point identified under my extension $\tilde{A}$ in the application. In the binary outcome Roy model, I use my method to relax Roy's sector selection assumption and characterize the identified set for the binary potential outcomes as a polyhedron.}
		\end{spacing}
	

	\JELclassification{ C12, C13, C18, C51, C52}

	\end{abstract}

	\pagebreak
	
	\section{Introduction}
		Empirical researchers often make convenient model assumptions in structural estimation. These assumptions usually come from economic theories or intuitions. For example, the `No Defiers' assumption in \citet{IA1994} assumes that the instrument has a monotone effect on the decision to take treatment; the `Pure Strategy Nash Equilibrium' assumption in \citet{BresnahanReiss1991} assumes only pure strategy Nash Equilibrium is played in a $2\times 2$ entry game;  the `Perfect Self Selection' assumption in \citet{Roy1951}  assumes employees perfectly observe their future earnings and choose a job sector to maximize discounted lifetime earnings. Such assumptions simplify the identification and estimation problems, and make the results easier to interpret. To study structural models and economic assumptions, I generalize the language of econometric structures \citep{koopmans1950identification} to incomplete structures. A (generalized) econometric structure includes a distribution of some exogenously given latent and observed variables, and a correspondence from the distribution of these exogenous variables to a set of distributions of observed variables. Assumptions are restrictions on the economic structures to reflect empirical researchers' understanding of the economic environment.
		
	 Unfortunately, assumptions in the three examples above, when combined with some other reasonable assumptions, can be rejected by some distributions of observables  \citep{kitagawa2015,mourifie2017testing,mourifie2018roy}. When the imposed assumption is refuted by data, the econometrician have an empty identified set for the parameter of interest. As a result, the econometrician cannot give a useful interpretation of the economic environment. A refutable assumption $A$ also imposes challenges to the interpretation of hypothesis testing of parameter values. For example, when we reject the null hypothesis that the parameter equals zero under $A$, it can either be the case that the true parameter equals another value, or the case that $A$ is rejected data. These two cases cannot be distinguished but they have quite different interpretations. 
		
	A way to prevent the data rejection problem is to find a relaxed assumption $\tilde{A}$ so that no distributions of observables can reject $\tilde{A}$. We call this the non-refutability criterion. By imposing a non-refutable $\tilde{A}$ before confronting the data, practitioners avoid the ex-post possibility of finding data evidence against their assumption. Therefore, non-refutability is the first criterion for the relaxed assumption $\tilde{A}$ to satisfy. On the other hand, we also do not want to deviate from the old assumption $A$, since it still reflects the economic theory behind it. Specifically, given a parameter of interest $\theta$ as a function of structures, we want the (sharp) identified set under the relaxed assumption $\tilde{A}$ to be equal to the (sharp) identified set under $A$, when $A$ is not rejected by the observed distribution. In other words, we want to preserve the identified set.  
		
	This paper aims to do three things: First, I formalize and extend the definition of refutability and confirmability of assumptions in \citet{breusch1986} using the language of generalized economic structures. These definitions are useful to characterize properties of a relaxed assumption $\tilde{A}$.  I characterize conditions that a relaxed assumption $\tilde{A}$ needs to satisfy so that: (a).No distribution of observables can be rejected under $\tilde{A}$; (b).The identified sets for the parameter of interest under $A$ and $\tilde{A}$  are the same whenever $A$ is not rejected by the observed data distribution.  I show that when structures are complete, a relaxed assumption $\tilde{A}$ that satisfies the two properties above always exists. I also characterize a sufficient condition for the existence of $\tilde{A}$ in case of incomplete structures. The possible failure to find $~\tilde{A}$ in incomplete structures encourages researchers to complete the structures, and then find a nice relaxed assumption in the completed structure universe. When the structures are complete, I also provide a general method to construct $\tilde{A}$ from $A$. 
	
	Second, I discuss the problem of testing hypothesis on structural parameters' values. I show that any statistical test cannot achieve pointwise size control and test consistency simultaneously when the null hypothesis of structural parameters does not induce a partition of the space of distribution of observables. This is an ill-behaved null hypothesis, and policy decisions based on the result of hypothesis testing can be problematic. Conversely, with a well-behaved null hypothesis, which induces a partition of the space of distribution of observables, I show the existence of statistical tests that achieve pointwise size control and test consistency simultaneously under mild conditions. I also show that when the parameter of interest $\theta$ is point identified, a null hypothesis on the value of $\theta$ is always well-behaved. When structures are complete, I also provide a way to minimally extend  (resp. shrink) the null hypothesis set such that the extend (resp. shrink) hypothesis is well-behaved.
	
	Third, I look at 2 applications with complete and incomplete 
	structures respectively. In the complete structure framework, I look at the identification of the local average treatment effects (LATE). \citet{kitagawa2015} provides the sharp testable implication of the Imbens and Angrist Monotonicity assumption (IA-M). Therefore, practitioners should anticipate the IA-M to be rejected by some distributions of observables. I provide several relaxations of the IA-M that cannot be rejected by any distribution of observables. The identified sets for LATE under these relaxed assumptions equal the identified set for LATE under the IA-M whenever the IA-M is not rejected by the distributions of observables. One relaxed assumption allows for defiers, and relaxes the independent instrument assumption. The logic of the relaxed assumption is to allow a minimal mass of defiers. It can be shown that the relaxed assumption not only preserves the identified set for LATE, but also preserves the identified set for other parameters of interests such as ATE or ATT. LATE is point identified under this relaxed assumption. I provide an estimator of the LATE which has a normal limit distribution. {To emphasis the fact that a non-refutable relaxation is not unique, I propose two other relaxations, each one of which may be preferable in some contexts.} I apply the method to \citet{card1993using}, and I show the local average treatment effect of education on earnings. Compared to naively using identification result under the IA-M assumption, my method delivers more reasonable sign and scale for the LATE estimates.
	
	For incomplete structures, I look at a binary outcome job sector selection model with a monotone instrument. In the sector choice model,  the Roy assumption does not specify the sector choice rule in case of ties, which may lead to multiple predicted distributions of observables.  After completing the structures, I then use a `minimal efficiency loss' criterion to characterize the relaxed assumption. The identified set of job sector potential outcome distribution can be characterized as a polyhedron. 

	\paragraph{Related Literature}  \citet{masten2018}
	propose an ex-post way to salvage a refutable assumption $A$. Their ex-post method characterizes a relaxation of $A$ after the distribution of observables is realized. This paper also relates to the literature that relaxes assumption make model robust to misspecification. In the macroeconomic literature, researchers use robust control to avoid the misspecification issue in their baseline model (see \citet{hansen2006robust} and \citet{hansen2007recursive}). The Robust control approach aims to accommodate local perturbations to the baseline model rather than to solve the refutability of the baseline model.\footnote{The perturbation is usually measured by relative entropy the in macroeconomic literature.} It may be true that the baseline model is not refutable by any data distribution. See  \citet{bonhomme2018minimizing} and \citet{christensen2019counterfactual} for more discussion. 
	
	{This paper also contributes to the literature that relaxes the IA-M assumption. \citet{chaisemartin2018Defier} discusses the economic meaning of the conventional LATE quantity $LATE^{Wald}\equiv ({E[Y_i|Z_i=1]-E[Y_i|Z_i=0]})/({E[D_i|Z_i=1]-E[D_i|Z_i=0]})$ when there are defiers. He shows that the $LATE^{Wald}$ identifies the net average treatment effect of a subgroup of compliers after deducting the average treatment effect of defiers.}
	
	The rest of the paper is organized as follows. Section 2 describes a theory of characterizing a refutable assumption $\tilde{A}$, finding a relaxed assumption $\tilde{A}$ and hypothesis testing. Core definitions in Section 2 are followed by shadowed links where their corresponding illustrations can be found. Section 3 applies complete structure theory to the model in \citet{IA1994}. Section 4 applies the theory to a binary outcome Roy model. Main proofs are collected in Appendices. Additional proofs and results are collected in the Online Appendices.

\subsubsection*{Notations}
Throughout this paper, I use $X$ to denote the vector of observed variables, and I use $F$ to denote the distribution of $X$. I use $\epsilon$ to denote the vector of latent variables and some observed variables, and I use $G$ to denote the distribution of $\epsilon$. I use $s$ to denote an economic structure. I use $A$ to denote a refutable assumption and use $\tilde{A}$ to denote a relaxed assumption. I use $\mathbb{F}_n$ to denote the empirical distribution of $F$.

\section{A Theory of Identification and Hypothesis Testing}
In this section, I develop a theory for identification and refutable assumptions. I will then discuss the problem of hypothesis testing. I start with a definition of the observation space.
\begin{definition}\label{def: observation space}
The observation space $\mathcal{F}$ is the collection of all possible distribution of $F(X)$. \colorbox{lightgray}{See \ref{eq: appli, observation space} for an illustration.}
\end{definition}

The distribution of observables $F(X)$ is generated by some distribution of underlying random \mbox{vectors $\epsilon$} through some mapping $M$.  A pair of  a distribution of $\epsilon$ and a mapping $M$ is called an econometric structure.  The following definition of econometric structure is a reformulation of the economic structure defined in \citet{koopmans1950identification} and \citet{jovanovic1989}. Since in most econometric problems, we focus on the distribution of outcomes $F$ instead of how each $X$ is related to $\epsilon$, I directly define the mapping $M$ as a correspondence from the space of underlying variable distributions to $\mathcal{F}$.

\begin{definition}\label{def: economic structure}
	An econometric structure (Model) $s=(G^s,M^s)$ consists of a distribution $G^s$, and an outcome mapping $M^s$. Let $\mathcal{G}$ denote the space of all possible regular distributions of $G^s(\epsilon)$. The outcome mapping $M^s$ is a correspondence $M^s:\mathcal{G}\rightrightarrows \mathcal{F}$. \colorbox{lightgray}{See (\ref{eq: appli, primitive space}),(\ref{eq: appli, M^s}), (\ref{eq: Roy appli, M^s}) for illustrations.}
\end{definition}
\begin{definition}\label{def: structure universe}
	A structure universe $\mathcal{S}$ is a collection of structures such that $\cup_{s\in\mathcal{S}} M^s(G^s)=\mathcal{F}$, and an assumption $A$ is a subset of $\mathcal{S}$. \colorbox{lightgray}{See (\ref{eq: appli, potential outcome model space}),(\ref{eq: IA instrument assumption }),(\ref{eq: roy model structure universe}),
		(\ref{eq: Roy Assumption})
		 for illustrations.}
\end{definition}

The mapping $M^s$ of structure $s$ relates the distribution $\epsilon$ to the distributions of $X$. Since the distributions of $X$ are generated by the distributions of $\epsilon$, we call $\epsilon$ the primitive variables. Definition \ref{def: structure universe} allows overlap between $X$ and $\epsilon$. A collection of structures is called a structure universe. We want to learn the distribution of $\epsilon$ and the mapping $M$ from the distribution of observables $F$.

Here I explicitly distinguish the  structure universe $\mathcal{S}$ and the assumption $A$, though both are just a collection of structures. The structure universe $\mathcal{S}$ is the paradigm that can span different empirical contexts.  On the other hand, an assumption $A$ places constraints that are suitable for a particular empirical context, or convenient for empirical analysis.

The condition $\cup_{s\in\mathcal{S}}M^s(G^s)=\mathcal{F}$ requires that  all possible distributions of observables can be generated by some structure in the universe. Moreover, no distribution outside $\mathcal{F}$ can be generated by $\mathcal{S}$.

\begin{definition}\label{def: complete and incomplete struc space}
	A structure $s:=(M^s,G^s)$ is called complete if $M^s(G^s)$ is a singleton. Otherwise it is called incomplete. A universe $\mathcal{S}$ is called complete if every structure $s$ in $\mathcal{S}$ is complete, otherwise it is incomplete.
\end{definition} 

The definition of completeness is slightly different from the definition of completeness in \citet{tamer2003incomplete}.  \citet{tamer2003incomplete} defines a model to be complete if the mapping from $\epsilon$ to $X$ is a singleton and non-empty, and incomplete if the mapping has multiple outputs, and incoherent if the mapping generates no output. Here, my definition of economic structure does not specify the mapping from each $\epsilon$ to $X$. Instead, I consider the mapping from the distribution of $\epsilon$ to the distribution of $X$. If the probability of multiple outcome is non-zero, an incomplete model in \citet{tamer2003incomplete} implies an incomplete structure in my definition.

My definition of structure, however, does not have a corresponding terminology for incoherent model. This is because I require $\cup_{s\in\mathcal{S}}M^s(G^s)=\mathcal{F}$ to hold. \citet{chesher2012simultaneous} propose four ways to deal with model incoherence and derive the distribution of observables under the model. My definition of $M^s$ as the mapping between distributions can be viewed as the consequences of \citet{chesher2012simultaneous}. 

\begin{definition} (Breusch) \label{def: complete theory, Breusch refutability}
	An assumption $A$ is called refutable if there exists an $F\in \mathcal{F}$ such that $F\notin\cup_{s\in A}M^s(G^s)$. An assumption $A$ is called confirmable if there exists an $F\in \mathcal{F}$ such that $F\notin\cup_{s\in A^c}M^s(G^s)$. 
\end{definition}

The notions of refutability and confirmability of a complete structure are given in \citet{breusch1986}. If an assumption is refutable, then there exists some $F$ that can reject $A$. The notions of refutability and confirmability are stated in terms of the observation space $\mathcal{F}$. Equivalently, we can characterize refutability and confirmability in terms of the structure universe $\mathcal{S}$. To do this, I first define the non-refutability and confirmation sets associated with $A$.
\begin{definition}\label{def: non-refutablity set in incomplete structure}
	(non-refutability set) \\
	The strong non-refutability set associated with $A$ under $\mathcal{S}$ is defined as 
	\[\mathcal{H}^{snf}_\mathcal{S}(A)=\left\{s\in\mathcal{S}: M^s(G^s)\subseteq \cup_{s^*\in A} M^{s^*}(G^{s^*}) \right\}.\]
	The weak non-refutability set associated with $A$ under $\mathcal{S}$ is defined as 
	\[\mathcal{H}^{wnf}_\mathcal{S}(A)=\left\{s\in\mathcal{S}: M^s(G^s)\cap \left(\cup_{s^*\in A} M^{s^*}(G^{s^*})\right)\ne \varnothing \right\}.\]
	\colorbox{lightgray}{See Lemma \ref{lem: testable implication of IA assumption}, (\ref{eq: testable implication of LATE}) and Proposition \ref{prop: roy model, non-refutable and confirmation set} for illustrations.}
\end{definition}
To accommodate the incomplete structures, we define two types of non-refutability sets. We call $\mathcal{H}_{\mathcal{S}}^{snf}(A)$ the strong non-refutability set associated with $A$, because if the true structure $s$ is in $\mathcal{H}_{\mathcal{S}}^{snf}(A)$, then for any distributions of observables in $M^s(G^s)$, we cannot refute $A$. In contrast, we call $\mathcal{H}_{\mathcal{S}}^{wnf}(A)$ the strong non-refutability set associated with $A$, because if the true structure $s$ is in $\mathcal{H}_{\mathcal{S}}^{wnf}(A)$, then for some distributions of observables in $M^s(G^s)$, we cannot refute $A$.

\begin{definition}\label{def: confirmation sets}
	(Confirmation set) \\
	The strong confirmation set associated with $A$ under ${S}$ is defined as 
	\[
	\mathcal{H}_{S}^{scon}(A)=\left\{s\in{S}: M^s(G^s)\subseteq \cap_{s^*\in A^c} \left(M^{s^*}(G^{s^*})^c\right)\right\}.
	\]
	The weak confirmation set associated with $A$ under ${S}$ is defined as 
	\[
	\mathcal{H}_{S}^{wcon}(A)=\left\{s\in{S}: M^s(G^s)\cap\left[ \cap_{s^*\in A^c} \left(M^{s^*}(G^{s^*})^c\right)\right]\ne \varnothing\right\}.
	\]
	\colorbox{lightgray}{See Proposition \ref{prop: roy model, non-refutable and confirmation set} for an illustration.}
\end{definition}
The strong confirmation set is the collection of structures that cannot be observationally equivalent to any structures outside $A$ for any observed distribution $F$. Weak confirmation set is the collection of structures that cannot be observationally equivalent to any structures outside $A$ for some observed \mbox{distribution $F$.} I call $\mathcal{H}^{scon}_{\mathcal{S}}(A)$ (resp. $\mathcal{H}^{wcon}_{\mathcal{S}}(A)$) the strong (resp. weak) confirmation set associated with $A$, because if the true structure $s$ is in $\mathcal{H}^{scon}_{\mathcal{S}}(A)$ (resp. $\mathcal{H}^{wcon}_{\mathcal{S}}(A)$), then for all (resp. some) distribution of observables in $M^s(G^s)$, we can confirm that the true structure must lies in $A$. In particular, 
\[\mathcal{H}_\mathcal{S}^{scon}(A)\subseteq \mathcal{H}_\mathcal{S}^{wcon}(A)\subseteq A\subseteq \mathcal{H}_\mathcal{S}^{snf}(A)\subseteq \mathcal{H}_\mathcal{S}^{wnf}(A).\]
When the structure universe $\mathcal{S}$ is complete, $M^s(G^s)$ is always a singleton, and  $\mathcal{H}_\mathcal{S}^{scon}(A)= \mathcal{H}_\mathcal{S}^{wcon}(A)$, $\mathcal{H}_\mathcal{S}^{snf}(A)= \mathcal{H}_\mathcal{S}^{wnf}(A)$. The following proposition helps to interpret the confirmation sets associated with $A$ as the non-refutability sets associated with $A^c$. It also shows that the strong non-refutability set and weak confirmation set as operation are idempotent. 
\begin{prop} \label{prop: operation of confirmable and refutable set }
	The following holds: 1. $\left[\mathcal{H}_\mathcal{S}^{snf}(A)\right]^c=\mathcal{H}_\mathcal{S}^{wcon}(A^c)$;  2. $\left[\mathcal{H}_\mathcal{S}^{wnf}(A)\right]^c=\mathcal{H}_\mathcal{S}^{scon}(A^c)$; 3. $\mathcal{H}_\mathcal{S}^{wcon}(\mathcal{H}_\mathcal{S}^{wcon}(A))=\mathcal{H}_\mathcal{S}^{wcon}(A)$; 4. $\mathcal{H}_\mathcal{S}^{snf}(\mathcal{H}_\mathcal{S}^{snf}(A))=\mathcal{H}_\mathcal{S}^{snf}(A)$.
\end{prop}
The definition of refutability and confirmability in \citet{breusch1986} is defined on the outcome space $\mathcal{F}$, but we can also characterize it on the structure universe $\mathcal{S}$.   
\begin{prop} \label{prop: refutable equivalent characterization}
		An assumption $A$ is refutable if and only if $\mathcal{H}_\mathcal{S}^{snf}(A)\ne \mathcal{S}$. An assumption $A$ is confirmable if and only if $\mathcal{H}_\mathcal{S}^{wcon}(A)\ne \varnothing$.
\end{prop}
By definition, we should have $\cup_{s\in \mathcal{H}_{\mathcal{S}}^{snf}(A)} M^{s}(G^{s})=\cup_{s\in A} M^{s}(G^{s})$, so $\mathcal{H}_{\mathcal{S}}^{snf}(A)$ is refutable if and only if $A$ is refutable. In many cases, it is easy to check whether $\mathcal{H}_{\mathcal{S}}^{snf}({A})=\mathcal{S}$ in Proposition \ref{prop: refutable equivalent characterization} than to check Definition \ref{def: complete theory, Breusch refutability}.

\subsection{Identification Problem}
In many empirical studies, we want to find the value of a parameter of interest rather than a class of structures that are consistent with data. This parameter can be a moment of unobserved primitive variables, or a counterfactual outcome of the structure. The parameter of interest can also give interpretation on the causal relation between outcome variables and primitive variables. Imposing strong assumptions helps to restrict the set of data-consistent parameter values, but an imposed assumption $A$ as in Definition \ref{def: structure universe} may be rejected by some distribution of observables. Therefore, in many empirical studies, researchers often first present some summary statistics that justify the assumption. If the assumption is rejected by the data, researchers can move to another assumption. This is an ex-post way of choosing a relaxed assumption. There are two major problems with this approach. First, such justifications are heuristic pre-testing procedures of assumption $A$, and any subsequent inference on the parameter of interest may have incorrect size control due to pre-testing. Second, researchers do not specify what they will do if $A$ is rejected. Most likely they will choose another assumption that will not be rejected by the data. To avoid the pre-testing issue, I propose to solve the problem from an ex-ante perspective, i.e. impose a non-refutable assumption before any distribution of observables is realized. 


I first formalize the definition of an identification system and discuss how to deal with an existing situation, where assumption $A$ may be rejected by the data.
\begin{definition} \label{def: Identification System}
	A parameter of interest $\theta$ is a function $\theta: \mathcal{S}\rightarrow \Theta$, where $\Theta$ is the parameter space. The identified set for $\theta$ is a correspondence $\Theta^{ID}_A:\mathcal{F} \rightrightarrows \Theta$ such that  
	\begin{equation}\label{eq: definition of identified set}
	\Theta_A^{ID}(F)=\{\theta(s): s\in A \quad and \quad F\in M^s(G^s) \}.
	\end{equation}
	We call $(\mathcal{S},A,\theta,\Theta^{ID}_A)$ an identification system. \colorbox{lightgray}{See (\ref{eq: appli, formular}) for $\theta$ and (\ref{eq: appli, ID set}) for an illustration of $\Theta^{ID}_A$.}
\end{definition}
A parameter of interest can take a very general form. It can be the structure $s$ itself, or it can be a counterfactual outcome. For example, suppose $M^s$ is known up to a finite dimensional  vector: $M^s(\cdot)=M(\cdot\,;(\beta_1,...,\beta_k))$. Further suppose the objective of our counterfactual analysis is to find the predicted distribution of observables when $\beta_1=0$. Then the parameter of interest $\theta$ can be defined as $\theta(s)= M(G^s;(0,...,\beta_k))$. For a parameter of interest $\theta$,   $\Theta^{ID}_A(F)$ is the set of parameters that are compatible with the data. For an empirical researcher, the main concern of the partial identification method is the possibility of an empty identified set. Here I characterize the equivalent condition of an empty identified set. 
\begin{prop}\label{prop: equivalence well-defined identification and non refutable}
	An assumption $A$ is non-refutable if and only if for any parameter of interest $\theta$, the associated identified set $\Theta_A^{ID}(F)\ne\varnothing$ holds  $\forall F\in \mathcal{F}$.
\end{prop}
Here is an intuition of Proposition \ref{prop: equivalence well-defined identification and non refutable}: If we have a non-empty identified set for an $F$, then there must exist a structure in $A$ that rationalizes $F$. Since this is true for all $F$, $A$ is non-refutable. Conversely, if $A$ is non-refutable, then for any $F$ we can find a structure $s\in A$ to rationalize $F$, and the corresponding $\theta(s)$ must lie in the identified set.

\subsection{Relaxed Assumption Approach}
For a refutable assumption $A$, there exist some $F$ such that $\Theta_A^{ID}(F)=\varnothing$.  This can be unsatisfying because empirical researchers cannot directly interpret the distribution of $\epsilon$. To avoid this, before seeing any outcome distribution, a practitioner can impose a relaxed assumption $\tilde{A}$ such that $\mathcal{H}_\mathcal{S}^{snf}(\tilde{A})= \mathcal{S}$ and $A\subseteq \tilde{A}$. 

\begin{definition}\label{def: consistent extension, theta and strong}
	Given structural universe $\mathcal{S}$, a refutable assumption $A$ and $\theta$, we call $\tilde{A}$
	\begin{enumerate}
		\item a well-defined extension, if $A\subseteq \tilde{A}$ and $\mathcal{H}_\mathcal{S}^{snf}(\tilde{A})=\mathcal{S}$;
		\item a $\theta$-consistent extension, if $\tilde{A}$ is a well-defined extension, and $\Theta^{ID}_{\tilde{A}}(F)=\Theta_A^{ID}(F)$ whenever $\Theta_A^{ID}(F)\ne \varnothing$; \colorbox{lightgray}{See Proposition \ref{prop: minimal marg ind as consistent extension} for an illustration.}
		\item a strong extension, if for any parameter of interest $\theta^*$ defined in Definition \ref{def: Identification System}, $\tilde{A}$ is a $\theta^*$-consistent extension.\colorbox{lightgray}{See Proposition \ref{prop: minimal deviation as strong extension} and \ref{prop: roy model, sharp characterization of identified set} for an illustration.}
	\end{enumerate}
\end{definition}
These three definitions are nested. A well-defined extension ensures that the identified set will never be empty; A $\theta$-consistent extension preserves the identified set for a parameter of interest $\theta$. A strong extension moreover ensures that the identified set for any parameter of interest will be preserved. In different empirical settings, researchers' parameters of interest can differ. If a strong extension is found, researchers can use this extension across different empirical contexts. The following proposition gives a characterization of whether $\tilde{A}$ is a strong extension. 
\begin{prop}\label{prop: conditions of strong extension}
	Given $\mathcal{S}$, suppose $A$ is refutable and $\tilde{A}$ is a well-defined extension, then $\tilde{A}$ is a strong consistent extension if and only if $\mathcal{H}_{\mathcal{S}}^{wnf}(A)\cap \tilde{A} =A$.
\end{prop}
In a complete structure universe, we can always find a strong extension $\tilde{A}$. This is a major difference between complete and incomplete structure universe. 
\begin{prop}\label{prop: exist of strong extension}
	If $\mathcal{S}$ is a complete structure universe, then $\tilde{A}=A\cup  [\mathcal{H}_\mathcal{S}^{snf}(A)]^c$ is a strong extension of $A$. Moreover, any strong extension $\tilde{A}'$ is a subset of $\tilde{A}$. We call this $\tilde{A}$ the maximal strong extension of $A$.
\end{prop}

In other words, $\tilde{A}=A\cup  [\mathcal{H}_\mathcal{S}^{snf}(A)]^c$ is the strong extension that puts the least structural assumption outside $\mathcal{H}_\mathcal{S}^{snf}(A)$. Should we always use $\tilde{A}=A\cup [ \mathcal{H}_\mathcal{S}^{snf}(A)]^c$ as the choice of strong extension when the structural universe $\mathcal{S}$ is complete? Unfortunately, using the maximal strong extension will lead to very badly behaved identified set $\Theta_{\tilde{A}}^{ID}(F)$.  Suppose $\mathcal{F}$ is equipped with some metric $d$, and $F_0$ is on some part of the boundary of the set of predicted observable distribution $\cup_{s\in {A}} M^s(G^s)$. It is possible that $\Theta_{\tilde{A}}^{ID}(F_0)$ gives an informative bound (i.e. $\Theta_{\tilde{A}}^{ID}(F_0)\ne \Theta$) on the parameter of interest , but for an $F'\in\left[\cup_{s\in {A}} M^s(G^s) \right]^c$ that is arbitrarily close to $F_0$, the identified set $\Theta_{\tilde{A}}^{ID}(F')$ is uninformative (i.e. $\Theta_{\tilde{A}}^{ID}(F')=\Theta$). See the identification result in Proposition \ref{prop: appli, maximal extension arbitrary instrument} for an illustration. This raises two concerns. First, at the identification level, the interpretation of an uninformative identified set $\Theta_{\tilde{A}}^{ID}(F')$ under the maximal strong extension $\tilde{A}$ is not very different from an empty identified set $\Theta_{A}^{ID}(F')=\varnothing$ under the original assumption $A$. An uninformative identified set says that any parameter value of $\theta$ is compatible with data, while an empty identified set  says that no parameter value of $\theta$ is compatible with data. In either case, the identification result does not help us to interpret the environment. Second, we may get spurious informative inference result due to sampling error. When $F'$ is close to the boundary of $\cup_{s\in {A}} M^s(G^s)$ but not in it, sampling error may lead us to a spurious but informative bound $\Theta_{\tilde{A}}^{ID}(F')$, even if the true identified set should have been uninformative. In other words, estimated identified set is not consistent. The maximal strong extension $\tilde{A}=A\cup  [\mathcal{H}_\mathcal{S}^{snf}(A)]^c$ solves the refutability issue, but it imposes too few constraints outside $A$ to generate informative result on $\theta$. 

For an incomplete structure universe, there is a gap between the strong and weak non-refutability sets. As a result, we may not be able to find a strong extension of $A$. 

\begin{prop} \label{prop: incomplete model impossible to find strong ext}
	If {\footnotesize $\left(\cup_{s\in \mathcal{H}_\mathcal{S}^{wnf}(A)}M^s(G^s)\right)\cap \left(\cup_{s\in [\mathcal{H}_\mathcal{S}^{wnf}(A)]^c}M^s(G^s)\right)=\varnothing $} and $\mathcal{H}_{\mathcal{S}}^{wnf}(A)\backslash \mathcal{H}_{\mathcal{S}}^{snf}(A)\ne\varnothing$  hold, then there does not exist a strong extension of $A$. 
\end{prop}

This situation happens when there is a nesting relation between $A$ and $A^c$:  suppose for each structure $s\in A$, we can find an $s^*\in A^c$ such that $M^s(G^s)\subseteq M^{s^*}(G^{s^*})$; and for every $s^*\in A^c$ we can find an $s\in A$ such that $M^s(G^s)\subseteq M^{s^*}(G^{s^*})$. If $A$ is refutable, then $\mathcal{H}_{\mathcal{S}}^{snf}(A)\ne \mathcal{S}$. However the nesting relation implies $\mathcal{H}_{\mathcal{S}}^{wnf}(A)=\mathcal{S}$. Then both conditions in Proposition \ref{prop: incomplete model impossible to find strong ext} hold, and there exists no strong extension of $A$. 
\begin{example}
	Figure \ref{fig:no-strong-extension-exists} illustrates a simple case where $\mathcal{S}=\{s_1,s_2,s_3\}$ and the assumption set $A=\{s_1\}$. There is a nesting relation between the predicted outcome distributions of $s_2$ and $s_1$, and the predicted outcome distributions of $\{s_1,s_2\}$ and $s_3$ are disjoint. This satisfies the condition in Proposition \ref{prop: incomplete model impossible to find strong ext}. The only well-defined extension of $A$ must be $\tilde{A}=\{s_1,s_2,s_3\}$, since we need to include $s_2$ to predict $\tilde{F}$. It is easy to check $\tilde{A}\cap \mathcal{H}_{\mathcal{S}}^{wnf}(A)=\{s_1,s_2\}\ne A$, so $\tilde{A}$ cannot be a strong extension.
	\begin{figure}[H]
		\centering
		\includegraphics[width=0.8\linewidth]{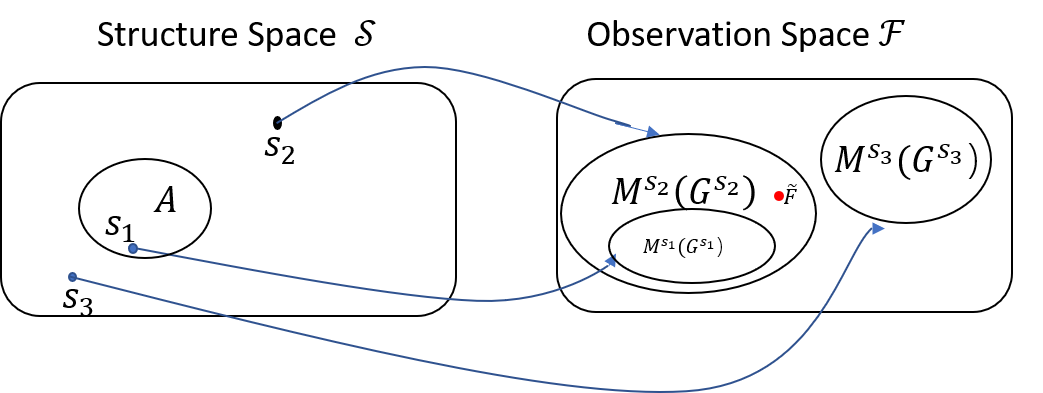}
		\caption{Impossibility to find a strong extension.}
		\label{fig:no-strong-extension-exists}
	\end{figure}
\end{example}

\subsection{Minimal Deviation Method}\label{section: minimal deviation method}
In many cases, we can find a function $m:\mathcal{S}\rightarrow \mathbb{R}_+$ such that $m(s)=0$ for all $s\in A$ \footnote{See (\ref{eq: defiers measure}), (\ref{eq: Roy model, efficiency loss}) for example.}.  While an assumption $A$ may be refutable, we impose $A$ in the first place because it reflects economic theory suitable in the empirical context. Therefore we would consider a departure from $A$ is abnormal and is against the economic intuition behind $A$.  I extend our assumption to allow for a minimal departure from the baseline assumption $A$. This way to relax assumption $A$ is called the minimal deviation method. Formally, suppose the refutable assumption $A$ can be written as an intersection of several larger assumptions: $A=\cap_{j=1}^J A_j$. This representation allows us to consider a departure from a \mbox{particular $A_j$}.  

\begin{definition}\label{def: well-defined relaxation measure}
	Fix an index $j\in \{1,2,...J\}$. A relaxation measure of departure from $A_j$ with respect to $\{A_l\}_{l\ne j}$ is a function $m_j: \mathcal{S}\rightarrow \mathbb{R}_+\cup\{+\infty\}$ such that $m_j(s)=0$ for all $s\in A_j$. We say $m_j$ is well-behaved if for any $F\in \mathcal{F}$, there exists a structure $s^*\in \cap_{l\ne j}A_l$ such that 
	\[
	m_j(s^*)=\inf\{m_j(s): F\in M^s(G^s) \quad and \quad s\in \cap_{l\ne j}A_l\},
	\]
	and $m_j(s^*)<\infty$.
	\colorbox{lightgray}{See (\ref{eq: defiers measure}) and (\ref{eq: Roy model, efficiency loss}) for illustrations.}
\end{definition}
We want the relaxation measure to be well-behaved such that when we push the deviation to infinity, we can generate any distributions in $\mathcal{F}$. The well-behaved condition ensures that there exists a structure in $\cap_{l\ne j}A_l$ that can achieve the minimal measure. This is essential for the construction of a extension using $m_j$. See Proposition \ref{prop: m^{MD} is not well behaved if use full independence} for examples of ill-behaved measures. Now I construct the minimal deviation extension $\tilde{A}$.

\begin{definition}\label{def: minimal deviation extension}
	Fix an index $j\in \{1,2,...J\}$ and a well-behaved relaxation measure $m_j$. For any $F$, let $m^{min}(F)\equiv \inf\{m_j(s): F\in M^s(G^s) \,\, and \,\, s\in \cap_{l\ne j}A_l\}$ be the minimal deviation for the observable distribution $F$. We call $\tilde{A}=\cup_{F\in \mathcal{F}} \left\{s\in \cap_{l\ne j}A_l: m_j(s)=m^{min}(F),\,\,F\in M^s(G^s)\right\}$ the minimal deviation extension of $A$ under $m_j$. \colorbox{lightgray}{See the constructions of Assumptions \ref{assumption:type ind instru} and \ref{assump: minimal efficiency loss Roy} for illustrations.}
\end{definition}
The construction in Definition \ref{def: minimal deviation extension} only relaxes the assumption $A_j$ and keeps other assumptions unchanged. The following proposition shows a way to check whether a minimal deviation extension is $\theta$-consistent or strong consistent. 

\begin{prop}\label{prop: minimal deviation as strong extension}
	Let $A=\cap_{l=1}^J A_l$ and fix an index $j\in \{1,2,...,J\}$. Suppose $m_j$ is a well-behaved relaxation measure with respect to $\{A_l\}_{l\ne j}$. Then a minimal deviation extension $\tilde{A}$ under $m_j$ is 
	\begin{enumerate}
		\item $\theta$-consistent if $\Theta_A^{ID}(F)$ defined in (\ref{eq: definition of identified set}) satisfies 
		\[\Theta_A^{ID}(F)=\{\theta(s):\quad F\in M^s(G^s)\quad and \quad  s\in \left(\cap_{l\ne j}A_l\right)  \quad and \quad m_j(s)=0\};\]
		\item a strong extension if $A_j=\{s\in \cap_{l\ne j}A_l: \,\,m_j(s)=0\}$. 
	\end{enumerate}
\end{prop}

\subsubsection*{Multiplicity of Extensions and the Extension Choice}
To facilitate the discussion of criteria of choosing among multiple extensions, I assume $\mathcal{F}$ is endowed with a metric $d_F$. Moreover, I fix the parameter of interest $\theta$, and assume $\Theta$ is endowed with a \mbox{metric $d_\theta$}. 

In some cases, there can be multiple ways to write an assumption, i.e. $A=A_j\cap(\cap_{l\ne j} A_l)=A_j\cap(\cap_{l\ne j} A_l^\prime)$. Fix a $j$, even if we use the same measure $m_j$, since the minimal deviation is defined with respect to $\{A_l\}_{l\ne j}$, the extension can differ when using a different representation. It should also be noted that to check whether  $m_j$ is a well-defined relaxation measure, we need to look at $\{A_l\}_{l\ne j}$. This means $m_j$ can be a well-defined relaxation measure with respect to $\{A_l\}_{l\ne j}$ but not $\{A_l^\prime\}_{l\ne j}$. Here I leave the choice of representation of the assumption to researchers and discuss the issue of multiple extensions that arises from two aspect: which assumption to relax and the choice of relaxation measure. 

Each minimal deviation extension $\tilde{A}$ corresponds to an index $j$ and a relaxation measure $m_j$. Given a representation of $A=\cap_{j=1}^J A_j$, we can choose which sub-assumption $A_j$ to relax, and we can also choose the relaxation measure $m_j$. Different choices of which assumption $j$ to relax, and different relaxation measures $m_j$ will result in different relaxed assumptions. Moreover, relaxed assumptions constructed from different relaxation measures can be non-nested with each other. Here, I discuss two criteria to choose an assumption among non-nested relaxed assumptions.

First, the relaxed assumption $\tilde{A}$ should be suitable for the empirical context. Given the empirical context, if we can find an economic story such that the $j$-th assumption in $A=\cap_{j=1}^J A_j$ fails, we will focus on finding a well-behaved relaxation measure $m_j$ corresponding to $A_j$. Second, we want some continuity property of the identified set with respect to the outcome distribution $F$. 
\begin{property}[Identified Set Continuity]\label{property: continuity}
	The identified set $\Theta_{\tilde{A}}^{ID}(F)$ under $\tilde{A}$ is a continuous correspondence\footnote{Recall that a correspondence $\Gamma :A\rightrightarrows B$ is called upper hemicontinuous at the point $a$ if for any open neighborhood $V$ of $\Gamma(a)$ there exists a neighborhood $U$ of $a$ such that for all $a'\in U$, $\Gamma(a')$ is a subset of $V$. A correspondence $\Gamma :A\rightrightarrows B$ is called lower hemicontinuous  at the point $a$ if for any open set $V$ intersecting $\Gamma(a)$ there exists a neighborhood $U$ of a such that $\Gamma(a')$ intersects $V$ for all $a'\in U$. A continuous correspondence is both upper and lower hemicontinuous.} from $\mathcal{F}$ to $\Theta$.  
\end{property}
\noindent Recall that we fix the parameter of interest $\theta$ in the beginning of this section. An $\tilde{A}$ that satisfies Property \ref{property: continuity} for $\theta$ may fail Property \ref{property: continuity} for a different parameter of interest. Without the continuity property, a consistent estimator of the identified set may not exist, and the identified set can be spuriously informative due to sampling error. Examples of discontinuous and continuous identified set correspondences can be found in Proposition \ref{prop: continuity of estimator in different cases}. Sufficient conditions to check Property \ref{property: continuity} and the further reasoning of Property \ref{property: continuity}  can be found in Appendix \ref{section: Discussion of Continuous ID Set Correspondence}.

\subsection{Structure Completion}\label{section: Structural Completion}
We have seen that for an incomplete structure universe, a refutable assumption $A$ may not have a strong extension. This is because in incomplete structures, we are agnostic about how distributions of outcomes are selected. If a structure $s$ has two predicted outcome distribution $M^s(G^s)=\{F_1,F_2\}$, we consider a completion procedure that separates $s$ into two complete structures $s_1^*$ and $s_2^*$ such that $M^{s_1^*}(G^{s_1^*})=\{F_1\}$ and $M^{s_2^*}(G^{s_2^*})=\{F_2\}$. The completion procedure then allows us to distinguish $s_1^*$ from $s_2^*$ by observing either $F_1$ \mbox{or $F_2$}.

\begin{definition}\label{def: completion of an incomplete structure universe}
	Given a structure $s=(M^s,G^s)$, let $\mathcal{C}(s)$ be the collection of all single-valued correspondences from $\mathcal{G}$ to $\mathcal{F}$ such that $M^*\in \mathcal{C}(s)$ implies $M^*(G^s)\subseteq M^s(G^s)$. We call 
	\begin{equation}\label{eq: completion of S}
	{\mathcal{S}}^*=\{({M}^{*},G): s\in \mathcal{S},\quad G=G^s,\quad   M^{*}\in \mathcal{C}(s)\}
	\end{equation}
	the completion of $\mathcal{S}$. \colorbox{lightgray}{See (\ref{eq: roy model, completed structure mapping}) for an illustration of $\mathcal{C}(s)$.}
	

\end{definition}

Definition \ref{def: completion of an incomplete structure universe} considers all possible completions $C$.  The cardinality of $\mathcal{C}(s)$ is the same as that of $M^s(G^s)$. The completion procedure is without loss of generality, since all possible selections are considered. The key property is that for any parameter of interest, the identified set is not changed if the parameter of interest in the completed structure is properly defined in the following way. 

\begin{prop}\label{prop: incomplete theory, preserve the id set after completion}
	Let $\mathcal{S}$ be an incomplete structural universe, let $A$ be any assumption and $\theta$ be any parameter parameter of interest. Let $\mathcal{S}^*,A^*,\theta^*,{\Theta}_{A^*}^{*ID}$ be defined in the following way:
	\begin{equation}
	\begin{split}
	&{\mathcal{S}}^* \,\, \text{is the completion of } \,\,\mathcal{S},\\
	&{A}^*=\{s^*:s^*=({M}^{s^*},G^{s^*}),\quad s\in \mathcal{S}\cap A,\quad G^{s^*}=G^s, \quad {M}^{s^*}\in \mathcal{C}(s)\},\\
	&{\theta}^*(s^*)=\theta(s)\quad for\,\,all\,\, s=(M^s,G^s),\, s^*=({M}^{s^*},G^{s^*})   \text{ such that } M^{s*}\in \mathcal{C}(s)\,\,and\,\, G^{s^*}=G^{s}\\
	&\Theta^{*ID}_{A^*}(F)= \{\theta^*(s^*): F\in M^{s^*}(G^s),\,\, s^*\in A^*\}.
	\end{split}
	\end{equation}
	Then $\Theta_A^{ID}(F)={\Theta}^{*ID}(F)$ for all $F$. 
\end{prop}

In many cases, finding a strong extension is not feasible for an incomplete structure universe, but feasible for its completion. See Proposition \ref{prop: roy model, sharp characterization of identified set} for an illustration.

\subsection{The Hypothesis Testing Problem}\label{section: binary decision and hypothesis testing of structures}
In empirical research, a commonly asked question is whether we can tell if the true value of the parameter of interest lies in a set, which can be written as a hypothesis $H$ on the parameter value. In this section, I consider the following formulation of a hypothesis $H$ on a structural \mbox{parameter $\theta$} under a non-refutable assumption $\tilde{A}$: $H=\{s\in \tilde{A}: \theta(s)\in \Theta^0\}$, where $ \Theta^0$ is a parameter value set. The implicit alternative is $H^c\cap \tilde{A}$. Here I only consider non-refutable assumption $\tilde{A}$. If an assumption $A$ is refutable and cannot generate all distributions of observables,  then for some distributions of observables, we cannot make say at least one of $H$ and $H^c\cap A$ holds true.

Policy makers sometimes use the result of hypothesis testing of a parameter value to guide their policy decisions. This decision procedure is called the `inference-based' approach in \citet{manski2019econometrics} and is a conventional practice in medical treatment policy decision \citep{manski2020covid}. However, the `inference-based' policy decision approach can be problematic if the hypothesis on parameter value $H$ does not induce a partition on the observation space $\mathcal{F}$: if \mbox{both $H$} and $H^c\cap \tilde{A}$ can generate some observed distribution $F_0$, then we cannot tell whether $H$ holds by observing $F_0$. To formally discuss this issue, I first discuss the `hypothesis testing' problem assuming that I know the distribution of observables. I call this the binary decision problem \footnote{The same problem is called the binary choice problem in \citet{manski2019econometrics}. To avoid the confusion with concepts in the discrete choice literature, I slightly change the name.}.



\begin{definition}\label{def: weakly binary decidable}
	We say a hypothesis $H$ can be decided by $F$ under $\tilde{A}$ if   either of the following conditions holds:
	\begin{enumerate}
		\item  $F\notin\cup_{s^*\in [H^c\cap \tilde{A}] }  M^{s^*}(G^{s^*})$;
		\item  $F\notin \cup_{s^*\in H }  M^{s^*}(G^{s^*}).$
	\end{enumerate}
	$H$ is called weakly binary decidable under $\tilde{A}$ if there exists an $F\in \mathcal{F}$ such that $H$ can be decided \mbox{by $F$.} $H$ is called strongly binary decidable under $\tilde{A}$, if for all $F\in\mathcal{F}$, $H$ can be decided by $F$.
\end{definition}
If condition 1 in Definition \ref{def: weakly binary decidable} holds, it implies that the true structure $s$ that generates $F$ must be in $H$, since $F$ cannot be predicted by $H^c\cap \tilde{A}$, and this confirms $s\in H$; if condition 2 in Definition \ref{def: weakly binary decidable} holds, it implies that the true structure $s$ cannot be in $H$, since $F$ cannot be predicted by $H$, and this refutes $s\in H$. If both conditions fail, it means $F$ can be predicted by structures both  inside and outside $H$, which creates an ambiguity in the binary decision problem. If $H$ can be decided by any $F$, we say it is strongly binary decidable.

\begin{example}
	Consider a simple linear regression model 
	\begin{equation}\label{eq: Binary Decision Example, linear regression}
	Y_i=\beta_0+\beta_1 Z_i+\eta_i,
	\end{equation}
	where primitive variables are $(Z_i,\eta_i)$, observed variables are $(Y_i,Z_i)$. The correspondence $M^s$ is determined by (\ref{eq: Binary Decision Example, linear regression}) \footnote{ The image of mapping $M^s$ is the push-forward measure of $(Y_i,Z_i)$ under the linear function.}. $M^s$ is determined by two parameters $\beta^s_0,\beta^s_1$.  Assumption $\tilde{A}$ is the classical zero conditional mean restriction: $
	\tilde{A}=\{s: E_{G^s}[\eta_i|Z_i]=0,\,\, (\beta^s_0,\beta^s_1)\in \mathbb{R}^2\}.$
	
	We can show that the hypothesis $H=\{s\in \tilde{A}: \beta^s_1\ge 0\}$ is strongly binary decidable. Indeed, we have $\cup_{s^*\in H }  M^{s^*}(G^{s^*})=\{F: Cov_F(Y_i,Z_i)\ge 0\}$ and $\cup_{s^*\in {H^c\cap\tilde{A}} }  M^{s^*}(G^{s^*})=\{F: Cov_F(Y_i,Z_i)< 0\}$. These two sets do not intersect, so conditions in Definition \ref{def: weakly binary decidable} can be verified for all $F$. We will see another strongly binary decidable hypothesis in an interval data example later.
\end{example}
The following lemma provides an equivalent condition to check whether $H$ is strongly binary decidable. The lemma below uses the definition of non-refutability set (Definition \ref{def: non-refutablity set in incomplete structure}) and confirmation set (Definition \ref{def: confirmation sets}) under $\tilde{A}$ instead of $\mathcal{S}$.

\begin{lem}\label{lem: charaterization of strong binary detectable}
	A hypothesis $H$ is strongly binary decidable under $\tilde{A}$ if and only if  $\mathcal{H}_{\tilde{A}}^{scon}(H)=\mathcal{H}_{\tilde{A}}^{wnf}(H)$.
\end{lem}

Intuitively, Lemma \ref{lem: charaterization of strong binary detectable} says that if we can confirm that the true structure $s$ is in $H$ for all distributions of observables, then we can refute $H^c\cap \tilde{A}$ for all distributions of observables. 

\subsubsection*{Finite Sample Testing}
Now I consider statistical testing of $H$ based on a finite sample. We want to test the null hypothesis that the true structure $s_0$, which generates the outcome distribution $F$, satisfies \mbox{hypothesis $H$} against its complement in $\tilde{A}$:
\[
\mathcal{H}_0: s_0\in H\quad \quad v.s. \quad \quad  \mathcal{H}_1: s_0\in H^c\cap \tilde{A}.
\]
We have a finite sample of $i.i.d$ realizations from  $F$ with empirical distribution $\mathbb{F}_n$ that converges weakly to $F$. A statistical test $T_n$ is a binary function that maps the empirical distribution and some random vector $\mathbf{\eta}$ to $\{0,1\}$:
\[
T_n(\mathbb{F}_n,\mathbf{\eta})=\begin{cases}
&1\quad \quad \text{means we faile to reject } \mathcal{H}_0,\\
&0\quad \quad \text{means we reject } \mathcal{H}_0.
\end{cases}
\]
\begin{definition}
	We say a test statistic $T_n$ achieves pointwise  structural size control at  level  $\alpha$  if:
	\begin{equation}\label{eq: pointwise size control}
	\inf_{F\in\cup_{s\in H} M^s(G^s)}{\lim\inf}_{n\rightarrow \infty} Pr(T_n(\mathbb{F}_n,\mathbf{\eta})=1)\ge 1-\alpha,
	\end{equation}
	and achieves structural test consistency if:
	\begin{equation} \label{eq: test consistency}
	\inf_{F\in\cup_{s\in [H^c\cap \tilde{A}]} M^s(G^s)} {\lim \sup}_{n\rightarrow \infty } Pr(T_n(\mathbb{F}_n,\mathbf{\eta})=0) =1.
	\end{equation}
\end{definition}
The names `structural size' and `structural test consistency' come from the fact that we construct the criteria  (\ref{eq: pointwise size control}), (\ref{eq: test consistency}) through a partition of the assumption  $\tilde{A}=H\cup[H^c\cap\tilde{A}]$ rather than a partition of the observation space $\mathcal{F}$. Structural size and structural power are what we care about since we aim to make a statement on the true structural parameter value. In particular, we may want to make binary decision on counterfactual outcomes.  As we discuss after Definition \ref{def: Identification System}, a counterfactual analysis can be written as a parameter of interest.

The following proposition shows strongly that binary decidability is closely related to structural size control and structural test consistency.
\begin{prop} \label{prop: binary decidable and size-power issue}
	If $H$ is not strongly binary decidable under $\tilde{A}$, then no statistic can simultaneously achieve pointwise structural size control (\ref{eq: pointwise size control}) for $\alpha<1$ and structural test consistency (\ref{eq: test consistency}).
\end{prop}
The converse of this proposition also holds under further regularity conditions: if $H$ is strongly binary decidable under $\tilde{A}$, we can always find a test statistic that achieves pointwise size control and test consistency. Let 
\[
\begin{split}
\mathcal{F}^d=\cup_{n=1}^\infty\big\{&\mathbb{F}_n: \quad \mathbb{F}_n \text{ supported on a finite subset of } supp(X),\\
&\quad Pr_{\mathbb{F}_n}(X_i=x)=\frac{m}{n},\quad m\in \mathbb{N}  \quad m\le n,\quad and \quad x\in supp(X)\big\}
\end{split}
\]
be the collection of all empirical distributions supported on a finite subset of $supp(X)$, and $Pr_{\mathbb{F}_n}(X_i=x)$ can be written as a fraction.
\begin{assumption}\label{assumption: consistency}
	Let $F$ be the true distribution of outcomes and $\Theta_A^{ID}(F)$ is the identified set. Let $d_{\tilde{\mathcal{F}}}$ be  a metric on $\tilde{\mathcal{F}}=\mathcal{F}\cup\mathcal{F}^d$. The following two conditions hold:
	\begin{enumerate}
		\item $\Theta_{\tilde{A}}^{ID}(F)$ is upper hemicontinuous at $F$;
		\item There exists a sequence of $a_n$ such that  $\mathcal{C}_n= \sqrt{a_n}d_{\tilde{\mathcal{F}}}(\mathbb{F}_n,F)=O_p(1)$, and a sequence of constant $c_n$ such that $c_n\ge \mathcal{C}_n$ holds with probability converging to 1, and $c_n/\sqrt{a_n}\rightarrow 0$. 
	\end{enumerate}
\end{assumption}
\begin{prop}\label{prop: existence of size-consistecy test stat for strongly bin decidable}
	Let Assumption \ref{assumption: consistency} hold. If $\Theta^0$ is a closed set, then there exists a test statistic $T_1(\mathbb{F}_n,\eta)$ that achieves pointwise structural size control (\ref{eq: pointwise size control}) for any  $\alpha\ge 0$ and structural test consistency (\ref{eq: test consistency}) simultaneously. 
\end{prop}

\subsubsection*{Point Identified and Partially Identified Models}

The following proposition shows that hypotheses about a point identified parameter of interest are always strongly binary decidable.
\begin{prop}
	\label{prop: point identification implies strong binary decidable}
	Let $\tilde{A}$ be non-refutable. If $\theta$ is point identified under $\tilde{A}$, i.e. $\Theta_{\tilde{A}}^{ID}(F)$ is a singleton for all $F\in \mathcal{F}$, then $H= \{s\in \tilde{A}: \theta(s)\in \Theta_0\}$ is strongly binary decidable for any parameter value set $\Theta_0\subseteq \Theta$.
	
	Conversely, suppose $\theta$ is partially identified under $\tilde{A}$, and there exist $F$ and $F'$ such that $\Theta_{\tilde{A}}^{ID}(F)\cap \Theta_{\tilde{A}}^{ID}(F')\ne \emptyset$, $\Theta_{\tilde{A}}^{ID}(F)\ne \Theta_{\tilde{A}}^{ID}(F')$. Then there exists a parameter value set $\Theta_0$ such that $H= \{s\in \tilde{A}: \theta(s)\in \Theta_0\}$ is not strongly binary decidable.
\end{prop}

Proposition \ref{prop: point identification implies strong binary decidable} shows that the traditional hypothesis testing approach works in a point identified model, regardless of the hypotheses on the parameter of interest. However, for a partially identified model, the formulation of a hypothesis is crucial. Let's consider the following policy decision rule: `we implement a policy $P$ if and only if the true structure is in $H$. When $H$ is not strongly binary decidable, we have size and power issue for any test statistic $T_n(\mathbb{F},\eta)$. If we decide to implement $P$ if and only if $T_n(\mathbb{F},\eta)=1$, we also know that the testing procedure cannot reject structures in $H^c$ consistently. If the policy $P$ is harmful when the true structure $s$ does not satisfy the parameter constraint of $H$, and we implement $P$ when $T_n(\mathbb{F},\eta)=1$, the policy $P$ can be harmful to the economy.

The problem does not arise from the sampling error but arises from the intrinsic inability to distinguish $H$ and $H^c$ by the distribution of observables. If we want to use a decision rule based on a hypothesis $\tilde{H}$ such that `we implement the policy $P$ if and only if the true structure is in $\tilde{H}$, the hypothesis $\tilde{H}$ must be strongly binary decidable. \footnote{An alternative approach is to formulate the hypothesis testing problem as a statistical decision problem, see Section 2.3 in \citet{manski2019econometrics} for discussion.}

\subsubsection*{Extended and Subset Hypotheses}
The next question is whether we can find a strongly binary decidable extended set or subset.
\begin{definition}
	A strongly binary decidable extension $H^{ext}$ is a strongly binary decidable set such that $H\subseteq H^{ext}\subseteq \tilde{A}$. 	A strongly binary decidable subset $H^{sub}$ is a strongly binary decidable set such that $H^{sub}\subseteq H$. 
\end{definition}


If the benefit to correctly implement a policy $P$ when $H$ is true is large, and the cost of mistakenly implementing $P$ when the true structure is in $H^{ext}\backslash H$ is small, we may want to test $H^{ext}$. Conversely, if there is a huge cost when we implement $P$ if $H^c$ is true, we may want to test $H^{sub}$. In this case, we sacrifice the benefit when the true structure is in $H\backslash H^{sub}$ to avoid the risk of mistakenly implementing $P$. The following proposition provides the minimal (resp. maximal) strongly binary decidable extension (resp. subset set).
\begin{prop} \label{Prop: smallest binary decidable extension (shrinkage)}
	If $\mathcal{H}_{\tilde{A}}^{snf}(H)=\mathcal{H}_{\tilde{A}}^{wnf}(H)$, then $\mathcal{H}_{\tilde{A}}^{snf}(H)$ is the smallest  strongly binary decidable extension.
	
	If $\mathcal{H}_{\tilde{A}}^{scon}(H)=\mathcal{H}_{\tilde{A}}^{wcon}(H)$, then $\mathcal{H}_{\tilde{A}}^{wcon}(H)$ is the largest strongly binary decidable subset set.	
\end{prop}

For complete structure universes,  $\mathcal{H}_{\tilde{A}}^{snf}(H)=\mathcal{H}_{\tilde{A}}^{wnf}(H)$ and $\mathcal{H}_{\tilde{A}}^{scon}(H)=\mathcal{H}_{\tilde{A}}^{wcon}(H)$ hold automatically, so we can always find a non-trivial strongly binary decidable extension (subset set). In the following, I present an example with a complete structure universe.

\begin{example}[label=exa: Interval Data] (Interval Data) Consider a classical missing data problem where $Y_i^*$ is the unobserved real random variable, bounded above and below by observed variables $Y_i^u$ and $Y_i^l$. In this case, we can consider two primitive random variables $\epsilon_i^u$ and $\epsilon_i^l$, such that $\epsilon_i^u$ is supported on $[0,\infty)$ and $\epsilon_i^l$ is supported on $(-\infty,0]$. Observed variables $Y_i^u$ and $Y_i^l$ are generated through:
	\begin{equation}\label{eq: interval data}
	Y_i^u=Y_i^*+\epsilon_i^u\quad \quad and\quad  \quad Y_i^l=Y_i^*+\epsilon_i^l.
	\end{equation}
	A structure consists of a joint distribution $G^s$ of $(Y_i^*,\epsilon_i^u,\epsilon_i^l)$ that satisfies the support conditions, and the mapping (\ref{eq: interval data}). The structure universe $\mathcal{S}$ contains  all structures with a distribution of $(Y_i^*,\epsilon_i^u,\epsilon_i^l)$ and the mapping (\ref{eq: interval data}). We impose no further assumption, so $\tilde{A}=\mathcal{S}$. Since the mapping outcome in (\ref{eq: interval data}) is unique, the structure universe is complete. Our hypothesis set is $H=\{s: E_{G^s}(Y_i^*)\in[a,b]\}$ and the corresponding hypothesis testing problem is:
	\[
	\mathcal{H}_0: \, E[Y_i^*]\in [a,b] \quad \quad v.s.\quad \quad \mathcal{H}_1: \, E[Y_i^*]\notin [a,b].
	\]
	The non-refutability set associated with $H$ is
	\begin{equation}\label{eq: Interval Data eg, intersec condi}
	\mathcal{H}_{\mathcal{S}}^{snf}(H)=\mathcal{H}_{\mathcal{S}}^{wnf}(H)=\left\{s:\, \left[E_{G^s}(Y_i^*+\epsilon_i^l),E_{G^s}(Y_i^*+\epsilon_i^u)\right]\cap [a,b]\ne \varnothing\right\}.
	\end{equation}
	Indeed, for any $s$ that satisfies the intersection condition above, suppose without loss of generality that $a\in\left[E_{G^s}(Y_i^*+\epsilon_i^l),E_{G^s}(Y_i^*+\epsilon_i^u)\right]$.
	We can construct $\tilde{s}$ such that 
	\[
	\begin{split}
	\tilde{Y}_i^*=Y_i^*+\epsilon_i^u,\quad 
	\tilde{\epsilon}_i^u=0, \quad \text{and}\quad 
	\tilde{\epsilon}_i^l=\epsilon_i^l-\epsilon_i^u\quad a.s.,
	\end{split}
	\]
	and $G^{\tilde{s}}$ is the distribution of $(\tilde{Y}_i^*,\tilde{\epsilon}_i^u,\tilde{\epsilon}_i^l)$. It is easy to see $E_{G^{\tilde{s}}}(\tilde{Y}_i^*)=a$ and $\tilde{\epsilon}_i^u\ge 0$, $\tilde{\epsilon}_i^l\le 0$ almost surely, so support conditions of $(\epsilon_i^u,\epsilon_i^l)$ are satisfied. This implies $\tilde{s}\in \mathcal{H}_{\mathcal{S}}^{snf}(H)$. Conversely, for any $s$ that fails the intersection condition (\ref{eq: Interval Data eg, intersec condi}), for example $E_{G^s}(Y_i^*+\epsilon_i^u)<a$, then for any $\tilde{s}$ such that $M^{\tilde{s}}(G^{\tilde{s}})=M^{{s}}(G^{{s}})$, we have  $E_{\tilde{s}}[Y_i^*]\le E_{\tilde{s}}[Y_i^u]=E_{G^{\tilde{s}}}(Y_i^*+\epsilon_i^u)<a$. As a result, $\tilde{s}\notin H$. The confirmation set can be derived through Proposition \ref{prop: operation of confirmable and refutable set }:
	\[
	\mathcal{H}_{\mathcal{S}}^{scon}(H)=\mathcal{H}_{\mathcal{S}}^{wcon}(H)=\left\{s:\, \left[E_{G^s}(Y_i^*+\epsilon_i^l),E_{G^s}(Y_i^*+\epsilon_i^u)\right]\subseteq [a,b]\right\}.
	\]
	If we want to test $\mathcal{H}_{\mathcal{S}}^{snf}(H)$, a natural statistic is 
	\[
	T_n^{nf}(\mathbb{F}_n)= \sqrt{n}\left[\left(\frac{1}{n}\sum_{i=1}^n Y_i^u-a\right)_-^2+\left(b-\frac{1}{n}\sum_{i=1}^n Y_i^l\right)_-^2\right],
	\]
	where  $(x)_-=\min(0,x)$.
	If we want to test $\mathcal{H}_{\mathcal{S}}^{scon}(H)$, a natural statistic is 
	\[
	T_n^{con}(\mathbb{F}_n)= \sqrt{n}\left[\left(b-\frac{1}{n}\sum_{i=1}^n Y_i^u\right)_-^2+\left(\frac{1}{n}\sum_{i=1}^n Y_i^l-a\right)_-^2\right].
	\]
\end{example}

In the example above, the non-refutability set and confirmation set associated with $H$ are easy to find, while in more complicated structural models, the non-refutable and confirmation sets can be hard to characterize. In a complete structure universe, if $ H$ is not a strongly binary decidable, and $H$ is refutable (resp. confirmable), we want to instead test $s_0\in \mathcal{H}_{\tilde{A}}^{snf}(H)$ (resp. $s_0\in \mathcal{H}_{\tilde{A}}^{wcon}(H)$), which is strongly binary decidable. The following proposition shows that in a complete structure universe, testing $s_0\in \mathcal{H}_{\tilde{A}}^{snf}(H)$  can be equivalently written as a test of the existence of a structure that rationalizes data.

\begin{prop} \label{prop: equivalent decision}
	(Equivalent Decision) Let $s^0$ be the true structure that generates $F$.  If $\mathcal{H}_{\tilde{A}}^{wnf}(H)= \mathcal{H}_{\tilde{A}}^{snf}(H)$, then the following two conditions are equivalent :
	\begin{enumerate}
		\item $\exists s\in H$ such that $F\in M^s(G^s)$.
		\item The true structure $s_0\in \mathcal{H}_{\tilde{A}}^{snf}(H)$.
	\end{enumerate}
	If $\mathcal{H}_{\tilde{A}}^{wcon}(H)= \mathcal{H}_{\tilde{A}}^{scon}(H)$, then the following two conditions are equivalent :
	\begin{enumerate}
		\item $\{s\in \tilde{A}: \,\,F\in M^s(G^s)\}\subseteq H$.
		\item The true structure $s_0\in \mathcal{H}_{\tilde{A}}^{scon}(H)$.
	\end{enumerate}
\end{prop}

\begin{example}[continues=exa: Interval Data]
	In the interval data example above, if there exists a structure $s$ such that $F\in M^s(G^s)$, and $E_{G^s}(Y_i^*)\in [a,b]$, the structure $s$ implies 
	\[
	E_{F}(Y_i^u)\ge E_{G^s}(Y_i^*)\ge a\quad \quad and \quad \quad  E_{F}(Y_i^l)\le E_{G^s}(Y_i^*)\le b.
	\]
	A possible test statistic to test this implication is to use $T_n^{nf}(\mathbb{F}_n)$ defined above. 
	
	On the other hand, if any struture $s$ that can generate $F$ is contained in $H$, the following two  extreme cases:
	\[
	\begin{split}
	Y_i^*= Y_i^u \quad \epsilon_i^u\equiv 0 \quad \epsilon_i^l=Y_i^l-Y_i^u,\\
	Y_i^*= Y_i^l \quad \epsilon_i^l\equiv 0 \quad \epsilon_i^u=Y_i^u-Y_i^l,
	\end{split}
	\]
	must also be included in $H$, which means 
	$E_F[Y_i^u]\le b$ and $E_F[Y_i^l]\ge a 	$
	must hold. A possible statistic to test this implication is to use $T_n^{con}(\mathbb{F}_n)$ defined above. 
\end{example}
%

\section{Application to Treatment Effects}
In this section, I apply the method to \citet{IA1994} with a binary treatment and a binary instrument. The observed outcome variable $Y_i$ and treatment decision $D_i$ are generated through 
\begin{align} \label{eq: potential outcome}
\begin{split}
Y_i&= Y_i(1,1)D_iZ_i+Y_i(0,1)(1-D_i)Z_i+Y_i(1,0)D_i(1-Z_i)+Y_i(0,0)(1-D_i)(1-Z_i),\\
D_i&= D_i(1)Z_i +D_i(0)(1-Z_i),
\end{split}
\end{align}
where $D_i(1),D_i(0)$ are potential treatment decisions, $Y_i(d,z)$ are the potential outcome and $Z_i$ is a binary instrument.

Primitive variables are $\epsilon_i=(D_i(1),D_i(0),Y_i(0,0),Y_i(1,0),Y_i(0,1),Y_i(1,1),Z_i)$ and observed variables are $X_i=(Y_i,D_i,Z_i)$.  Let $\mathcal{Y}$ be the space of $Y_i$ and let $\mathcal{B}$ be a Borel-sigma algebra on $\mathcal{Y}$. The observation space is 
\begin{equation}\label{eq: appli, observation space}
\mathcal{F}=\{F_X(y,d,z):\quad D_i,Z_i\in\{0,1\}\},
\end{equation}
and space of potential distribution 
\begin{equation}\label{eq: appli, primitive space}
\begin{split}
\mathcal{G}=\big\{&G_{\epsilon} \text{ is distribution of $\epsilon$}: \epsilon\in\{0,1\}^2\times\mathcal{Y}^4\times \{0,1\}  \big\}.
\end{split}
\end{equation}
All structures agrees on the functional relation between $X_i$ and $\epsilon_i$ specified in (\ref{eq: potential outcome}). The mapping\footnote{See Definition \ref{def: economic structure}}  $M^s$ is defined as:
\begin{equation}\label{eq: appli, M^s}
\begin{split}
M^s(G^s)=\big\{F\in \mathcal{F}:\, &Pr_{F}(Y_i\in B,D_i=d,Z_i=z)= Pr_{G^s}(Y_i(d,z)\in B,D_i(z)=d,Z_i=z),\\
& \forall B\in \mathcal{B}, \quad d,z\in\{0,1\}
\big\}.
\end{split}
\end{equation}
$M^s$ contains exactly one predicted distribution of observables and all structures are complete. The structure universe $\mathcal{S}$ is:
\begin{equation} \label{eq: appli, potential outcome model space}
\mathcal{S}=\left\{s|\,\, G^s\in \mathcal{G},\,\, M^s \,\,satisfies \,\, (\ref{eq: appli, M^s}) \right\}.
\end{equation}      

Following \citet{kitagawa2015}, I define the following two quantities for all $B\in \mathcal{B}$ and $d\in \{0,1\}$:
	\begin{equation}\label{eq: potential outcome, link between F and G}
	\begin{split}
	P(B,d)\equiv Pr_F(Y_i\in B,D_i=d|Z_i=1), &\\
	Q(B,d)\equiv Pr_F(Y_i\in B,D_i=d|Z_i=0). &\\
	\end{split}
	\end{equation}

The Imbens-Angrist Monotonicity assumption (IA-M) assumes exogeneity, exclusion and monotonicity of the instrument $Z_i$:
\begin{equation} \label{eq: IA instrument assumption }
\begin{split}
A=\bigg\{ s\bigg| &G^s\,\,\text{satisfies}: D_i(1)\ge D_i(0)\,\,\text{a.s.,}\\
& Z_i\perp \left(Y_i(1,1),Y_i(0,1),Y_i(1,0),Y_i(0,0),D_i(1),D_i(0)\right),\\
& Y_i(1,1)=Y_i(1,0) \quad \text{and} \quad Y_i(0,1)=Y_i(0,0)\bigg\}.
\end{split}
\end{equation}
\citet{kitagawa2015} derives the sharp testable implications of the IA-M assumption (\ref{eq: IA instrument assumption }). I reformulate the result in the language of non-refutability sets in the following lemma. 

\begin{lem}\label{lem: testable implication of IA assumption}
	Let $P(\cdot,d)$ and $Q(\cdot,d)$, $d\in\{0,1\}$, be absolutely continuous with respect to some measure $\mu_F$.\footnote{Such dominating measure always exists, for example define $\mu_F(B)=P(B,1)+Q(B,1)+P(B,0)+Q(B,0)$ for all $B\in \mathcal{B}$.} The  non-refutability set associated with IA-M assumption 
	$\mathcal{H}_\mathcal{S}^{snf}(A)$ is the collection of structures $s$ such that if $F\in M^s(G^s)$, then for all Borel set $B$:
	\begin{equation} \label{eq: testable implication of LATE}
	\begin{split}
	P(B,1)&\ge Q(B,1),\\
	Q(B,0)&\ge P(B,0).
	\end{split}
	\end{equation}
\end{lem}

\citet{kitagawa2015} proposes using the core determining class \citep{galichon2011set} such as the class of closed intervals to test  (\ref{eq: testable implication of LATE}). Alternatively, (\ref{eq: testable implication of LATE}) can be equivalently formulated using Radon-Nikodym densities. We will see the advantage of densities when we construct extensions. 
\begin{thm}\label{thm: testable implication in density form}
	 Let $\mu_F$ be the common dominating measure in Lemma \ref{lem: testable implication of IA assumption}.  Let $p(y,d)$ and $q(y,d)$ be defined as:
	\begin{equation} \label{eq: R-D derivatives of P Q}
	p(y,d)=\frac{dP(B,d)}{d\mu_F}\quad\quad q(y,d)=\frac{dQ(B,d)}{d\mu_F}.
	\end{equation}
	Then the testable implication (\ref{eq: testable implication of LATE}) holds if and only if 
	\begin{equation}\label{eq: testable implication in density form}
	\begin{split}
	p(y,1)-q(y,1)\ge 0 \quad \quad &\mu_F-\text{a.s}.\,,\\
	q(y,0)-p(y,0)\ge 0 \quad \quad &\mu_F-\text{a.s}.\, .
	\end{split}
	\end{equation}
\end{thm}
 Our main parameter of interest $\theta$ is the local average treatment effect for compliers:
\begin{equation}\label{eq: appli, formular}
LATE\equiv E[Y_i(1,1)-Y_i(0,0)|D_i(1)=1,D_i(0)=0].
\end{equation}
Under the IA-M assumption $A$, the identified set for LATE is characterized by
\begin{equation}\label{eq: appli, ID set}
LATE_{A}^{ID}(F)=\begin{cases}
\frac{E[Y_i|Z_i=1]-E[Y_i|Z_i=0]}{E[D_i|Z_i=1]-E[D_i|Z_i=0]}\quad \text{if  (\ref{eq: testable implication in density form})  holds},\\
\varnothing \quad \text{otherwise}.\\
\end{cases}
\end{equation}
As shown in Lemma \ref{lem: testable implication of IA assumption}, the IA-M assumption $A$ is refutable. In most empirical applications, researchers do  not test this implication, neither do they specify what should be done when the testable implication is rejected. In the next section, I use the relaxed assumption approach to find relaxed assumptions $\tilde{A}$ such that $\tilde{A}$ is non-refutable, find the identified set under $\tilde{A}$, and discuss the estimation and inference on LATE under $\tilde{A}$.

\subsection{Extensions of the IA-M Assumption}
In this section, I will first show that the IA-M assumption have an alternative representation. As discussed in Section \ref{section: minimal deviation method}, different representations of an assumption can result in different extensions. 
The canonical representation (\ref{eq: IA instrument assumption }) and the alternative representation will be used to constructed different extensions. I will then look at the maximal relaxation in Definition \ref{prop: exist of strong extension} and show that the identified set for LATE under the maximal relaxation does not satisfy Property \ref{property: continuity}. Then I proceed to construct extensions using the minimal deviation method.

The following is an alternative representation of the IA-M assumption that will be used throughout this section.

\begin{lem} \label{lem: alternative representation of IA-M assumption}
	The IA-M assumption defined in (\ref{eq: IA instrument assumption }) can be equivalently written as the intersection: $A=A^{ER}\cap A^{TI}\cap A^{EM-NTAT}\cap A^{ND}$ where:\begin{enumerate}
		\item $A^{ER}=\left\{s\big|Y_i(1,1)=Y_i(1,0) \quad and \quad Y_i(0,1)=Y_i(0,0) \right\}$ is the exclusion restriction;
		\item $A^{TI}=\left\{s\big| Z_i\perp \left(Y_i(1,1),Y_i(0,1),Y_i(1,0),Y_i(0,0)\right)|D_i(1),D_i(0) \right\}$ is the type independent instrument assumption;
		\item Assumption $A^{EM-NTAT}$ is the set of structures $s$ such that the measures of always takers and never takers are independent of $Z_i$, i.e. \[
		\begin{split}
		E_{G^s}[\mathbbm{1}(D_i(1)=D_i(0)=1)|Z_i=1]&=E_{G^s}[\mathbbm{1}(D_i(1)=D_i(0)=1)|Z_i=0],\\
		E_{G^s}[\mathbbm{1}(D_i(1)=D_i(0)=0)|Z_i=1]&=E_{G^s}[\mathbbm{1}(D_i(1)=D_i(0)=0)|Z_i=0].
		\end{split}\]
		\item $A^{ND}=\left\{ s\big| G^s\,\,satisfies: D_i(1)\ge D_i(0)\right\}$ is the no defiers assumption.
	\end{enumerate}
\end{lem}

\subsubsection{The Maximal Extension with Exclusion Restriction and `No Defiers'}
To fix the idea, let's consider the case that $A^{ER}$ and $A^{ND}$ holds but we relax the independent instrument assumption. Moreover we consider the extension set $\tilde{A}^{max}\equiv \left(A\cup[{\mathcal{H}_{\mathcal{S}}^{snf}(A)}]^c\right) \cap A^{ER}\cap A^{ND}$. $\tilde{A}^{max}$ is the maximal strong extension defined in Proposition \ref{prop: exist of strong extension} intersected with the exclusion restriction and the `No Defiers' assumption. 

\begin{prop}\label{prop: appli, maximal extension arbitrary instrument}
$\tilde{A}^{max}\equiv \left(A\cup[{\mathcal{H}_{\mathcal{S}}^{snf}(A)}]^c\right) \cap A^{ER}\cap A^{ND}$ is a strong extension. The closure of the identified set for LATE under $\tilde{A}^{max}$ is 
\[
\overline{{LATE}^{ID}_{\tilde{A}^{max}}(F)}=\begin{cases}
\frac{E[Y_i|Z_i=1]-E[Y_i|Z_i=0]}{E[D_i|Z_i=1]-E[D_i|Z_i=0]}\quad &\text{if (\ref{eq: testable implication in density form}) holds for } \,\,F,\\
\left[ \underline{\mathcal{Y}}_{P(B,1)}-\bar{\mathcal{Y}}_{Q(B,0)},\bar{\mathcal{Y}}_{P(B,1)}-\underline{\mathcal{Y}}_{Q(B,0)}\right] &\text{otherwise},
\end{cases}
\]
where for $V\in\{P,Q\}$, $\underline{\mathcal{Y}}_{V(B,0)}$  is the lower bound of the support of $Y_i$ under measure $V(B,0)$, and $\bar{\mathcal{Y}}_{V(B,1)}$ is the upper bound of the support of $Y_i$ under measure $V(B,1)$.
\end{prop}
The $\tilde{A}^{max}$ above allows arbitrary dependence of the instrument  on the potential outcomes whenever the testable implication (\ref{eq: testable implication in density form}) fails. First, we should note that the identified set for LATE is very unstable when $F$ satisfies $p(y,1)-q(y,1)=0$ for some $y\in \mathcal{Y}$. Whenever we perturb $F$ slightly such that $p(y,1)-q(y,1)<0$ and (\ref{eq: testable implication in density form})  fails, the identified set for LATE explodes. Second, the identification set ${LATE}^{ID}_{\tilde{A}^{max}}(F)$ is not any better than the $LATE^{ID}_A(F)$ in equation (\ref{eq: appli, ID set}). In terms of interpretation, an uninformative identified set\footnote{Note that the identification result in Proposition \ref{prop: appli, maximal extension arbitrary instrument} contains only support information when  (\ref{eq: testable implication in density form}) fails.} for LATE is not different from an empty identified set. This is because whenever $F$ fails (\ref{eq: testable implication in density form}), we give up the `Independent Instrument' assumption. Therefore the remaining assumptions $A^{ER}$ and $A^{ND}$ cannot generate any restrictions on the parameter of interest. As a result, I focus on deriving extensions using minimal deviation method in Definition \ref{def: minimal deviation extension}. I will relax the `No Defiers' and the independent instrument assumption in the following. An extension that relaxes the independent instrument assumption and an extension that relaxes  the exclusion restriction are given Appendix \ref{append: minimal marginal difference}.

\subsubsection{The Minimal Defiers Extension}
Recall that $\mathcal{S}$ is a complete structure universe. I consider a strong extension that use measure of defiers as deviation from the no defiers assumption.  The extension relaxes the independent instrument to a type independent instrument assumption. I first define the measure of defiers in $G^s$ as  
\begin{equation}\label{eq: defiers measure}
m^d(s) = E_{G^s}\left[\mathbbm{1}\{D_i(1)=0,D_i(0)=1\} \right].
\end{equation}

%

\begin{assumption}\label{assumption:type ind instru}
	 Let
	$m^{min}(F)\equiv \inf\{m^d(s): F\in M^s(G^s) \,\, and \,\, s\in A^{ER}\cap A^{TI}\cap A^{EM-NTAT}\}$ be the minimal defier amount under $F$. We call \[\tilde{A}=\cup_{F\in \mathcal{F}} \left\{s\in A^{ER}\cap A^{TI}\cap A^{EM-NTAT}: m^d(s)=m^{min}(F),\,\,F\in M^s(G^s)\right\}\] the minimal defiers extension with type independent instrument. 
\end{assumption}

In the above extension, I also relax the independent instrument condition. The type independence assumption is also used in other empirical contexts to study LATE (e.g. see \citet{kedagni2019}). This is because, by \citet{kitagawa2009identification}, exclusion restrictions (ER) and instrument condition (IV) has testable implication, so any non-refutable relaxation should relax either ER or IV condition. The second condition in Assumption \ref{assumption:type ind instru} requires the measure of always takers (AT) and never takers (NT) to be independent of the instrument. 

\begin{prop} \label{prop: IA strong extension}
	The extension $\tilde{A}$ defined in Assumption \ref{assumption:type ind instru}  is a strong extension of $A$. 
\end{prop}
\begin{proof}
	Suffice to check conditions in Proposition \ref{prop: minimal deviation as strong extension}, see  Appendix \ref{section: Proofs in section 3}.
\end{proof}
\begin{remark}
 To emphasize that extensions $\tilde{A}$ constructed by minimal deviation method may not be strong extensions, I provide two examples of extensions in Appendix \ref{append: minimal marginal difference} that are LATE-consistent extension but are not strong  extensions. 
\end{remark}

\subsection{Identified Set under Different Extensions}
This section describes the identified set for LATE under $\tilde{A}$ in Assumption \ref{assumption:type ind instru}.  Let
$
\mathcal{Y}_d\equiv\{y\in \mathcal{Y}: (-1)^{d}[q(y,d)-p(y,d)]\ge 0\}
$
be the collection of $y\in \mathcal{Y}$ such that the density differences are positive. 
\begin{assumption}\label{assumption: well-defined LATE}
	There exists a constant $c\ge 0$ such that: (i) $Pr_F(Z_i=1)\in (c,1-c)$; (ii) $Q(\mathcal{Y}_0,0)-P(\mathcal{Y}_0,0)>c$ and $P(\mathcal{Y}_1,1)-Q(\mathcal{Y}_1,1)>c$.
\end{assumption}

Assumption \ref{assumption: well-defined LATE} is a regularity assumption. For the identification result, we only need it to hold with $c=0$ so that LATE is well-defined. For inference purpose, I require $c>0$ to avoid weak instrument issue.

\begin{prop}\label{prop: identified under minimal defiers type ind inst}
	Let the extension $\tilde{A}$ satisfies Assumptions \ref{assumption:type ind instru}. If Assumption \ref{assumption: well-defined LATE} holds for $c=0$, then the identified $ {LATE}^{ID}_{\tilde{A}}$ satisfies
	\begin{equation}\label{eq: Identified LATE, type indep}
	{LATE}_{\tilde{A}}^{ID}(F)=\frac{\int_{\mathcal{Y}_1}{y (p(y,1)-q(y,1))}d\mu_F(y)}{P(\mathcal{Y}_1,1)-Q(\mathcal{Y}_1,1)}
	-\frac{\int_{\mathcal{Y}_0}{y (q(y,0)-p(y,0))}d\mu_F(y)}{Q(\mathcal{Y}_0,0)-P(\mathcal{Y}_0,0)},
	\end{equation}
	where $\mathcal{Y}_1=\{y\in\mathcal{Y}: p(y,1)-q(y,1)\ge 0\}$ and $\mathcal{Y}_0=\{y\in\mathcal{Y}: q(y,0)-p(y,0)\ge 0\}$.
	Note that LATE is point identified.
\end{prop}

In Assumptions \ref{assumption:type ind instru}, I require that potential outcomes are independent of the instrument conditional on compliers, i.e. $\{Y_i(d,z)\}_{d,z\in\{0,1\}}\perp Z_i\big| (D_i(1)=1,D_i(0)=0)$. As a result, the identified probability of compliers conditioned on $Z_i=1$ is $P(\mathcal{Y}_1,1)-Q(\mathcal{Y}_1,1)$ and the identified probability of compliers conditioned on $Z_i=0$ is $Q(\mathcal{Y}_0,0)-P(\mathcal{Y}_0,0)$. 


As we discussed in Section 2, we want the identified set for LATE under the chosen extensions to be a continuous correspondence with respect to the distribution of observables $F$. If the identified set for LATE is discontinuous under the extension, we may get spuriously informative bound for LATE due to sampling error. I compare the identified set for LATE under the maximal extension in Proposition \ref{prop: appli, maximal extension arbitrary instrument} and the identified set for LATE under the minimal defiers extension in Proposition \ref{prop: identified under minimal defiers type ind inst} in terms of Property \ref{property: continuity}. I equip $\mathcal{F}$ with the Sobolev norm: $||F||_{1,\infty}\equiv \max_{i=0,1} ||F^{(i)}||_{\infty}$, where $F^{(i)}$ is the $i$-th Radon-Nikodym density of $F$ with respect \mbox{to $\mu_F$}. 
\begin{prop} \label{prop: continuity of estimator in different cases}
	Let $\tilde{A}_1$ be the maximal extension and ${LATE}^{ID}_{\tilde{A}_1}(F)$ be the corresponding identified set defined in Proposition \ref{prop: appli, maximal extension arbitrary instrument}. Let $\tilde{A}_2$ be the minimal defiers extension defined in Assumption \ref{assumption:type ind instru} and ${LATE}^{ID}_{\tilde{A}_2}(F)$ be the corresponding identified set defined in (\ref{eq: Identified LATE, type indep}). Suppose $\forall F\in \mathcal{F}$, the support of $Y_i$ is bounded above by $M^u_s$ and bounded below by $M^l_s$, then ${LATE}^{ID}_{\tilde{A}_1}(F)$ is not upper hemicontinuous with respect to the Sobolev norm $||\cdot ||_{1,\infty}$, and ${LATE}^{ID}_{\tilde{A}_2}(F)$ is continuous with respect to $||\cdot||_{1,\infty}$. 
\end{prop} 

Proposition \ref{prop: continuity of estimator in different cases} shows that the maximal extension is not a good choice if the parameter of interest is LATE.

\subsection{Estimation and Inference}
The identification result in Proposition \ref{prop: identified under minimal defiers type ind inst} relies on the sets $\mathcal{Y}_1,\mathcal{Y}_0$. Throughout this section, I focus on the estimation and inference problem when $Y_i$ is continuously distributed on, and $\mu_F$ is the Lebesgue measure. 

\begin{assumption}\label{assump: continuous Y}
	$Y_i$ is continuously distributed with unbounded support and the measure $P(B,d)$, $Q(B,d)$ is absolutely continuous with respect to the Lebesgue measure. 
\end{assumption}
To estimate $\hat{\mathcal{Y}}_0,\hat{\mathcal{Y}}_1$, I estimate the density $p(y,d)$ and $q(y,d)$ using kernel density estimators:
\begin{equation} \label{eq: kernel density estimator}
\begin{split}
f_h(y,1)&=\frac{\frac{1}{h_n}\sum_{i=1}^n K\left(\frac{Y_i-y}{h}\right)\mathbbm{1}(D_j=1,Z_j=1)}{\sum_{i=1}^n \mathbbm{1}(Z_j=1)}-\frac{\frac{1}{h_n}\sum_{i=1}^n K\left(\frac{Y_i-y}{h}\right)\mathbbm{1}(D_j=1,Z_j=0)}{\sum_{i=1}^n \mathbbm{1}(Z_j=0)},\\
f_h(y,0)&=\frac{\frac{1}{h_n}\sum_{i=1}^n K\left(\frac{Y_i-y}{h}\right)\mathbbm{1}(D_j=0,Z_j=0)}{\sum_{i=1}^n \mathbbm{1}(Z_j=0)}-\frac{\frac{1}{h_n}\sum_{i=1}^n K\left(\frac{Y_i-y}{h}\right)\mathbbm{1}(D_j=0,Z_j=1)}{\sum_{i=1}^n \mathbbm{1}(Z_j=1)}.
\end{split}
\end{equation}

\begin{assumption}\label{assump: tail density sign}
	There exist constants $M_l$ and $M_u$ such that for $d=0,1$ such that  $\mathcal{Y}_d\cap [M_u,\infty) \in \{\varnothing, [M_u,\infty)\}$ and $\mathcal{Y}_d\cap (-\infty,M_l] \in \{\varnothing, (-\infty,M_l]\}$. Moreover, we know  $\mathcal{Y}_d\cap [M_u,\infty)$ and $\mathcal{Y}_d\cap (-\infty,M_l]$. 
\end{assumption}
The assumption above assumes that the sign of $p(y,d)-q(y,d)$ is known and fixed in the large value of $y$. As a result, we only need to estimate the set $\mathcal{Y}_d\cap [M_l,M_u]$. 
\begin{example}[Gaussian tail dominance.]
Suppose $p(y,d)$ and $q(y,d)$ have Gaussian tails:   $p(y,d)=C_pe^{-y^2/\sigma_p(d)^2}$ and $q(y,d)=C_qe^{-y^2/\sigma_q(d)^2}$  for $|y|>C^{tail}>0$. If $\sigma_p(1)>\sigma_q(1)$, then $\mathcal{Y}_1\cap [C^{tail},\infty)=[C^{tail},\infty)$ and $\mathcal{Y}_1\cap (-\infty,-C^{tail}]=(-\infty,-C^{tail}]$.
\end{example}

  Define the upper tail set $\mathcal{Y}_d^{ut}= \mathcal{Y}_d\cap [M_u,\infty)$ and the lower tail set $\mathcal{Y}_d^{lt}= \mathcal{Y}_d\cap (-\infty,M_l]$ and we estimate 
\[
\hat{\mathcal{Y}}_d(b_n)= \{y\in (M_l,M_u): f_h(y,d)\ge b_n\}\cup \mathcal{Y}_d^{ut}\cup \mathcal{Y}_d^{lt},
\]
where $b_n$ is a sequence of positive constants that converges to zero. The estimated set above only uses density $f_h(y,d)$ to distinguish whether $y\in \mathcal{Y}_d$ in the range $(M_l,M_u)$, and uses the known tail sign in Assumption \ref{assump: tail density sign} directly. When the relaxed assumption is defined in Assumption \ref{assumption:type ind instru}, I construct an  estimator of $LATE^{ID}_{\tilde{A}}(F)$ (\ref{eq: Identified LATE, type indep}) as:
\begin{equation}\label{eq: estimator, type indep}
\begin{split}
\widehat{LATE} &= \frac{\frac{1}{n} \sum_{i=1}^n Y_i \left[ \frac{\mathbbm{1}(D_i=1,Z_i=1)}{\frac{1}{n}\sum_{j=1}^n \mathbbm{1}(Z_j=1)}- \frac{\mathbbm{1}(D_i=1,Z_i=0)}{\frac{1}{n}\sum_{j=1}^n \mathbbm{1}(Z_j=0)} \right] \mathbbm{1}(Y_i\in \hat{\mathcal{Y}}_1(b_n))}{\frac{1}{n}\sum_{i=1}^n\left[ \frac{\mathbbm{1}(D_i=1,Z_i=1)}{\frac{1}{n}\sum_{j=1}^n \mathbbm{1}(Z_j=1)}- \frac{\mathbbm{1}(D_i=1,Z_i=0)}{\frac{1}{n}\sum_{j=1}^n \mathbbm{1}(Z_j=0)} \right]\mathbbm{1}(Y_i\in \hat{\mathcal{Y}}_1(b_n))}\\
&- \frac{\frac{1}{n} \sum_{i=1}^n Y_i \left[ \frac{\mathbbm{1}(D_i=0,Z_i=0)}{\frac{1}{n}\sum_{j=1}^n \mathbbm{1}(Z_j=0)}- \frac{\mathbbm{1}(D_i=0,Z_i=1)}{\frac{1}{n}\sum_{j=1}^n \mathbbm{1}(Z_j=1)} \right] \mathbbm{1}(Y_i\in \hat{\mathcal{Y}}_0(b_n))}{\frac{1}{n}\sum_{i=1}^n\left[ \frac{\mathbbm{1}(D_i=0,Z_i=0)}{\frac{1}{n}\sum_{j=1}^n \mathbbm{1}(Z_j=0)}- \frac{\mathbbm{1}(D_i=0,Z_i=1)}{\frac{1}{n}\sum_{j=1}^n \mathbbm{1}(Z_j=1)} \right]\mathbbm{1}(Y_i\in \hat{\mathcal{Y}}_0(b_n))}.
\end{split}
\end{equation}

\subsubsection*{Limit Distribution of $\widehat{LATE}$}
I present the limit distribution of $\widehat{LATE}$ defined in (\ref{eq: estimator, type indep}). The following assumptions are sufficient to guarantee $\widehat{LATE}$ in (\ref{eq: estimator, type indep}) will converge to a normal distribution.
\begin{assumption} \label{assumption: kernel }
	The kernel function $K$ satisfies: (i) $K(u)$ is continuous and supported on $[-A,A]$ and  $\int_u K(u) du=1$; (ii) $\int_u uK(u)du=0$; (iii) $\int u^2K(u)du<\infty$. 
\end{assumption}

\begin{assumption}\label{assumption: density}
	The conditional distribution $F(y|D_i=k,Z_i=l)$ has a density $f(y|k,l)$ for all $k,l\in\{0,1\}$, and $f''(y|k,l)$ exists and is uniformly bounded by a constant $c_f$; (iii) $E(Y_i^{2+\delta})<\infty$ for some $\delta>0$. 
\end{assumption}
The above two assumptions are standard in literature and guarantee the density difference estimator $f_h(y,d)$ will converges uniformly in probability to its limit $(-1)^{1-d}(p(y,d)-q(y,d))$ at polynomial rate. 
\begin{assumption}\label{assumption: trimming bias 1} 
	Let $f(y,1)=p(y,1)-q(y,1)$ and $f(y,0)=q(y,0)-p(y,0)$. The following condition holds for any sequence $b_n\rightarrow 0_+$:  $\int_{M_l}^{M_u} |f(y,d)| \mathbbm{1}(-b_n\le f(y,d)\le b_n)dy =O(b_n^2).$ 
\end{assumption}
Assumption  \ref{assumption: trimming bias 1} controls the bias from trimming $\{y\in[M_l,M_u]:0<f_h(y,d)<b_n\}$. Essentially, we rule out all outcome distributions such that  $\{y:f(y,d)=0\}$ has a positive measure.  This assumption is imposed to remove the bias from sampling error in kernel estimator $f_h(y,d)$. Assumption \ref{assumption: trimming bias 1} can be replaced by a sufficient primitive condition.
\begin{assumption}\label{assumption: trimming bias 2}
	Let $M_0<\infty$ be a positive integer. For $d=0,1$, the set $\mathcal{C}_d=\{y: f(y,d)=0,\,\,y\in[M_l,M_u]\}$ has at most $M_0$ points. Let $B(\mathcal{C}_d,\delta)=\cup_{y\in\mathcal{C}_d} B(y,\delta)$ be the $\delta$-neighborhood of $\mathcal{C}_d$ for $d=0,1$. For both $d=0,1$, we have $\sup_{y\in B(\mathcal{C}_d,\delta)}|\frac{d(f(y,d))}{dy}|>1/C$ for some $C,\delta>0$.
\end{assumption}
\begin{lem} \label{lem: appli, primitive condition on trimming bias}
	Assumption \ref{assumption: trimming bias 2} implies Assumption \ref{assumption: trimming bias 1}.
\end{lem}
\begin{proof}
	By the bounded density condition, the Lebesgue measure of set $\{y:\mathbbm{1}(-b_n\le f(y,d)\le b_n)\}$ is less than $CM_0b_n$. Therefore
	\[\int_{M_l}^{M_u} |f(y,d)| \mathbbm{1}(-b_n\le f(y,d)\le b_n)dy\le CM_0b_n^2.\]
	So Assumption \ref{assumption: trimming bias 2} implies Assumption \ref{assumption: trimming bias 1}.
\end{proof}

\begin{thm} \label{thm: asymptotic property of LATE}
	Let $\widehat{LATE}$ be defined in (\ref{eq: estimator, type indep}) and ${LATE}^{ID}_{\tilde{A}}(F)$ be defined in (\ref{eq: Identified LATE, type indep}). Suppose Assumption \ref{assumption: well-defined LATE} holds for $c>0$ and 
	Assumption
	 \ref{assump: continuous Y} -\ref{assumption: trimming bias 1} hold. Let $b_n\asymp n^{-1/4}/\log n$ and $h_n\asymp n^{-1/5}$, then 
	$\sqrt{n}(\widehat{LATE}-{LATE}_{\tilde{A}}^{ID}(F))\rightarrow_d N(0,\Pi'\Gamma\Sigma\Gamma'\Pi)$, where 
	\begin{equation*}
	\footnotesize
	\Sigma=Var\begin{pmatrix}
	\mathbbm{1}(Z_i=0)\\
	\mathbbm{1}(Z_i=1)\\
	Y_i \mathbbm{1}(D_i=1,Z_i=1)\mathbbm{1}(Y_i\in \mathcal{Y}_1)\\
	Y_i \mathbbm{1}(D_i=1,Z_i=0)\mathbbm{1}(Y_i\in \mathcal{Y}_1)\\
	Y_i \mathbbm{1}(D_i=0,Z_i=0)\mathbbm{1}(Y_i\in \mathcal{Y}_0)\\
	Y_i \mathbbm{1}(D_i=0,Z_i=1)\mathbbm{1}(Y_i\in \mathcal{Y}_0)\\
	\mathbbm{1}(D_i=1,Z_i=1)\mathbbm{1}(Y_i\in \mathcal{Y}_1)\\
	\mathbbm{1}(D_i=1,Z_i=0)\mathbbm{1}(Y_i\in \mathcal{Y}_1)\\
	\mathbbm{1}(D_i=0,Z_i=0)\mathbbm{1}(Y_i\in \mathcal{Y}_0)\\
	\mathbbm{1}(D_i=0,Z_i=1)\mathbbm{1}(Y_i\in \mathcal{Y}_0)\end{pmatrix},
	\end{equation*}
 	and matrices $\Pi$ and $\Gamma$ are specified as by
	{\footnotesize
		\[
		\Pi=\begin{pmatrix}
		\frac{1}{\pi_3}\\  -\frac{1}{\pi_4}\\ -\frac{\pi_1}{\pi_3^2}\\ \frac{\pi_2}{\pi_4^2}
		\end{pmatrix},\quad\pi\equiv\begin{pmatrix} \pi_1\\ \pi_2\\ \pi_3\\ \pi_4
		\end{pmatrix}= \begin{pmatrix}
		\int_{\mathcal{Y}_1} y(p(y,1)-q(y,1))dy\\
		\int_{\mathcal{Y}_0} y(q(y,0)-p(y,0))dy\\
		\int_{\mathcal{Y}_1} (p(y,1)-q(y,1))dy\\
		\int_{\mathcal{Y}_0} (q(y,0)-p(y,0))dy
		\end{pmatrix},	 \quad 	\Gamma=\frac{1}{Pr(Z_i=1)Pr(Z_i=0)}\begin{pmatrix}
		\Gamma_1 &\Gamma_3 &\bm{0}_{2\times 4}\\
		\Gamma_2 &\bm{0}_{2\times 4} &\Gamma_3 
		\end{pmatrix},
		\]where 
		\begin{equation*}
		\Gamma_1=
		\begin{pmatrix}
		E[Y_i\mathbbm{1}(D_i=1,Z_i=1)\mathbbm{1}(Y_i\in \mathcal{Y}_1)] & -E[Y_i\mathbbm{1}(D_i=1,Z_i=0)\mathbbm{1}(Y_i\in \mathcal{Y}_1)]\\
		-E[Y_i\mathbbm{1}(D_i=0,Z_i=1)\mathbbm{1}(Y_i\in \mathcal{Y}_0)] & E[Y_i\mathbbm{1}(D_i=0,Z_i=0)\mathbbm{1}(Y_i\in \mathcal{Y}_0)]
		\end{pmatrix},
		\end{equation*}
		\begin{equation*}
		\Gamma_2=\begin{pmatrix}
		E[\mathbbm{1}(D_i=1,Z_i=1)\mathbbm{1}(Y_i\in \mathcal{Y}_1)] & -E[\mathbbm{1}(D_i=1,Z_i=0)\mathbbm{1}(Y_i\in \mathcal{Y}_1)]\\
		-E[\mathbbm{1}(D_i=0,Z_i=1)\mathbbm{1}(Y_i\in \mathcal{Y}_0)] & E[\mathbbm{1}(D_i=0,Z_i=0)\mathbbm{1}(Y_i\in \mathcal{Y}_0)]
		\end{pmatrix},
		\end{equation*}
		\begin{equation*}
		\Gamma_3=\begin{pmatrix}
		Pr(Z_i=0) & -Pr(Z_i=1) &0 &0\\
		0 &0 &Pr(Z_i=1) & -Pr(Z_i=0)
		\end{pmatrix}.
		\end{equation*}
	
}
\end{thm}

\begin{Corollary}\label{Corollary: CI for LATE}
	Let $(\hat{\Gamma},\hat{\Pi},\hat{\Sigma})\rightarrow_p ({\Gamma},{\Pi},{\Sigma})$, and let $\hat{\sigma}=\sqrt{\hat{\Pi}'\hat{\Gamma}\hat{\Sigma}\hat{\Gamma}'\hat{\Pi}}$. Then the set 
	\begin{equation}\label{eq: appli, confi interval for LATE}
	\left[\widehat{LATE}-\frac{\hat{\sigma}}{\sqrt{n}}\Phi(\frac{\alpha}{2}),\widehat{LATE}+\frac{\hat{\sigma}}{\sqrt{n}}\Phi(1-\frac{\alpha}{2})\right]
	\end{equation}
	is a valid $\alpha$-confidence interval for ${LATE}_{\tilde{A}}^{ID}(F)$, where $\Phi$ is the normal CDF function.
\end{Corollary}

Theorem \ref{thm: asymptotic property of LATE} shows that the LATE estimator in (\ref{eq: Identified LATE, type indep}) is $\sqrt{n}$ consistent. Once the matrices $\Pi$, $\Gamma$ and $\Sigma$ are estimated by consistent estimators, we can test hypothesis such as $H_0:{LATE}_{\tilde{A}}^{ID}(F)=0$. Since LATE is point identified, conventional hypothesis testing method can achieve structural size control and test consistency simultaneously. However, Assumption \ref{assump: tail density sign} requires the econometrician to know the sign of tail behavior of $p(y,1)-q(y,1)$ and $q(y,0)-p(y,0)$. In some empirical application, we may want to be agnostic about tail signs or only impose less restrictive conditions on tail signs. In this case, we can calculate the confidence interval for each possible tail condition, and then take the union, but this confidence interval will be conservative. 


\subsubsection*{Simulation}
This section illustrates the finite sample performance of the proposed inference method. I consider two simulation settings. In the first setting, the IA-M assumption is violated, and the goal of the simulation is to see how the inference method works under the known and unknown tail signs. In the second setting, the IA-M assumption is not violated, and the goal of the simulation is to compare the numerical difference of the 2SLS estimator of (\ref{eq: appli, ID set}) and estimator (\ref{eq: estimator, type indep}).

\subsubsection*{\textit{Simulation Setting I}}
 Instead of simulating the primitive variable $Y_i(d,z),D_i(z)$, I directly simulate the distribution of observed variable such that $Pr(Z_i=1)=0.6$ and 
\[
\begin{split}
p(y,1)=p(y,0)=\frac{1}{2\sqrt{2\pi}}exp(-\frac{(x-3)^2}{2}),\\
q(y,1)=q(y,0)=\frac{1}{2\sqrt{6\pi}}exp(-\frac{(x-2.5)^2}{6}).
\end{split}
\]
In this simulation, Assumptions \ref{assumption: well-defined LATE}, \ref{assump: continuous Y} and \ref{assumption: density}  are satisfied. The trimming band is $[M_l,M_u]=[-2.5,7]$ and Assumption \ref{assump: tail density sign} is satisfied since $\mathcal{Y}^{ut}_1=\mathcal{Y}^{lt}_1=\varnothing$,  $\mathcal{Y}^{ut}_0=[M_u,\infty)$, and $\mathcal{Y}^{lt}_0=(-\infty,M_l]$.  Assumption \ref{assumption: trimming bias 2} is satisfied since $p(y,1)-q(y,1)=0$ and $q(y,0)-p(y,1)=0$ have two solutions in interval $[M_l,M_u]$ and the derivatives are bounded away from zero (see Figure \ref{fig: density in the simulation setting 1}).
\begin{figure}[H]
	\centering
	\includegraphics[width=0.7\linewidth]{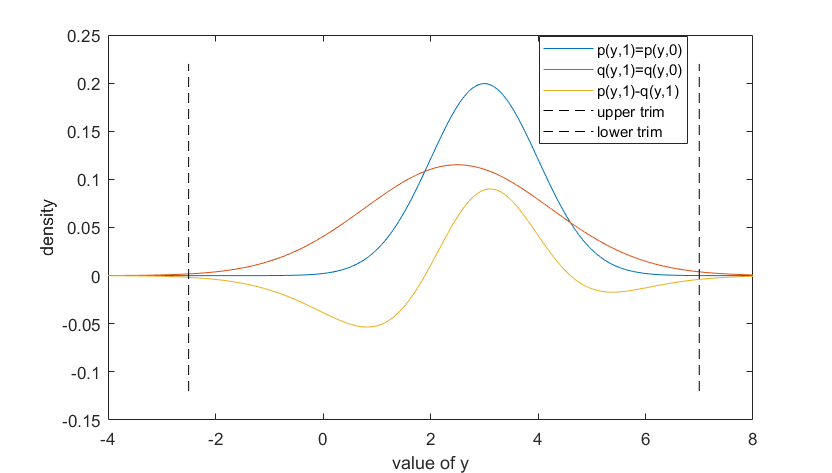}
	\caption{Density of $p(y,d)$ and $q(y,d)$ of $d\in\{0,1\}$, and their difference.}
	\label{fig: density in the simulation setting 1}
\end{figure}
The true identified value of ${LATE}^{ID}_{\tilde{A}}(F)$ is  $1.7385$. Simulation results are given in Table \ref{table: simulation for inference}. Coverage probability are calculated from $1000$ replications, and I compare the coverage probability under different sample size $n$ in each replication and the choice of trimming constant $b$. The row of `known tail' in Table \ref{table: simulation for inference} corresponds to the constraints $\mathcal{Y}^{ut}_1=\mathcal{Y}^{lt}_1=\varnothing$,  $\mathcal{Y}^{ut}_0=[M_u,\infty)$, and $\mathcal{Y}^{lt}_0=(-\infty,M_l]$
as in the simulation design. The row of `Conservative' in Table \ref{table: simulation for inference} corresponds to the case that the tail set is unknown, and I take the union of confidence intervals under all $16$ possible tail conditions.
\begin{table}[H]
	\centering
	\caption{Finite Sample Coverage Probability}\label{table: simulation for inference}
	\begin{tabular}{lccc}
		$\alpha=0.05$ & n=1000 & n=5000 & n=5000 \\
		&(b=0.02,h=0.4) & (b=0.012,h=0.2) &(b=0.0135,h=0.2)\\\hline
		Known Tail    & 0.965  &0.963  &0.935       \\
		Conservative  & 0.9825 &0.988  &0.975     
	\end{tabular}
\end{table}

The simulation result shows that if we can correctly impose the tail condition as in the known tail case, inference on the true LATE value based on (\ref{eq: appli, confi interval for LATE}) is asymptotically exact but can be sensitive to the choice of trimming sequence $b_n$. If we want to be agnostic about the true tail condition, the union method is conservative. 

\subsubsection*{Simulation Setting II}
When the IA-M assumption holds, by \citet{vytlacil2002independence}, the potential outcome model is equivalent to the latent index model. In this simulation setting, let $U_i\sim U[0,1]$ and $D_i=\mathbbm{1}(0.2+0.6Z_i>U_i)$, where $U[0,1]$ is a uniform distribution on interval $[0,1]$, and $Z_i\sim \mathtt{Bernoulli}(0.5)$ is independent of $U_i$. The potential outcome $(Y_i(1),Y_i(0))\sim N(\mu,\Sigma)$, where $\mu=[2,1.5]$ and $Var(Y_i(1))=2$, $Var(Y_i(0))=1.5$, $Corr(Y_i(1),Y_i(0))=0.7$. We can show $LATE^{ID}(F)=0.5$.

Let $\widehat{LATE}^{Wald}$ denote the Wald estimator of $LATE^{ID}(F)$, and let $\widehat{LATE}^{New}$ denote the estimator in equation (\ref{eq: estimator, type indep}). Table \ref{table: comparison of LATE-wald and my method when IAM hold} shows some summary statistics of the numerical difference between these two estimators over $m=1000$ simulation replications. The first row shows that when the sample size increase, the numerical difference of the two estimator converges to zero in the second moment. The second row in Table \ref{table: comparison of LATE-wald and my method when IAM hold} reports the efficiency loss of $\widehat{LATE}^{New}$. We see that at finite sample, not imposing IA-M assumption when it holds can lead to efficiency loss, but the efficiency loss decreases with sample size. 

\begin{table}[H]
	\centering
	\caption{Monte Carlo simulation over $1000$ replications.}
	\label{table: comparison of LATE-wald and my method when IAM hold}
\begin{threeparttable}
\begin{tabular}{lcccc}
			&$n=500$ & $n=1000$ & $n=2500$ & $n=5000$ \\ \hline 
	 $Var\left[\widehat{LATE}^{Wald}-\widehat{LATE}^{New}\right]$ &0.0097 & 0.0031    & 0.0011    &0.0004   \\[1ex]
	$\sigma(\widehat{LATE}^{New})-\sigma(\widehat{LATE}^{Wald})$ &0.020 &.008 &0.006 &0.003\\[1ex]
	MSE($\widehat{LATE}^{New}$)$-$MSE($\widehat{LATE}^{Wald}$) &0.0082 &0.0023 &0.0011 &0.0004
\end{tabular}
	\end{threeparttable}
\end{table}

\subsection{Empirical Illustration}
In this section, I apply my results in Proposition \ref{prop: identified under minimal defiers type ind inst} and Theorem \ref{thm: asymptotic property of LATE} to \citet{card1993using}, who studied the causal effect of college attendance on earnings. In this application, the outcome variable $Y_i$ is an individual $i$'s log wage in 1976, $D_i=1$ means individual $i$ attended a four-year college, and $Z_i=1$ means the individual was born near a four-year college. This data set has been used by both \citet{kitagawa2015} and \citet{mourifie2017testing} to test the IA-M assumption, and they both reject the IA-M assumption. If a child grew up near a college, he or she may hear more stories of heavy tuition burden, which may discourage him or her from attending college. On the other hand, if this child grew up far away from a college, he or she may instead choose to attend college. Therefore, we would expect defiers to exist in this empirical setting.  Moreover, it is unclear why this instrument is fully independent of the potential income, since the choice of residence may depend on parents' potential income, which may be correlated with their children's income. 

I conditioned $(Y_i,D_i,Z_i)$ on three characteristics: living in the south (S/NS), living in a metropolitan area (M/NM), and ethnic group (B/NB). I follow \citet{mourifie2017testing} in excluding subgroup NS/NM/B due to the small sample size, and also exclude subgroup NS/M/B due to the high frequency of $Z=1$. I conduct estimation and inference on each of the remaining 6 subgroups and the pooled sample. The choices of trimming sequence $b_n$, kernel bandwidth $h$, upper and lower band $M_u,M_l$, tail set $\mathcal{Y}_d^{ut},\mathcal{Y}_d^{lt}$ are specified in Appendix \ref{section: Details of Implementation of Empirical Illustration}. 

Estimation results are reported in Table \ref{table: estimation and inference for LATE}. I also report the LATE estimates when we directly use the IA-M assumption and Wald statistics. The estimated measure of compliers under $\tilde{A}$ satisfying \ref{assumption:type ind instru} conditioned on $Z_i=1$ and $Z_i=0$ are reported as $P(\mathcal{Y}_1,1)-Q(\mathcal{Y}_1,1)$ and $Q(\mathcal{Y}_0,0)-P(\mathcal{Y}_0,0)$, while the estimated measure of compliers under IA-M assumption is $E[D_i|Z_i=1]-E[D_i|Z_i=0]$.  The estimates of ${LATE}^{ID}_{\tilde{A}}$ and $LATE^{Wald}$ differ the most for three groups: S/NM/NB, S/M/NB and S/M/B. It should be noted that for all these three groups, estimated  $E[D_i|Z_i=1]-E[D_i|Z_i=0]$ differs from $P(\mathcal{Y}_1,1)-Q(\mathcal{Y}_1,1)$ and $Q(\mathcal{Y}_0,0)-P(\mathcal{Y}_0,0)$. If we blindly use the identification result under the IA-M assumption and use the standard LATE Wald estimator, the `identified' local average treatment effect can be negative (subgroups S/NM/NB and S/M/NB), or be unrealistically large (subgroup S/M/B). Once we use a strong extension, the estimated LATE for each of the 6 subgroups is positive, and the value of LATE is all between zero and one. We fail to reject education decrease future earning for the complier group ($LATE<0$) for all 6 subgroups for the $LATE^{wald}$. On the other hand, my method can reject the hypothesis $LATE<0$ for the NS/NM/NB and the S/NM/B group at 95\% confidence level. When I look at the African-American only, while $LATE^{wald}$ is large, the hypothesis fail to reject that education is harmful to their earning, while my method will reject the hypothesis $LATE^{wald}$ for the  African-American is negative.

{
	\begin{table}[H]
		\renewcommand{\arraystretch}{1.5}
		\footnotesize
		
		\caption{Estimation Result under Extensions  Assumption \ref{assumption:type ind instru}}
		\scalebox{0.9}{
			\begin{tabular}{lccccccc}\label{table: estimation and inference for LATE}
				Group                                                                    & NS,NM,NB         & NS,M,NB              & S,NM,NB            & S,NM,B             & S,M,NB               & S,M,B                & B-Group Only              \\
				$Pr(Z_i=1)$                                                             & 0.464            & 0.879                & 0.349              & 0.322              & 0.608                & 0.802                & 0.6188           \\
				Observations                                                             & 429              & 1191                 & 307                & 314                & 380                  & 246                  & 703            \\ \hline
				$LATE^{ID}_{\tilde{A}}(F)$                                                       & 0.5599           & 0.1546               & 0.2524             & 0.4773             & 0.5276               & 0.4358               & 1.0993           \\
				CI for $LATE^{ID}_{\tilde{A}}(F)$                                                & {[}0.01, 1.11{]} & {[}-0.51, 0.82{]} & {[}-1.22,	1.73{]} & {[}-0.09,1.04{]} & {[}-2.54,	3.59{]}     & {[}-5.15,	6.02{]} & {[}0.58, 1.62{]}  \\
				$P(\mathcal{Y}_1,1)-Q(\mathcal{Y}_1,1)$                                  & 0.1120           & 0.1084               & 0.0265             & 0.0739             & 0.0164               & 0.0338               & 0.0375           \\
				$Q(\mathcal{Y}_0,0)-P(\mathcal{Y}_0,0)$                                  & 0.1148           & 0.0960               & 0.0684             & 0.1495             & 0.0922               & 0.0308               & 0.1583           \\ \hline
				$LATE^{wald}(F)$                                                            & 0.5976           & 0.0761               & -6.4251            & 1.1873             & -1.5412              & 17.9620              & 5.0499           \\
				CI for ${LATE}^{wald}(F)$                                                   & {[}-0.20	1.39{]} & {[}-1.24,	1.39{]} & {[}-105,92{]}     & {[}-0.53,2.90{]} & {[}-5.09,2.01{]} & {[}-1.7e4,1.7e4{]}   & {[}-6.16,16.26{]} \\
				\begin{tabular}[c]{@{}l@{}}$E[D_i|Z_i=1]$\\ -$E[D_i|Z_i=0]$\end{tabular} & 0.1080           & 0.1084               & -0.0070            & 0.0692             & -0.0697              & 0.0002               & 0.0317          
		\end{tabular}}

\end{table}}

\section{Application to Binary Outcome Sector Choice}
In this section, I apply the method to a binary outcome sector choice model with a binary instrument \citep{mourifie2018roy}. This model is incomplete. Observed outcome $Y_i\in \{0,1\}$ is binary and the observed job sector choice  $D_i\in\{0,1\}$ is binary. A binary instrument $Z_i$ is observed. The $\epsilon$ variables include $Y_i(1)$ and $Y_i(0)$, which are the potential outcome in job sector 1 and 0 respectively, and the instrument variable $Z_i$. Observed sector outcome $Y_i$ is generated through:
\begin{equation}\label{eq: roy model potential outcome}
Y_i=Y_i(1)D_i+Y_i(0)(1-D_i).
\end{equation}
Without imposing further assumptions, equation (\ref{eq: roy model potential outcome}) does not specify how sector choice $D_i$ is determined, and $D_i$ can take
values in one of the three sets $\{0\},\{1\},\{0,1\}$. If we impose the classical Roy sector selection rule, for all  $Z_i=z\in\{0,1\}$ we have:
\begin{equation}\label{eq: roy model, perfect selection D_i mapping}
D_i\in\begin{cases}
\{1\} \quad &if \quad Y_i(1)>Y_i(0),\\
\{0\} \quad &if  \quad Y_i(0)<Y_i(1),\\
\{0,1\}\quad &if \quad Y_i(1)=Y_i(0).
\end{cases}
\end{equation}
The Roy sector selection rule (\ref{eq: roy model, perfect selection D_i mapping}) is just a special case of how $D_i$ is determined. To specify the structure universe $\mathcal{S}$, we consider all possible sector selection rules. \footnote{For each $Y_i(0)=y_0,Y_i(1)=y_1,Z_i=z$, $D_i$ can take values in three sets $\{0\},\{1\},\{0,1\}$. Therefore, there are $3^8$ ways to specify the sector selection rule.}
\begin{definition}\label{def: definition of M^s for job sector model}
	A sector selection rule is a set-valued function
	\[
	\begin{split}
	D^{sel}: \{0,1\}^3&\rightarrow  \{\{0\},\{1\},\{0,1\}\},\\
	(y_1,y_0,z)&\rightarrow D^{sel}(y_1,y_0,z).
	\end{split} \] Let $\mathcal{D}^{sel}$  be the collection of all possible sector selection rules. Let $D^{s,sel}\in \mathcal{D}^{sel}$ be a sector selection rule. Given a distribution $G^s(\epsilon)$ and a $D^{s,sel}$, the associated correspondence $M^s$ is defined as:
	{\footnotesize
	\begin{equation}\label{eq: Roy appli, M^s}
	\begin{split}
		M^s(G^s)\equiv \bigg\{ F\in \mathcal{F}: &Pr_F(Y_i=y,D_i=d,Z_i=z)= 
		(C^{y0z}_1+C^{y1z}_1)\mathbbm{1}(d=1)+(C^{0yz}_0+C^{1yz}_0)\mathbbm{1}(d=0),\\
		&\text{holds for some vector}  \left(C^{ykz}_{d}\right)_{y,k,z,d\in\{0,1\}} \text{ such that }\\
		&C^{ykz}_1+C^{ykz}_0=Pr_{G^s}(Y_i(1)=y,Y_i(0)=k,Z_i=z),\\
		&C^{ykz}_d = 0 \quad if \quad D^{s,sel}(y,k,z)=\{1-d\}\quad and \quad 
		C^{ykz}_d\ge 0	\quad \quad \forall y,k,z,d\in\{0,1\} 
		\bigg\}.
	\end{split}
	\end{equation}}
\end{definition}
In Definition \ref{def: definition of M^s for job sector model}, $C^{ykz}_d$ is the probability of choosing sector $d$ when $Y_i(1)=y,Y_i(0)=k,Z_i=z$. When $D^{s,sel}(y,k,z)$ is the set $\{0,1\}$, a structure $s$ associated with $D^{s,sel}$ does not specify how a sector choice is determined, so the only constraint is $C^{ykz}_1+C^{ykz}_0=Pr_{G^s}(Y_i(1)=y,Y_i(1)=k,Z_i=z)$. When $D^{s,sel}(y,k,z)=\{d\}$, the constraint $C^{ykz}_{1-d}=0$ in the last row of (\ref{eq: Roy appli, M^s})  requires the probability of choosing sector $1-d$ is zero. Given the set of sector selection rules $\mathcal{D}^{sel}$, we can specify the structure universe $\mathcal{S}$ in this application as follows:
\begin{equation}\label{eq: roy model structure universe}
\begin{split}
\mathcal{S}=\bigg\{s=(M^s,G^s)\bigg| &G^s \,\, \text{is  a distribution of } (Y_i(1),Y_i(0),Z_i),\quad \\
&M^s \,\,\text{is associated with a } D^{sel}\in \mathcal{D}^{s,sel}\bigg\}.
\end{split}
\end{equation}
Instead of imposing the strong independent instrument condition $\left(Y_i(1),Y_i(0)\right)\perp Z_i$, we require the instrument to have monotone effects on the potential outcomes:

\begin{definition} \label{def: roy model, monotone instrument at best and worst outcome}
	We say $(Y_i(1),Y_i(0))|Z_i=1$ dominates $(Y_i(1),Y_i(0))|Z_i=0$ at the best and worst outcomes if 
	\begin{equation}\label{eq: roy model, dominating instrument at best and worst outcome}
	\begin{split}
	Pr(Y_i(1)=Y_i(0)=1|Z_i=1)&\ge Pr(Y_i(1)=Y_i(0)=1|Z_i=0),\\ 	Pr(Y_i(1)=Y_i(0)=0|Z_i=1)&\le Pr(Y_i(1)=Y_i(0)=0|Z_i=0).
	\end{split}
	\end{equation}
\end{definition}
Definition \ref{def: roy model, monotone instrument at best and worst outcome} only requires the instrument to generate the best (resp. worst) potential outcome $Y_i(1)=Y_i(0)=1$ (resp. $Y_i(1)=Y_i(0)=0$) with higher (resp. lower) probability at $Z_i=1$ than that at $Z_i=0$.   This requirement is weaker than Assumption 5 in \citep{mourifie2018roy}, where they  require $Pr(Y_i(d)=1|Z_i=1)\ge Pr(Y_i(d)=1|Z_i=0)$ for $d\in \{0,1\}$ in addition to  (\ref{eq: roy model, dominating instrument at best and worst outcome}). \footnote{ Condition (\ref{eq: roy model, dominating instrument at best and worst outcome}) along with the additional requirement $Pr(Y_i(1)=1|Z_i=1)\ge Pr(Y_i(1)=1|Z_i=0)$ will imply $1-Pr_F(Y_i=0,D_i=1|Z_i=1)\ge Pr_F(Y_i=1,D_i=1|Z_i=0)$ for all outcome distributions $F$. This implication holds even in the absence of the Roy selection condition. On the other hand, equation (\ref{eq: roy model, dominating instrument at best and worst outcome}) alone does not imply any constraints on $F$.}

Dominance at the best and worst outcome in Definition \ref{def: roy model, monotone instrument at best and worst outcome} can accommodate broader empirical scenarios compared with Assumption 5 in \citep{mourifie2018roy}. For example, suppose $Y_i(d)=1$ means  individual $i$ gets tenure in sector $d$, and $Z_i=1$ means individual $i$ participates in a job training program. If the skill obtained from the training program can be applied to both sectors, we would expect $Pr(Y_i(1)=Y_i(0)=1|Z_i=1)\ge Pr(Y_i(1)=Y_i(0)=1|Z_i=0)$. On the other hand, each job sector may require specific skill that cannot be obtained from the job training program. If the training program is costly and prevents an individual $i$ from developing skills specific to the sector $d$, it is possible that $Pr(Y_i(d)=1|Z_i=1)\le Pr(Y_i(d)=1|Z_i=0)$. Such a scenario  violates Assumption 5 in \citet{mourifie2018roy} but not (\ref{eq: roy model, dominating instrument at best and worst outcome}). Lastly, $Pr(Y_i(1)=Y_i(0)=0|Z_i=1)\le Pr(Y_i(1)=Y_i(0)=0|Z_i=0)$ ensures the training program is beneficial: it increases the probability of success in at least one sector.

However, when (\ref{eq: roy model, dominating instrument at best and worst outcome}) is combined with Roy's selection assumption, they are jointly refutable. Formally, the assumption with Roy's selection rule and a monotone instrument satisfying (\ref{eq: roy model, dominating instrument at best and worst outcome}) is defined in the following.
\begin{assumption}
	The assumption of Roy's selection rule with instrument condition (\ref{eq: roy model, dominating instrument at best and worst outcome}), denoted as  $A^{Roy}$, is the collection of structures such that 
	\begin{equation}\label{eq: Roy Assumption}
	\begin{split}
	A^{Roy}=\bigg\{s\in \mathcal{S}:\,\,& M^s\,\,\text{is associated with the }\,\,D^{sel}\,\,\text{in}\,\,(\ref{eq: roy model, perfect selection D_i mapping}),\\
	&M^s(G^s)\,\,\text{satisfies} (\ref{eq: Roy appli, M^s}),\quad 
	G^s\,\, \text{satisfies} (\ref{eq: roy model, dominating instrument at best and worst outcome}) 
	\bigg\}.
	\end{split}	
	\end{equation}
\end{assumption}
The following proposition characterizes the non-refutability and confirmation sets associated with $A^{Roy}$.

\begin{prop}\label{prop: roy model, non-refutable and confirmation set}
	Let $\mathcal{F}^{nf}$ be the collection of outcome distributions such that 
	\begin{equation}\label{eq: roy model, testable implication}
	\begin{split}
	Pr_F(Y_i=0|Z_i=1)\le Pr_F(Y_i=0|Z_i=0).
	\end{split}
	\end{equation}
	The  non-refutability and confirmation sets of $A^{Roy}$ are
	\[
	\begin{split}
	\mathcal{H}_{\mathcal{S}}^{snf}(A^{Roy})&= \{s:\, M^s(G^s)\subseteq \mathcal{F}^{nf}\},\\
	\mathcal{H}_{\mathcal{S}}^{wnf}(A^{Roy})&= \{s:\, M^s(G^s)\cap  \mathcal{F}^{nf}\ne \varnothing \},\\
	\mathcal{H}_{\mathcal{S}}^{scon}(A^{Roy})&=\mathcal{H}_{\mathcal{S}}^{wcon}(A^{Roy})=\varnothing .
	\end{split}
	\]
\end{prop}
Proposition \ref{prop: roy model, non-refutable and confirmation set} reveals several things: first $\mathcal{H}_{\mathcal{S}}^{snf}(A^{Roy})\ne \mathcal{S}$, so the assumption $A^{Roy}$ is refutable; second, $\mathcal{H}_{\mathcal{S}}^{wnf}(A^{Roy})\ne \mathcal{H}_{\mathcal{S}}^{snf}(A^{Roy})$ so we cannot find a strong extension of $A^{Roy}$.  Therefore, I look at the completion of the binary sector choice model. The completion given in Definition \ref{def: completion of an incomplete structure universe} is abstract and in the following I give an explicit form of $\mathcal{C}(s)$.  

Recall that each $s$ in the incomplete space is associated with a sector decision rule, denoted by $D^{s,sel}$. We consider a tie breaking rule associated with $s$: $\left(C^{s^*,tb}_{d}(y,k,z)\right)_{y,k,z,d\in\{0,1\}}$, which specifies the probability of choosing sector $d$ for different values of $Y_i(0)=y,Y_i(1)=k,Z_i=z$.  Different values of the vector $\left(C^{s^*,tb}_{d}(y,k,z)\right)_{y,k,z,d\in\{0,1\}}$ correspond to different selections. The completion $\mathcal{C}(s)$ is the collection of $M^{s^*}(\cdot)$ such that
{
	\begin{equation}\label{eq: roy model, completed structure mapping}
	\begin{split}
	M^{s^*}(G^{s})=\bigg\{ &F\in \mathcal{F}\bigg| \forall y,k,z,d\in\{0,1\}: C^{ykz}_d=C^{s^*,tb}_d(y,k,z)\times Pr_{G^s}(Y_i(1)=y,Y_i(0)=k,Z_i=z),\\
	&Pr_F(Y_i=y,D_i=d,Z_i=z)= 
	(C^{y0z}_d+C^{y1z}_d)\mathbbm{1}(d=1)+(C^{0yz}_d+C^{1yz}_d)\mathbbm{1}(d=0) \\
	&\text{holds for a vector }  \left(C^{s^*,tb}_{d}(y,k,z)\right)_{y,k,z,d\in\{0,1\}} \text{ such that }\\
	&\sum_{d=0}^1 C^{s^*,tb}_{d}(y,k,z)\equiv 1,\quad C^{s^*,tb}_{d}(y,k,z)= 1 \text{ if } D^{s,sel}(y,k,z)=\{d\}
	\bigg\}.
	\end{split}
	\end{equation}}
Compared with the incomplete mapping in (\ref{eq: Roy appli, M^s}), the tie breaking rule  $C_{d}^{s^*,tb}$ makes $M^{s^*}(G^{s})$ a singleton. The completed universe is given in Definition \ref{def: completion of an incomplete structure universe} and the corresponding Roy assumption set in the completed structural universe is given in Definition Proposition \ref{prop: incomplete theory, preserve the id set after completion}. 

\subsubsection*{Minimal Efficiency Loss and the Corresponding Strong Extension}
I now consider a strong extension under the completed structure universe $\mathcal{S}^*$. I first define the efficiency loss of a structure $s^*$, which can be viewed as a deviation from Roy's sector selection assumption.

\begin{definition}\label{def: minimal efficiency loss}
	The efficiency loss of a structure $s^*\in\mathcal{S}^*$
	\begin{equation} \label{eq: Roy model, efficiency loss}
	m^{EL}(s^*)=E_{G^{s^*}}[\max\{Y_i(1),Y_i(0)\}]-E_{F}[Y_i] \quad \quad for \quad F\in M^{s^*}(G^{s^*})
	\end{equation} 
	is the difference between the expected optimal sector selection outcome and the expected predicted outcome.	
\end{definition}
The efficiency loss is a function since  $M^{s^*}(G^{s^*})$ is a singleton. It is easy to see that when the Roy sector selection rule holds, $m^{EL}(s^*)=0$. Conversely, by (\ref{eq: roy model potential outcome}),  $m^{EL}(s^*)=0$ implies 
\[
D_i=\begin{cases}
1 \quad \quad \text{if} \quad Y_i(1)>Y_i(0)\\
0\quad \quad \text{if} \quad Y_i(1)<Y_i(0)\\
\end{cases}
\]
with probability 1, so the Roy sector selection rule holds. Once we verify $m^{EL}(s^*)$ is a well-behaved minimal deviation measure (see Definition \ref{def: well-defined relaxation measure}), we can use minimal efficiency loss to construct a strong extension. 
\begin{assumption} \label{assump: minimal efficiency loss Roy} Let
	$m^{min}(F)\equiv \inf\{m^{EL}(s^*): F\in M^{s^*}(G^{s^*}) \,\, and \,\, G^{s^*}\text{ satisfies }(\ref{eq: roy model, dominating instrument at best and worst outcome}) \}$ be the minimal efficiency loss under $F$. We call 
	\[\tilde{A}=\cup_{F\in \mathcal{F}} \left\{s^*: m^{EL}(s^*)=m^{min}(F),\,G^{s^*}\text{ satisfies }(\ref{eq: roy model, dominating instrument at best and worst outcome})\right\}\]
	the minimal efficiency loss extension of $A^{Roy}$.
\end{assumption}

\begin{prop}\label{prop: roy model, sharp characterization of identified set}
	$\tilde{A}^{*Roy}$ is a strong extension of $A^{*Roy}$. Moreover, given an observed distribution $F$, the identified set for $G^{s^*}$ under $\tilde{A}^{*Roy}$ is { \footnotesize
	\begin{equation} \label{eq: roy model, sharp indeitified set of G^s}
	\begin{split}
	\Bigg\{ G^{s^*}\bigg| & \text{ there exists } \{C_d^{ykz}\} \quad \forall {y,k,z,d\in\{0,1\}}\,\,s.t.\,\, C_d^{ykz}\ge 0,\\ 
	&Pr_F(Y_i=y,D_i=d,Z_i=z)= 
	(C^{y0z}_d+C^{y1z}_d)\mathbbm{1}(d=1)+(C^{0yz}_d+C^{1yz}_d)\mathbbm{1}(d=0),\\
	 &C^{ykz}_1+C^{ykz}_0=Pr_{G^{s*}}(Y_i(1)=y,Y_i(0)=k,Z_i=z),\\
	 & C^{010}_1=C^{100}_0=0,\quad\frac{C^{110}_1+C^{110}_0}{Pr_F(Z_i=0)}\le \frac{C^{111}_1+C^{111}_0}{Pr_F(Z_i=1)},\\
	 &C^{101}_0+C^{011}_1=\max\left\{Pr_F(Y_i=0,Z_i=1)-\frac{Pr_F(Y_i=0,Z_i=0)Pr_F(Z_i=1)}{Pr_F(Z_i=0)},0\right\}\Bigg\}.
	\end{split}
	\end{equation}}  
\end{prop}


Proposition \ref{prop: roy model, sharp characterization of identified set} characterizes the sharp identified set of distributions of $(Y_i(1),Y_i(0),Z_i)$.
The identified set (\ref{eq: roy model, sharp indeitified set of G^s}) under $\tilde{A}^{*Roy}$ satisfies: (1). There is no efficiency loss when $Z_i=0$ ($C^{010}_1=C^{100}_0=0$); (2). The minimal efficiency loss is $\max\left\{Pr_F(Y_i=0,Z_i=1)-\frac{Pr_F(Y_i=0,Z_i=0)Pr_F(Z_i=1)}{Pr_F(Z_i=0)},0\right\}$; (3). Condition (\ref{eq: roy model,  dominating instrument at best and worst outcome}) holds as long as $(C^{110}_1+C^{110}_0){Pr_F(Z_i=1)}\le (C^{111}_1+C^{111}_0){Pr_F(Z_i=0)}$ holds. The identified set of $G^{s^*}$ is a polyhedron characterized by the 16-dimensional vector $(C_d^{ykz})_{y,k,z,d\in\{0,1\}}$.  Many parameters of interest are linear functions of $(C^{jkz}_d)$, and linear-programming can be used to find the identified set.  One example is given in Corollary \ref{Corollary: Roy model, ID set of $Pr(Y_i(1)=1|Z_i=z)$}.
\begin{Corollary}\label{Corollary: Roy model, ID set of $Pr(Y_i(1)=1|Z_i=z)$}
	The identified set for $Pr(Y_i(1)=1|Z_i=z)$ under $\tilde{A}^{*Roy}$ is 
	\[
	\begin{split}
	Pr(Y_i=1,D_i=1|Z_i=0)\le &Pr(Y_i(1)=1|Z_i=0)\le Pr(Y_i=1|Z_i=0),\\
	Pr(Y_i=1,D_i=1|Z_i=1)\le &Pr(Y_i(1)=1|Z_i=1)\le Pr(Y_i=1|Z_i=1)+ \frac{m^{EL,min}(F)}{Pr(Z_i=1)},\\
	\end{split}
	\]
	where $ m^{EL,min}(F)=\max\left\{Pr_F(Y_i=0,Z_i=1)-\frac{Pr_F(Y_i=0,Z_i=0)Pr_F(Z_i=1)}{Pr_F(Z_i=0)},0\right\}$.
\end{Corollary} 

It is worth noticing that if $A^{*Roy}$ cannot be rejected by $F$, the identified set for $\theta\equiv Pr(Y_i(1)=1|Z_i=1)$ is given by $[Pr(Y_i=1,D_i=1|Z_i=1),Pr(Y_i=1|Z_i=1)]$. However, suppose we ignore the testable implication of $A^{*Roy}$ and directly use $[Pr(Y_i=1,D_i=1|Z_i=1),Pr(Y_i=1|Z_i=1)]$ as the identified set for $\theta$, we get a spuriously informative identified set  when $m^{EL,min}(F)>0$. This happens when the true structure that generates the data lies in $\tilde{A}^{*Roy}\backslash {A}^{*Roy}$, and the spuriously informative identified set is a proper subset of the true identified set.

The upper bound of the identified set for $Pr(Y_i(1)=1|Z_i=1)$ is not Fr\'echet differentiable with respect to $F$ due to the $\max$ operator. However, by Example 2 in \citet{fang2019inference}, the upper bound $Pr(Y_i=1|Z_i=1)+ \frac{m^{EL,min}(F)}{Pr(Z_i=1)}$ is directionally differentiable in $F$. The bootstrap method in \citet{fang2019inference} can be used to construct confidence interval for $Pr(Y_i(1)=1|Z_i=1)$. However, since $Pr(Y_i(1)=1|Z_i=1)$ is partially identified and satisfies the conditions in Proposition \ref{prop: point identification implies strong binary decidable}, we cannot directly test a hypothesis on the value of $Pr(Y_i(1)=1|Z_i=1)$. Instead, we should test the equivalent existence hypothesis as in Proposition \ref{prop: equivalent decision}.

	\bibliography{testable_implication_reference}
	\bibliographystyle{chicago}
	\appendix
\section{Discussion of Property \ref{property: continuity}}\label{section: Discussion of Continuous ID Set Correspondence}
	When the structure universe $\mathcal{S}$ is complete, the following proposition provides sufficient high level conditions to check whether $\Theta_{\tilde{A}}^{ID}$ is a continuous correspondence in $F$. 
	
	\begin{prop}\label{prop: continuity of the identified set}
		Let $(\mathcal{F},d_F)$ and $(\Theta,d_\theta)$ be metric spaces, and let $\mathcal{S}$ be a complete structure universe. Let $\tau_{\mathcal{F}}$ be the topology on $\mathcal{F}$ induced by $d_F$. Let $A=A_j\cap(\cap_{l\ne j}A_l)$, and let $j$ be the index, $m_j$ be the well-defined relaxation measure in Definition \ref{def: minimal deviation extension}. We equip the $\cap_{l\ne j}A_l$ space with the weak topology $\tau_{A_{-j}}$ induced by the mapping $h(s)=M^s(G^s)$ \footnote{This is a slight abuse of the notation since $h(s)$ is a single-valued correspondence and its image space is $2^{\mathcal{F}}$. I abuse the notation and use $h(s)$ to denote the $M^s(G^s)$ mapping composited with the unique selection from the image set $M^s(G^s)$.  }:
		\[
		\tau_{A_{-j}}\equiv\{O\subset \cap_{l\ne j}A_l:\,O=h^{-1}(P)\cap (\cap_{l\ne j}A_l)\,\,for\,\,some\,\,P\in \tau_{\mathcal{F}}\}.
		\]
		Suppose: (1) $\theta(s):\cap_{l\ne j}A_l\rightarrow \Theta$ is a continuous function; and (2) the function 
		\[
		\min(F;m_j) \equiv \inf\{m_j(s): F\in M^s(G^s) \quad and \quad s\in \cap_{l\ne j}A_l\}
		\]  
		is a continuous mapping from  $\mathcal{F}$ to $\mathbb{R}$. If $m_j^{-1}:\mathbb{R}\rightrightarrows \cap_{l\ne j}A_l$ is an upper (resp. lower) hemicontinuous correspondence, then $\Theta^{ID}_{\tilde{A}}(F)$ is an upper (resp. lower) hemicontinuous correspondence from $(\mathcal{F},d_F)$ to $(\Theta,d_\theta)$.
	\end{prop}
\begin{proof}
	See Online Appendix \ref{Proof of continuous id set condition}.
\end{proof}
	Here is a reasoning behind Property \ref{property: continuity}: We may want a continuous relation between the structure universe $\mathcal{S}$ and the observation space $\mathcal{F}$. When the true structure $s$ change a little, the predicted observation distribution should not change drastically. Similarly, the parameter of interest $\theta$ should also be continuous with respect to change in the true structure. The relation can be represented as $\mathcal{F} \xleftarrow{M^s(G^s)} \mathcal{S} \xrightarrow{\,\,\theta(s)\,\,} \Theta$.
	Unfortunately, there may not exist a natural topology embedded in $\mathcal{S}$. The weak topology defined in Proposition \ref{prop: continuity of the identified set} is the smallest topology such that the mapping $M^s(G^s)$ is continuous. The construction of $\tau_{A_{-j}}$ in Proposition \ref{prop: continuity of the identified set} uses the inverse of $M^s(G^s)$ to induce a topology on the structure universe $\mathcal{S}$. With this construction, the relations become: $ \mathcal{F} \xrightarrow{(M^s(G^s))^{-1}} \mathcal{S} \xrightarrow{\,\,\theta(s)\,\,} \Theta$. The identified set can then be viewed as the composite mapping of $(M^s(G^s))^{-1}$ and $\theta$, defined on the extended assumption $\tilde{A}\subset \mathcal{S}$. If $\theta(s)$ is continuous, the composite map should also be continuous. Therefore, Property \ref{property: continuity} can be viewed as a consequence of the continuity of $\mathcal{F}\rightarrow\mathcal{S} $ and the continuity of $\mathcal{S}\rightarrow \Theta$.

\section{Two LATE-Consistent Extensions} \label{append: minimal marginal difference}
\subsection{Minimal Distance to Marginal Independence as LATE-consistent Extension}

Testable implications (\ref{eq: testable implication of LATE}) also arise from the independent instrument assumption. In this section, I provide a relaxed assumption  that relax the independent instrument assumption while keeping the `No Defiers' assumption. However, independence of an instrument on the potential outcomes is an infinite-dimensional constraint. As a result, there are infinitely many ways to relax it and will result in different identified sets when the IA-M Assumption is rejected. 

I keep the exclusion restriction $Y_i(d,z)=Y_i(d,1-z)$, and only consider the marginal distribution of $Y_i(d,z)$. By the $M^s$ mapping defined in (\ref{eq: appli, M^s}), the probability measure  $Pr_{G^s}(Y_i(d_1,1)\in B_{d_11},D_i(1)=d_1,D_i(0)=d_0|Z_i=1)$ and $Pr_{G^s}(Y_i(d_0,0)\in B_{d_11},D_i(1)=d_1,D_i(0)=d_0|Z_i=0)$ are absolutely continuous with respect to $\mu_F$, and denote  for $d_1,d_0,z\in\{0,1\}$
\[
\begin{split}
g^s_{y_{d_1 1}}(y,d_1,d_0|Z_i=1)&= \frac{d Pr_{G^s}(Y_i(d_1,1)\in B_{d_11},D_i(1)=d_1,D_i(0)=d_0|Z_i=1) }{d\mu_F},\\
g^s_{y_{d_0 0}}(y,d_1,d_0|Z_i=0)&= \frac{d Pr_{G^s}(Y_i(d_0,0)\in B_{d_00},D_i(1)=d_1,D_i(0)=d_0|Z_i=0) }{d\mu_F}
\end{split}
\]
as the Radon-Nikodym derivatives with respect to $\mu_F$. Throughout this section, I use $g^s_{y_{dz}}(y,d,d'|Z_i)$ to denote the density of $G^s(Y_i(d,z)\in B_{dz},D_i(1)=d,D_i(0)=0|Z_i)$.

We consider the following deviation measure:
{\footnotesize
	\[
	m^{MI}(s)= \begin{cases}
	\begin{split}
	&\sum_{d=0}^1\int_{\mathcal{Y}} \left[g^s_{y_{d1}}(y,d,d|Z_i=1)-g^s_{y_{d0}}(y,d,d|Z_i=0)\right]^2 d\mu_F(y)
	\end{split} \quad &if\quad Y_i(d,z)=Y_i(d,1-z)\,\, \,\,G^s-a.s,\\
	+\infty \quad &otherwise.
	\end{cases}
	\]}
The  $m^{MI}$ measures the difference of marginal distributions of the potential outcomes when the instrument $Z_i$ takes different values. When $d=1$, $g^s_{y_{11}}(y,1,1|Z_i=1)$ is the marginal density of $Y_i(1,1)=Y_i(1,0)$ and $D_i(1)=D_i(0)=1$ conditional on $Z_i=1$, and $g^s_{y_{10}}(y,1,1|Z_i=0)$ is the same object but conditional on $Z_i=0$. Note that when the instrument $Z_i$ is independent of the potential outcomes, for $d\in\{0,1\}$ and almost all $y$,
$g^s_{y_{d1}}(y,d,d|Z_i=1)-g^s_{y_{d0}}(y,d,d|Z_i=0)=0$ holds. Therefore,
 $m^{MI}(s)=0$ whenever $s$ satisfies the independent instrument assumption.

\begin{assumption} \label{assump: minimal dist to marg ind inst}
	(Minimal Distance to Marginal Independent Instrument) Let $m^{min}(F)\equiv \inf\{m^{MI}(s): F\in M^s(G^s) \,\, and \,\, s\in A^{ER}\cap A^{TI-CP}\cap A^{ND}\}$ be the minimal distance. We call 
	\[\tilde{A}=\cup_{F\in \mathcal{F}} \left\{s\in A^{ER}\cap A^{TI-CP}\cap A^{ND}: m^{MI}(s)=m^{min}(F),\,\,F\in M^s(G^s)\right\}\]
	the minimal marginal dependence extension.
\end{assumption}
We do not give up the independent instrument assumption completely: We still keep type independent instrument assumption for the compliers ($A^{TI-CP}$). With $A^{TI-CP}$, we can show this relaxation is LATE-consistent.
\begin{prop} \label{prop: minimal marg ind as consistent extension}
	The $\tilde{A}$ defined in Assumption \ref{assump: minimal dist to marg ind inst} is a LATE-consistent extension of $A$. 
\end{prop}

The $\tilde{A}$ defined in Assumption \ref{assump: minimal dist to marg ind inst} is LATE-consistent but not a strong extension. This is because when $m^{MI}(s)=0$, we cannot say $Z_i$ is an independent instrument under $s$.\footnote{In particular, consider the indirect effect of instrument on ATE:
	$
	\tilde{\theta}(s)=E[Y_i(1)-Y_i(0)|Z_i=1]-E[Y_i(1)-Y_i(0)|Z_i=0].
	$
	Whenever the IA-M assumption $A$ is not rejected by $F$, the identified set for 
	is $\tilde{\Theta}_{A}^{ID}(F)=\{0\}$. However, if we use the extension $\tilde{A}$ in Assumption \ref{assump: minimal dist to marg ind inst}, the identified set is not a singleton under $F$. 
}

\subsection{Minimal Marginal Difference Extension}
This section considers an extension that relaxes the exclusion restriction. The exclusion restriction fails when the instrument $z_i$ has a direct effect on potential outcomes. Like the independent instrument assumption, the exclusion restriction is a distributional assumption and there are infinitely many ways to relax it.  Let 
\[
\begin{split}
g_{y_{d_1 1}}^s(y,d_1,d_0|Z_i=1)= \frac{d Pr_{G^s}(Y_i(d_1,1),D_i(1)=d_1,D_i(0)=d_0|Z_i=1)}{d \mu_F}\\
g_{y_{d_0 0}}^s(y,d_1,d_0|Z_i=0)= \frac{d Pr_{G^s}(Y_i(d_0,0),D_i(1)=d_1,D_i(0)=d_0|Z_i=0)}{d \mu_F}
\end{split}
\]
be the Radon-Nikodym derivatives of marginal distributions of $G$ with respect to $\mu_F$. We consider the following deviation measure:
\[
\begin{split}
m^{MD}(s)&= \int_y [g^s_{y_{10}}(y,1,1|Z_i=0)-g^s_{y_{11}}(y,1,1|Z_i=1)]^2d\mu_F(y)\\
& +\int_y [g^s_{y_{00}}(y,0,0|Z_i=0)-g^s_{y_{01}}(y,0,0|Z_i=1)]^2d\mu_F(y).
\end{split}
\]
The quantity $m^{MD}(s)$ measures the marginal distributions difference for  potential outcomes. Under exclusion restriction and the independent instrument assumptions,  $Y_i(1,0)=Y_i(1,1)$ holds almost surely and $m^{MD}$ equals zero. The converse is not true: when the independent instrument assumption holds, $m^{MD}(s)=0$ does not imply exclusion restriction.\footnote{This is because in the construction of $m^{MD}$, we ignore the compliers and defiers.}
\begin{assumption} \label{assump: minimal dist to zero marg diff }
	(Minimal Marginal Difference) Let $A^{ER-CP}=\{s: Y_{i}(d,1)= Y_{i}(d,0)|D_{i}(1)-D_i(0)=1\,a.s.\}$ ge the exclusion restriction for the compliers. Let
	 $m^{min}(F)\equiv \inf\{m^{MD}(s): F\in M^s(G^s) \,\, and \,\, s\in A^{ER-CP}\cap A^{TI}\cap A^{ND}\}$ be the minimal distance. We call 
	\[\tilde{A}=\cup_{F\in \mathcal{F}} \left\{s\in A^{ER-CP}\cap A^{TI}\cap A^{ND}: m^{MD}(s)=m^{min}(F),\,\,F\in M^s(G^s)\right\}\]
	the minimal marginal difference extension.
\end{assumption}
Condition 3 in Assumption \ref{assump: minimal dist to zero marg diff }  is similar to the type independence for compliers condition in Assumption \ref{assump: minimal dist to marg ind inst}, under which we can generate informative constraint on LATE. We can show this extension is LATE-consistent.
\begin{prop} \label{prop: minimal marg diff as LATE-consistent}
	The $\tilde{A}$ defined in Assumption \ref{assump: minimal dist to zero marg diff } is a LATE-consistent extension of $A$. 
\end{prop}

We should note that the extension in Assumption \ref{assump: minimal dist to zero marg diff } also relaxes the independent instrument assumption, since I only require the instrument to be type independent. This is because, if we use the fully independent instrument $Z_i$, the measure $m^{MD}(s)$ is not a well-behaved relaxation measure. 

\begin{prop}\label{prop: m^{MD} is not well behaved if use full independence}
	(I)The measure $m^{MD}(s)$ is not a well-behaved measure with respect to $\{A^{ND},A^{FI}\}$, where 
	\[
	A^{FI}=\{s:\quad(Y_i(1,1),Y_i(0,1),Y_i(0,1),Y_i(0,0),D_i(1),D_i(0))\perp Z_i \}.
	\]
	(II)
	The measure $m^{MD}(s)$ is not a well-behaved measure with respect to $\{A^{ND},A^{TI},A^{EM-C}\}$, where $A^{EM-C}$ is the assumption that measure of compliers does not change with the value of $Z_i$
	\[
	A^{EM-C}=\{s:\quad Pr_{G^s}(D_i(1)=1,D_i(0)=0|Z_i=1)=Pr_{G^s}(D_i(1)=1,D_i(0)=0|Z_i=0)\}.
	\]
\end{prop}

\begin{remark}
	Recall by Definition \ref{def: well-defined relaxation measure}, a relaxation measure is well-defined if $\inf\{m_j(s): F\in M^s(G^s),\,\, s\in \cap_{l\ne j}A_l\}$  is finite and achievable by some $s\in \cap_{l\ne j}A_l$ for all $F$. In the case (I) of Proposition \ref{prop: m^{MD} is not well behaved if use full independence}, there exists an $F_0$ such that $\{m^{MD}(s): F\in M^s(G^s) ,\,\, s\in A^{ND}\cap A^{FI}\}=\varnothing$ and infimum over $\varnothing$ is $+\infty$; In the case (II) of  Proposition \ref{prop: m^{MD} is not well behaved if use full independence}, while $\inf\{m_j(s): F\in M^s(G^s) ,\,\, s\in A^{ND}\cap A^{TI}\cap A^{EM-C}\}$ exists, it is not achievable by any structure $s\in A^{ND}\cap A^{TI}\cap A^{EM-C}$.
\end{remark}

\section{Proofs in Section 2}\label{section: Proofs in section 2}

\subsection{Lemmas}
\begin{lem}\label{lem: charaterization of A set with scon}
	The following three conditions are equivalent:
	\begin{enumerate}
		\item $\mathcal{H}_\mathcal{S}^{wnf}(A)=\mathcal{H}_\mathcal{S}^{scon}(A)$;
		\item $A=\mathcal{H}_\mathcal{S}^{wnf}(A)$;
		\item $A=\mathcal{H}_\mathcal{S}^{scon}(A)$.
	\end{enumerate}
\end{lem}
\begin{proof}
	Recall that the following set relations hold:
	\begin{equation}\label{eq: append, set relation for nf and con}
	\mathcal{H}_\mathcal{S}^{scon}(A)\subseteq \mathcal{H}_\mathcal{S}^{wcon}(A)\subseteq A
	\subseteq \mathcal{H}_\mathcal{S}^{snf}(A)\subseteq \mathcal{H}_\mathcal{S}^{wnf}(A).
	\end{equation}
	
	$1\Rightarrow 2$ holds by the sandwich form (\ref{eq: append, set relation for nf and con}).
	
	To show $2\Rightarrow3$, it suffices to show that $A\subseteq \mathcal{H}_\mathcal{S}^{scon}(A)$, because (\ref{eq: append, set relation for nf and con}) holds. Suppose $A\nsubseteq \mathcal{H}_\mathcal{S}^{scon}(A)$, so there exists a $s\in A\backslash\mathcal{H}_\mathcal{S}^{scon}(A)$. Since $s\notin \mathcal{H}_\mathcal{S}^{scon}(A)$, by Definition \ref{def: confirmation sets}, there exists an $s^*\in A^c$ such that $M^s(G^s)\cap M^{s^*}(G^{s^*})\ne \varnothing$.
	Now, since $s\in A$, $M^s(G^s)\cap M^{s^*}(G^{s^*})\ne \varnothing$ implies that $s^*\in \mathcal{H}^{wnf}_{\mathcal{S}}(A)$. However, by condition 2 in this Lemma, $\mathcal{H}^{wnf}_{\mathcal{S}}(A)=A$, $s^*\in A$, so this yields the contradiction.
	
	To show $3\Rightarrow 1$,  it suffices to show that $\mathcal{H}_\mathcal{S}^{wnf}(A)\subseteq A$, because (\ref{eq: append, set relation for nf and con}) holds. Suppose $\mathcal{H}_\mathcal{S}^{wnf}(A)\nsubseteq A$, there exists an $s\in \mathcal{H}_\mathcal{S}^{wnf}(A)\backslash A$. By the definition of $\mathcal{H}_\mathcal{S}^{wnf}(A)$:
	\[M^s(G^s)\cap [\cup_{s^*\in A}M^{s^*}(G^{s^*})]\ne \varnothing.
	\]
	By condition 3 in this Lemma, $A=\mathcal{H}_{\mathcal{S}}^{scon}(A)$, we have
	\[M^s(G^s)\cap [\cup_{s^*\in \mathcal{H}_\mathcal{S}^{scon}(A)}M^{s^*}(G^{s^*})]\ne\varnothing.\] 
	So we can find an $s^*\in \mathcal{H}_\mathcal{S}^{scon}(A)$ such that  $M^{s^*}(G^{s^*})\cap M^s(G^s)\ne\varnothing$. However, by the definition of $\mathcal{H}_\mathcal{S}^{scon}(A)$, $M^{s^*}(G^{s^*})\cap \left( \cup_{s^\prime\in A^c} M^{s^\prime}(G^{s^\prime})\right)=\varnothing$. As a result,  $s\in A$ must hold. This contradicts $s\in \mathcal{H}_\mathcal{S}^{wnf}(A)\backslash A$.
\end{proof}

\subsection{Proof of Proposition \ref{prop: operation of confirmable and refutable set }}
\begin{proof}

	1. If $s\in \left[\mathcal{H}_\mathcal{S}^{snf}(A)\right]^c$, by Definition \ref{def: non-refutablity set in incomplete structure} $M^s(G^s)\cap \left[\cup_{s^*\in A} M^{s^*}(G^{s^*}) \right]^c\ne \varnothing $. Since $(A^c)^c=A$, we have 
	\begin{equation} \label{eq: append, set relation1}
	\left[\cup_{s^*\in A} M^{s^*}(G^{s^*}) \right]^c=\cap_{s^*\in (A^c)^c} [M^{s^*}(G^{s^*})]^c.
	\end{equation}
	Therefore, $s\in \mathcal{H}_\mathcal{S}^{wcon}(A^c)$ and $\mathcal{H}_\mathcal{S}^{snf}(A)^c\subseteq\mathcal{H}_\mathcal{S}^{wcon}(A^c)$.
	
	Similarly, if $s\in \mathcal{H}_\mathcal{S}^{wcon}(A^c)$, by Definition \ref{def: confirmation sets} $M^s(G^s)\cap \left(\cap_{s^*\in (A^c)^c}[ M^{s^*}(G^{s^*})]^c\right)=\varnothing$.  Since $(A^c)^c=A$, we have 
	$
	 \left[\cup_{s^*\in A} M^{s^*}(G^{s^*}) \right]^c=\cap_{s^*\in (A^c)^c} [M^{s^*}(G^{s^*})]^c.
	$. We can use equation \ref{eq: append, set relation1} to show 
	$\mathcal{H}_\mathcal{S}^{snf}(A)^c\supseteq\mathcal{H}_\mathcal{S}^{wcon}(A^c)$.
	
	2.If $s\in [\mathcal{H}_\mathcal{S}^{wnf}(A)]^c$, by Definition \ref{def: non-refutablity set in incomplete structure} $\forall s^*\in A$, $M^s(G^s)\cap M^{s^*}(G^{s^*})=\varnothing$.  As a result, $M^s(G^s)\subseteq \left[\cup_{s^*\in A} M^{s^*}(G^{s^*})\right]$. Since $(A^c)^c=A$ and 
	\[
	M^s(G^s)\subseteq \cap_{s^*\in (A^c)^c} \left(M^{s^*}(G^{s^*})^c\right).
	\]
	By the definition of $\mathcal{H}_\mathcal{S}^{scon}(A^c)$,  we have $s\in \mathcal{H}_\mathcal{S}^{scon}(A^c)$. We can use the the same set operation to find the reversed inclusion.
	
	3. Suppose not, we can find $s\in \mathcal{H}_\mathcal{S}^{wcon}(A)$ but $s\notin \mathcal{H}_\mathcal{S}^{wcon}(\mathcal{H}_\mathcal{S}^{wcon}(A))$. By the definition of weak confirmation set, it means there exists $s^*\in \mathcal{H}_\mathcal{S}^{wcon}(A)^c$ such that 
	\[M^s(G^s)\cap [M^{s^*}(G^{s^*})]^c=\varnothing \Leftrightarrow M^s(G^s)\subseteq M^{s^*}(G^{s^*}).\]
	Now, since $s\in \mathcal{H}_\mathcal{S}^{wcon}(A)$, by definition
$M^s(G^s)\cap [ \cap_{\tilde{s}\in A^c} M^{\tilde{s}}(G^{\tilde{s}})^c]\ne\varnothing.$
	Since $M^s(G^s)\subseteq M^{s^*}(G^{s^*})$, we have $  M^{s^*}(G^{s^*})\cap [ \cap_{\tilde{s}\in A^c} M^{\tilde{s}}(G^{\tilde{s}})^c]\ne\varnothing$, 
	which by definition implies $s^*\in \mathcal{H}_\mathcal{S}^{wcon}(A)$. This is a contradiction.
	
	4. The last statement follows from 3 and 1 by set operation.
\end{proof}

\subsection{Proof of Proposition \ref{prop: refutable equivalent characterization}}
\begin{proof}
	First we note that by the definition of the strong non-refutability set, we have \[\cup_{s\in A} M^s(G^s) = \cup_{s\in \mathcal{H}^{snf}_\mathcal{S}(A)} M^s(G^s).\]
	If $A$ is refutable in the Breusch sense (Definition \ref{def: complete theory, Breusch refutability}), then there exists $F_0$ that can reject $A$, so $\mathcal{F}\ne \cup_{s\in A} M^s(G^s)$. Since $\mathcal{F}= \cup_{s\in \mathcal{S}} M^s(G^s)$ by the definition of structure universe, so $\cup_{s\in \mathcal{H}^{snf}_\mathcal{S}(A)} M^s(G^s)\ne \cup_{s\in\mathcal{S}} M^s(G^s)$. Therefore, $\mathcal{H}_\mathcal{S}^{snf}(A)\ne \mathcal{S}$ holds.
	
	Conversely, if $A$ is non-refutable, then $\cup_{s\in A} M^s(G^s) =\mathcal{F}$ must hold. Then by definition $\mathcal{H}^{snf}_\mathcal{S}(A)=\{s\in\mathcal{S}: M^s(G^s)\subseteq \mathcal{F} \}=\mathcal{S}$.
\end{proof}

\subsection{Proof of Proposition \ref{prop: equivalence well-defined identification and non refutable}} 
\begin{proof}
	Define $\mathcal{S}^{-1}(F)=\{s\in S: \quad F\in M^s(G^s)\}$ as the pre-image of $F$. When $\mathcal{H}_\mathcal{S}^{snf}(A)=\mathcal{S}$, $A$ is non-refutable, so $A\cap\mathcal{S}^{-1}(F)\ne \varnothing $ holds for all $F$. By the definition of the identified set, $\Theta_A^{ID}(F)=\{\theta(s)|\, s\in A\cap\mathcal{S}^{-1}(F)\} \ne \varnothing $  holds for all $\theta$. 
	
	Conversely, if $\Theta_A^{ID}(F)=\varnothing$ for some $\theta$ and $F$, that means $A\cap\mathcal{S}^{-1}(F) =\varnothing$ by definition of identified set. As a result, $A$ is refutable, since $F\notin \cup_{s\in A} M^s(G^s)$, or equivalently $\mathcal{H}_\mathcal{S}^{snf}(A)\ne \mathcal{S}$.
\end{proof} 

\subsection{Proof of Proposition \ref{prop: conditions of strong extension} }
\begin{proof}
	First note that $A\subseteq \tilde{A}$ and ${\Theta}^{ID}_{\tilde{A}}(F)\supsetneq \Theta^{ID}_{A}(F)$. We prove the proposition  by contradiction.	Suppose $\tilde{A}$ is not a strong extension, then there exists a parameter of interest $\theta$ and $F$ such that: 1.$\Theta^{ID}_{A}(F)\ne \varnothing$, and 2.${\Theta}^{ID}_{\tilde{A}}(F)\backslash \Theta^{ID}_{A}(F)\ne \varnothing$.
	Therefore, we can find some $s\in \tilde{A}\backslash A$ such that $F\in M^s(G^s)$ and $\theta(s)\in{\Theta}^{ID}_{\tilde{A}}(F)\backslash \Theta^{ID}_A(F)$. By Definition \ref{def: non-refutablity set in incomplete structure}, $F\in M^s(G^s)\cap (\cup_{s'\in A} M^{s'}(G^{s'})\ne \varnothing$ implies $s\in \mathcal{H}_\mathcal{S}^{wnf}(A)$. As a result,  $\mathcal{H}_\mathcal{S}^{wnf}(A)\cap \tilde{A}\ne A$.
	
	Conversely, if $\mathcal{H}_\mathcal{S}^{wnf}(A)\cap \tilde{A}\ne A$, there exists $\tilde{s}\in (\tilde{A}\backslash A)\cap \mathcal{H}_\mathcal{S}^{wnf}(A)$. By the definition of $\mathcal{H}_\mathcal{S}^{wnf}(A)$, we can find some $s^*\in A$ and an $F$ such that $F\in M^{s^*}(G^{s^*})\cap M^{\tilde{s}}(G^{\tilde{s}})$.
	Let the parameter of interest $\theta$ be the structure itself: $\theta(s)=s$. Then 
	$\tilde{s}\in \Theta^{ID}_{\tilde{A}}(F)\backslash \Theta_A^{ID}(F).$
	So $\tilde{A}$ is not $\theta$-consistent hence is not a strong extension.
\end{proof}

\subsection{Proof of Proposition \ref{prop: exist of strong extension}}
\begin{proof}
	First,the maximal extension $\tilde{A}=A\cup[\mathcal{H}_{\mathcal{S}}^{snf}(A)]^c$,  so we can write \begin{equation}\label{eq: append, decomposition of maximal extension}
	\cup_{s\in \tilde{A}}M^s(G^s)= [\cup_{s\in {A}}M^s(G^s)]\cup [\cup_{s\in [\mathcal{H}_\mathcal{S}^{snf}(A)]^c}M^s(G^s)]
	\end{equation}.
	By Definition \ref{def: non-refutablity set in incomplete structure}, $\cup_{s\in {A}}M^s(G^s)=\cup_{s\in \mathcal{H}_\mathcal{S}^{snf}(A)}M^s(G^s)$, so along with (\ref{eq: append, decomposition of maximal extension}), we have
	\[\cup_{s\in \tilde{A}}M^s(G^s)=\cup_{s\in \mathcal{S}}M^s(G^s)=\mathcal{F}. \]
	As a result, $\tilde{A}$ is non-refutable and hence a well-defined extension. And by construction $\tilde{A}\cap \mathcal{H}_\mathcal{S}^{snf}(A)=A$, so $\tilde{A}$ is a strong extension.
	
	To show $\tilde{A}=A\cup[\mathcal{H}_{\mathcal{S}}^{snf}(A)]^c$ is maximal, let $\tilde{A}'$ be any strong extension. Suppose $\tilde{A}'\nsubseteq \tilde{A}$, then we can find $ s\in \tilde{A}'\backslash \tilde{A}$. By $\tilde{A}=A\cup[\mathcal{H}_{\mathcal{S}}^{snf}(A)]^c$, $s\notin A$ and $s\notin [\mathcal{H}_{\mathcal{S}}^{snf}(A)]^c$ holds. As a result $s\in \mathcal{H}_\mathcal{S}^{snf}(A)\backslash A$ holds, and $\tilde{A}'\cap\mathcal{H}_\mathcal{S}^{snf}(A)\ne A$. Since $\mathcal{S}$ is a complete structure universe, $\mathcal{H}_\mathcal{S}^{snf}(\tilde{A}')=\mathcal{H}_\mathcal{S}^{wnf}(\tilde{A}')$. As a result, $\tilde{A}'\cap\mathcal{H}_\mathcal{S}^{wnf}(A)\ne A$. By Proposition \ref{prop: conditions of strong extension}, $\tilde{A}'$ is not a strong extension. This leads to a contradiction. 
\end{proof}

\subsection{Proof of Proposition \ref{prop: incomplete model impossible to find strong ext}}
\begin{proof}
	 By  $\mathcal{H}_{\mathcal{S}}^{wnf}(A)\backslash \mathcal{H}_{\mathcal{S}}^{snf}(A)\ne ~\varnothing$ and Definition \ref{def: non-refutablity set in incomplete structure}, we can find an $F^{*}$ such that \[F^*\in \left(\cup_{s\in \mathcal{H}_\mathcal{S}^{wnf}(A)}M^s(G^s)\right)\bigg\backslash \left(\cup_{s\in \mathcal{H}_\mathcal{S}^{snf}(A)}M^s(G^s)\right).\] 
	 This means that $F^*$ cannot be generated by any structures in $\mathcal{H}_{\mathcal{S}}^{snf}(A)$.
	 
	 Let $\tilde{A}$ be any well-defined extension, and it must satisfy $F^*\in \cup_{s\in \tilde{A}} M^s(G^s)$. By the condition {\footnotesize $\left(\cup_{s\in \mathcal{H}_\mathcal{S}^{wnf}(A)}M^s(G^s)\right)\cap \left(\cup_{s\in [\mathcal{H}_\mathcal{S}^{wnf}(A)]^c}M^s(G^s)\right)=\varnothing $}, $F^*$ can only be generated by structures in $\mathcal{H}_{\mathcal{S}}^{wnf}(A)$. As a result, $\tilde{A}$ must include a structure $s^*\in \mathcal{H}_{\mathcal{S}}^{wnf}(A)\backslash \mathcal{H}_{\mathcal{S}}^{snf}(A)$ such that $F^*\in M^{s^*}(G^{s^*})$. This means $\tilde{A}\cap \mathcal{H}_\mathcal{S}^{wnf}(A)\ne A$. By Proposition  \ref{prop: conditions of strong extension}, $\tilde{A}$ cannot be a strong extension. 
\end{proof}

\subsection{Proof of Proposition \ref{prop: minimal deviation as strong extension}}
\begin{proof}
	First I show $\tilde{A}$ defined through minimal deviation extension is a well-defined extension. For any $F\in \mathcal{F}$, let
	$s_1\in \cap_{l\ne j}A_l$ be the structure that achieves the minimal deviation, i.e. $ m_j(s_1)=m^{min}(F)$ and $F\in M^{s^1}(G^{s^1})$. By the definition of the minimal deviation extension,  $s_1\in \tilde{A}$. Therefore, $\tilde{A}$ is a well-defined extension.
	
	Second, we show the $\theta$-consistent result. For any $F$ such that $\Theta_A^{ID}(F)\ne \varnothing$, there exists some structure $s_2\in A$ such that $F\in M^{s_2}(G^{s_2})$. Moreover, $s_2\in A$, by Definition \ref{def: well-defined relaxation measure} $m_j(s_2)=0$. Let $s_2^*$ be any structure in $\tilde{A}$ that also rationalizes $F$, i.e. $F\in M^{s_2^*}(G^{s_2^*})$. By construction of $\tilde{A}$, $s^*_2$ achieves minimal deviation ($m_j(s_2^*)=m^{min}(F)$), so $s_2^*$ must satisfy $0\le m_j(s_2^*)\le m_j(s_2)$. This implies that $m_j(s_2^*)=0$ holds for any $s_2^*\in \tilde{A}$ such that $F\in M^{s_2^*}(G^{s_2^*})$. Then the identified set $\Theta^{ID}_{\tilde{A}}(F)$ satisfies
	\[
	\begin{split}
	\Theta^{ID}_{\tilde{A}}(F)&= \{\theta(s): s\in \tilde{A}\quad and \quad F\in M^s(G^s)\}\\
	&=\{\theta(s): s\in \tilde{A}\quad and \quad F\in M^s(G^s)\quad and \quad m_j(s)=0\}\\
	&=\{\theta(s): s\in \cap_{l\ne j}A_l \quad and \quad F\in M^s(G^s)\quad and \quad m_j(s)=0\}
	\end{split}
	\]
	The first equality holds by definition of the identified set, the second holds by $m_j(s_2^*)=0$ for all $s_2^*\in \tilde{A}$, the third equality holds by the definition of $\tilde{A}$ and $m_j\ge 0$. So $\tilde{A}$ is $\theta$-consistent if 
	\[\Theta_A^{ID}(F)=\{\theta(s): s\in \cap_{l\ne j}A_l\quad and \quad F\in M^s(G^s)\quad and \quad m_j(s)=0\}\]
	holds.
	
	Last, I prove the strong extension statement by contradiction. Suppose $\tilde{A}$ is not a strong extension, then there exist some $s_3\in (\mathcal{H}_\mathcal{S}^{wnf}(A)\backslash A)\cap \tilde{A}$, $s_4\in A$, and an observable distribution $F$ such that $F\in M^{s_3}(G^{s_3})\cap M^{s_4}(G^{s_4})\ne \varnothing$. Since $s_4\in A$, $m_j(s_4)=m^{min}(F)=0$. By construction, $s_3\in \tilde{A}$, and $s_3$ achieves the minimal deviation, so $m_j(s_3)=m^{min}(F)=0$ must hold. Since $s_3\notin A=\cap_{l}A_l$, but at the same time $s_3\in \cap_{l\ne j}A_l$, then it must be the case that  $ A_j\ne \{s\in \cap_{l\ne j}A_l: \,\,m_j(s)=0\}$ holds. The result follows by contradiction.
\end{proof}

\subsection{Proof of Lemma \ref{lem: charaterization of strong binary detectable}}

\begin{proof}
	$\Rightarrow$:
	Suppose $\mathcal{H}_{\tilde{A}}^{scon}(H)\ne \mathcal{H}_{\tilde{A}}^{wnf}(H)$, by Lemma \ref{lem: charaterization of A set with scon}, there exists an
	$s\in \mathcal{H}_{\tilde{A}}^{wnf}(H)$ but $s\notin H$. By definition of $s\in\mathcal{H}_{\tilde{A}}^{wnf}(H)$,  $\exists F\in M^s(G^s)\cap (\cup_{s^*\in H}M^{s^*}(G^{s^*}))$, it means $H$ cannot be decided by $F$, because $F\in \cup_{s^*\in H^c}M^{s^*}(G^{s^*})$ and there exists some $s^\prime\in H$ such that $F\in M^{s^\prime}(G^{s^\prime})$.
	
	$\Leftarrow$: Suppose there exists $F$ such that $H$ can not be decided by $F$, then $F\in\cup_{s^*\in [H^c\cap {\tilde{A}}]}  M^{s^*}(G^{s^*})$ and $F\in \cup_{s\in H}M^s(G^s)$. This means we can find an $s\in [H^c\cap \tilde{A}]$, such that $F\in M^{s}(G^{s})$, and find $\tilde{s}\in H$ such that $F\in M^{\tilde{s}}(G^{\tilde{s}})$. By Definition \ref{def: non-refutablity set in incomplete structure}, $s\in \mathcal{H}_{\tilde{A}}^{wnf}(H)$. As a result, $H\ne \mathcal{H}_{\tilde{A}}^{wnf}(H)$. By Lemma \ref{lem: charaterization of A set with scon},  $\mathcal{H}_{\tilde{A}}^{scon}(H)\subsetneq H$, we have $\mathcal{H}_{\tilde{A}}^{scon}(H)\ne \mathcal{H}_{\tilde{A}}^{wnf}(H)$. 
\end{proof}

\subsection{Proof of Proposition \ref{prop: binary decidable and size-power issue}}

\begin{proof}
	By Lemma \ref{lem: charaterization of strong binary detectable} and \ref{lem: charaterization of A set with scon}, $H$ is not strongly binary decidable implies $H\ne \mathcal{H}_{\tilde{A}}^{wnf}(H)$. So we can find $s\in H$ and $\tilde{s}\in \mathcal{H}_{\tilde{A}}^{wnf}(H)\backslash H$ and an $\tilde{F}$ such that $\tilde{F}\in M^s(G^s)\cap M^{\tilde                                                                                                                                                                                                                                                                                                                                                                                                                                                                                                                                                                                                                                                                                             {s}}(G^{\tilde{s}})$.
	Let $\tilde{\mathbb{F}}_n$ be any empirical distribution sampled from $\tilde{F}$ such that $\tilde{\mathbb{F}}_n$ converges to $\tilde{F}$ weakly. If (\ref{eq: pointwise size control}) holds for some $\alpha<1$, then since $\tilde{s}\in H^c$, we look at the LHS of equation (\ref{eq: test consistency}):
	\[
	\begin{split}
	\quad \inf_{F\in\cup_{s\in A^c} M^s(G^s)} &{\lim \sup}_{n\rightarrow \infty } Pr(T_n(\mathbb{F}_n,\mathbf{\eta})=0)\\
	&\le  {\lim \sup}_{n\rightarrow \infty } Pr(T_n(\tilde{\mathbb{F}}_n,\tilde{\eta})=0)\\
	&= {\lim \sup}_{n\rightarrow \infty } 1-Pr(T_n(\tilde{\mathbb{F}}_n,\tilde{\eta})=1)\\
	&= 1-{\lim \inf}_{n\rightarrow \infty } Pr(T_n(\tilde{\mathbb{F}}_n,\tilde{\eta})=1)\\
	&\le_{(1)} 1-(1-\alpha)=\alpha<1,
	\end{split}
	\]
	where inequality $(1)$ follows by the test consistency requirement (\ref{eq: pointwise size control}) for $s\in H$ and $\tilde{F}\in M^s(G^s)$. Therefore, we cannot achieve pointwise size control and test consistency simultaneously. 
\end{proof}

\subsection{Proof of Proposition \ref{prop: existence of size-consistecy test stat for strongly bin decidable}}
\subsubsection{Some Lemmas and Additional Propositions}
\begin{lem}\label{lem: append, strong binary decidable and id set}
	If the hypothesis $H=\{s\in {\tilde{A}}: \,\theta(s)\in \Theta^0\}$ is strongly binary decidable, then for all $F\in \mathcal{F}$, exactly one of the following holds:
	\begin{enumerate}
		\item $\Theta^0\cap \Theta^{ID}_{{\tilde{A}}}(F)=\varnothing$;
		\item $\Theta^{ID}_{{\tilde{A}}}(F)\subseteq \Theta^0$.
	\end{enumerate}
\end{lem}
\begin{proof}
	Note that 1 and 2 in the Lemma cannot hold simultaneously, since $\tilde{A}$ is non-refutable (see Proposition \ref{prop: equivalence well-defined identification and non refutable}). Now suppose both 1 and 2 do not hold. Then we can find an $F$ and parameter values $\theta^{val},\tilde{\theta}^{val}$ such that: ($i$). $\theta^{val}\in\Theta^0\cap \Theta^{ID}_{{\tilde{A}}}(F)$; and ($ii$). $\tilde{\theta}^{val}\in \Theta^{ID}_{{\tilde{A}}}(F)\backslash \Theta^0$. Condition ($i$) implies that $F$ can be generated by a structure in $H$:$
	F\in \cup_{s\in H}M^s(G^s).
	$ Condition  ($ii$) implies that there exists a $\tilde{\theta}\in H^c\cap {\tilde{A}}$ such that $F\in M^s(G^s)$. As a result, $
	F\in \cup_{s\in H^c\cap {\tilde{A}}}M^s(G^s).$
	Therefore, by Definition \ref{def: weakly binary decidable}, $H$ is not binary decidable by $F$. 
\end{proof}

\begin{lem} \label{lem: implication of upper hemi-continuity}
	Let $\Theta^{ID,\epsilon}_{\tilde{A}}(F)=\{\theta\in\Theta\big| d_{\theta}(\theta,\Theta_{\tilde{A}}^{ID}(F))<\epsilon\}$ be the $\epsilon$-enlargement of identified set $\Theta_{\tilde{A}}^{ID}(F)$.  Let Assumption \ref{assumption: consistency} holds, then $\forall \epsilon>0$, there exists a $\delta(\epsilon)>0$ such that for all $\tilde{\theta}\in \Theta\backslash \Theta^{ID,\epsilon}_{\tilde{A}}(F)$ and $\forall s$ such that $\theta(s)=\tilde{\theta}$, the following holds:
	\[\inf_{F^*\in M^s(G^s)} d_{\tilde{\mathcal{F}}}(F^*,F) \ge \delta(\epsilon).\]	
\end{lem}
\begin{proof}
	Since $\Theta_{\tilde{A}}^{ID}(F)$ is upper hemicontinuous at $F$, and $\Theta^{ID,\epsilon}_{\tilde{A}}(F)$ is an open neighborhood of $\Theta_{\tilde{A}}^{ID}(F)$, there exists an open neighborhood $U$ of $F$ such that $\Theta^{ID}_{\tilde{A}}(F^*)\subseteq \Theta^{ID,\epsilon}_{\tilde{A}}({F})$ for all $\tilde{F}\in U$. Let $\delta(\epsilon)= \sup_{F_1,F_2\in U} d_{\tilde{\mathcal{F}}}(F_1,F_2)/2 $ be the diameter of $U$. Consider any $\tilde{\theta}\in \Theta\backslash \Theta^{ID,\epsilon}_{\tilde{A}}(F)$ and $s$ such that $\theta(s)=\tilde{\theta}$. 
	
	We claim that $M^s(G^s)\cap U=\varnothing$. Suppose not, there exists an $\tilde{F} \in M^s(G^s)\cap U$ and by the upper hemi-continuity property in the previous paragraph, $\tilde{\theta}=\theta(s)\in \Theta^{ID}_{\tilde{A}}(\tilde{F})\subseteq \Theta^{ID,\epsilon}_{\tilde{A}}(F)$ $\Rightarrow$, which is a contradiction. Since $M^s(G^s)\cap U=\varnothing$ holds, $\inf_{F^*\in M^s(G^s)} d_{\tilde{\mathcal{F}}}(F^*,F) \ge \delta(\epsilon)$ holds. 
\end{proof}
We now introduce the following notation:
\[J_{n}(F;c_n,a_n)\equiv\left\{ F^*\in \mathcal{F}\cup\mathcal{F}^d\big| d_{\tilde{\mathcal{F}}}(F^*,F)<c_n/\sqrt{a_n}	\right\},\]
which is a generalization of the dilation map in \citet{Galichon2013DilationBoostrap}. We consider the following set:
\begin{equation}\label{eq: General Est Inf method, est ID set}
\hat{\Theta}^{ID}=\{\theta(s)\big| s\in {\tilde{A}}\,,\,\, \mathbb{F}_n\in J_{n}(M^s(G^s);c_n,a_n)\}. 
\end{equation}
\begin{prop}\label{prop: consistency of estimated identified set}
	Under Assumption \ref{assumption: consistency}, $\Theta_{\tilde{A}}^{ID}(F)\subseteq \hat{\Theta}^{ID}$ with probability approaching 1, and $d_H(\Theta_{\tilde{A}}^{ID}(F), \hat{\Theta}^{ID})=o_p(1)$, where $d_H$ is the Hausdorff distance.
\end{prop}
\begin{proof}
	Recall that $\mathbb{F}_n$ is the empirical distribution sampled from the observed distribution $F$.
	
	If $\tilde{\theta}\in \Theta_{\tilde{A}}^{ID}(F)$, there exists an $s$ such that $\theta(s)=\tilde{\theta}$, and $F\in M^s(G^s)$. Therefore, we have that
	\begin{equation}\label{eq:append, convergence of the identified set 1}
	\inf_{F^*\in M^s(G^s)} \sqrt{a_n}d_{\tilde{\mathcal{F}}}(\mathbb{F}_n,F^*) \le \sqrt{a_n}d_{\tilde{\mathcal{F}}}(\mathbb{F}_n,F) \le c_n
	\end{equation} 
	holds with probability approaching 1 by the assumption of this proposition. In equation (\ref{eq:append, convergence of the identified set 1}), the right hand side $c_n$ does not depend on the value of $\tilde{\theta}$. Therefore $\Theta_{\tilde{A}}^{ID}(F)\subseteq \hat{\Theta}^{ID}$ with probability approaching 1. This proves the first claim.
	
	Next, I show $d_{H}(\hat{\Theta}^{ID},\Theta_{\tilde{A}}^{ID}(F))\rightarrow_p 0$ by showing that $\hat{\Theta}^{ID}$ does not intersect $\Theta\backslash \Theta^{ID,\epsilon}$ with probability approaching 1 for all $\epsilon>0$. By the definition of $\hat{\Theta}^{ID}(F)$, it suffices to show that 	
	\[\inf_{s:\theta(s)\in \Theta\backslash \Theta^{ID,\epsilon}_{\tilde{A}}(F)}\left[ \inf_{F^*\in M^s(G^s)} d_{\tilde{\mathcal{F}}}(\mathbb{F}_n,F^*)\right] >c_n/\sqrt{a_n}\]
	holds with probability approaching 1. Note that 
	\begin{equation*}
	\begin{split}
	&\quad\inf_{s:\theta(s)\in \Theta\backslash \Theta^{ID,\epsilon}_{\tilde{A}}(F)}\left[ \inf_{F^*\in M^s(G^s)} d_{\tilde{\mathcal{F}}}(\mathbb{F}_n,F^*)\right]\\
	&\ge \inf_{s:\theta(s)\in \Theta\backslash \Theta^{ID,\epsilon}_{\tilde{A}}(F)}\left[ \inf_{F^*\in M^s(G^s)} d_{\tilde{\mathcal{F}}}(F,F^*)-d_{\tilde{\mathcal{F}}}(\mathbb{F}_n,F)\right]\\
	&\ge \delta(\epsilon) -O_p(1/\sqrt{a_n}),
	\end{split}
	\end{equation*}
	where the last inequality follows from Lemma \ref{lem: implication of upper hemi-continuity}. Since $c_n/\sqrt{a_n}\rightarrow 0$, $Pr(\delta(\epsilon) -O_p(1/\sqrt{a_n})>c_n/\sqrt{a_n})\rightarrow 1$. Therefore, $\Theta^{ID}_{\tilde{A}}$ does not intersect $\Theta\backslash \Theta^{ID,\epsilon}$ with probability approaching 1. 
\end{proof}
\subsubsection{Main Proof of Proposition \ref{prop: existence of size-consistecy test stat for strongly bin decidable}}
\begin{proof}
	 Let $\hat{\Theta}^{ID}$ be the set given in (\ref{eq: General Est Inf method, est ID set}). We consider a test statistic $T_1(\mathbb{F}_n,\eta)$ such that $T_1(\mathbb{F}_n,\eta)=1$ if $\Theta^0\cap \hat{\Theta}^{ID}\ne \varnothing$, and $T_1(\mathbb{F}_n,\eta)=0$ otherwise.

	By Proposition \ref{prop: consistency of estimated identified set}, for any $\epsilon>0$, $\Theta^{ID}_{{\tilde{A}}}(F)\subset \hat{\Theta}^{ID}\subset \Theta^{ID,\epsilon}_{{\tilde{A}}}(F)$ holds with probability approaching 1, where $\Theta^{ID,\epsilon}_{\tilde{A}}(F)=\{\theta\in\Theta\big| d(\theta,\Theta_{\tilde{A}}^{ID}(F))<\epsilon\}$.
	
	By Lemma \ref{lem: append, strong binary decidable and id set}, for any $F\in \cup_{s\in H} M^s(G^s)$, $\Theta^{ID}_{{\tilde{A}}}(F)\subseteq \Theta^0$ holds. Therefore
	\[
	\begin{split}
	{\lim\inf}_{n\rightarrow \infty} Pr(T_1(\mathbb{F}_n,\mathbf{\eta})=1)&=	{\lim\inf}_{n\rightarrow \infty} Pr(\hat{\Theta}^{ID}\cap\Theta^0\ne \varnothing)\\
	&\ge {\lim\inf}_{n\rightarrow \infty} Pr(\Theta^{ID}_{{\tilde{A}}}(F)\subseteq\hat{\Theta}^{ID})=1.
	\end{split}
	\]
	This shows that the test statistic $T_1$ achieves pointwise size control.
	
	To show the test consistency, let $s\in H^c\cap \tilde{A}$ be any structure such that $\theta(s)\in [\Theta^{0}]^c$ and $F\in M^s(G^s)$. By definition $F\in \cup_{s\in [H^c\cap \tilde{A}]}M^s(G^s)$. Since $\Theta^{0}$ is closed, we can find an $\epsilon>0$ such that $\theta(s)\in [\Theta^{0,\epsilon}]^c$, where $\Theta^{0,\epsilon}$ is the $\epsilon$-enlargement of $\Theta^{0}$. By Lemma \ref{lem: append, strong binary decidable and id set}, $\Theta^{ID}_{{\tilde{A}}}(F)\cap \Theta^{0,\epsilon}=\varnothing$, or equivalently $\Theta^{ID,\epsilon}_{{\tilde{A}}}(F)\cap \Theta^{0}=\varnothing$ holds. Therefore 
	\[
	\begin{split}
	{\lim \sup}_{n\rightarrow \infty } Pr(T_1(\mathbb{F}_n,\mathbf{\eta})=0)&={\lim \sup}_{n\rightarrow \infty } Pr(\hat{\Theta}^{ID}\cap\Theta^{0}=\varnothing)\\
	&\ge {\lim \sup}_{n\rightarrow \infty } Pr(\hat{\Theta}^{ID}\subseteq \Theta^{ID,\epsilon}_{{\tilde{A}}}(F))=_{(*)}1.
	\end{split}
	\]
	where $(*)$ follows by Proposition  \ref{prop: consistency of estimated identified set}. As a result, the test statistic $T_1$ achieves the test consistency (\ref{eq: test consistency}).
\end{proof}

\subsection{Proof of Proposition \ref{prop: point identification implies strong binary decidable}}
\begin{proof}
	I first show the result for the point identified parameter $\theta$. By Lemma \ref{lem: charaterization of A set with scon}, it suffices to show $\mathcal{H}_{{\tilde{A}}}^{wnf}(H)=H$. Since $H\subseteq \mathcal{H}_{{\tilde{A}}}^{wnf}(H)$ it suffices to show $\mathcal{H}_{{\tilde{A}}}^{wnf}(H)\subseteq H$. 
	
	Take a structure $\tilde{s}\in \mathcal{H}_{{\tilde{A}}}^{wnf}(H)$, by definition, there exists an $s'\in H$ and an $F$ such that $F\in M^{\tilde{s}}(G^{\tilde{s}})\cap M^{s'}(G^{s'})$. By definition of the identified set $\Theta_{\tilde{A}}^{ID}(F)=\{\theta(s): s\in \tilde{A},\,\, F\in M^s(G^s)\}$, and therefore $\{\theta(\tilde{s}),\theta(s')\}\subseteq \Theta_{\tilde{A}}^{ID}(F)$. By point identification assumption,  $\Theta_{\tilde{A}}^{ID}(F)$ is a singleton, so $\theta(\tilde{s})=\theta(s')\in \Theta_0$. This shows $\tilde{s}\in H$.

	Now I show the result for the partially identified $\theta$. Without loss of generality, let's assume $\Theta_{\tilde{A}}^{ID}(F)\backslash \Theta_{\tilde{A}}^{ID}(F')\ne \emptyset$. Let $\theta_1\in \Theta_{\tilde{A}}^{ID}(F)\cap \Theta_{\tilde{A}}^{ID}(F')$ and $\theta_2 \in \Theta_{\tilde{A}}^{ID}(F)\backslash \Theta_{\tilde{A}}^{ID}(F')$. I claim that $H= \{s\in \tilde{A}: \theta(s)\in \Theta_0\}$ is not binary decidable for $\Theta_0=\{\theta_2\}$. 
	
	Indeed, let $S_k=\{s\in \tilde{A}: \theta(s)=\theta_k\}$ for $k=1,2$.  Then $H= S_2$ and $\mathcal{H}_{\tilde{A}}^{wnf}(H)=S_1\cup S_2$. By our assumption of Proposition \ref{prop: point identification implies strong binary decidable}, $S_k\ne \emptyset$ for all $k=1,2$, so $H\ne \mathcal{H}_{\tilde{A}}^{wnf}(H)$. By Lemma \ref{lem: charaterization of A set with scon}, $H$ is not strongly binary decidable.
\end{proof}

\subsection{Proof of Proposition \ref{Prop: smallest binary decidable extension (shrinkage)}}
\begin{proof}
	Let $H^{ext}$ be any non-trivial strongly binary decidable extension ($H^{ext}\ne \tilde{A}$). By Lemma \ref{lem: charaterization of strong binary detectable} 
	and \ref{lem: charaterization of A set with scon}, we have 
	$H^{ext}=\mathcal{H}^{wnf}_{{\tilde{A}}}(H^{ext})=\mathcal{H}^{scon}_{{\tilde{A}}}(H^{ext})$. By (\ref{eq: append, set relation for nf and con}), we have $
	H^{ext}=\mathcal{H}^{snf}_{{\tilde{A}}}(H^{ext})$.

	Since $H\subseteq H^{ext}$, we have \[\mathcal{H}^{snf}_{{\tilde{A}}}(H)\subseteq \mathcal{H}^{snf}_{{\tilde{A}}}(H^{ext})=H^{ext}.\]
	The above inclusion says that any binary decidable extension must includes $\mathcal{H}^{snf}_{{\tilde{A}}}(H)$. If we show $\mathcal{H}^{snf}_{{\tilde{A}}}(H)$ is a strongly binary decidable extension, it must be the smallest.
	
	Note that by definition of $\mathcal{H}_{{\tilde{A}}}^{snf}(H)$, 
$
	\cup_{s\in \mathcal{H}_{{\tilde{A}}}^{snf}(H) } M^s(G^s) =\cup_{s\in H } M^s(G^s).
$
	By applying Definition \ref{def: non-refutablity set in incomplete structure}, we have 
	\[
	\begin{split}
	\mathcal{H}_{{\tilde{A}}}^{wnf}(\mathcal{H}_{{\tilde{A}}}^{snf}(H)) 
	&=\left\{s\in {\tilde{A}}: M^s(G^s)\cap \left(\cup_{s^*\in \mathcal{H}_{{\tilde{A}}}^{snf}(H) } M^{s^*}(G^{s^*})\right)\ne \varnothing \right\}\\
	&=\left\{s\in {\tilde{A}}: M^s(G^s)\cap \left(\cup_{s^*\in H } M^{s^*}(G^{s^*})\right)\ne \varnothing \right\} \\
	&=\mathcal{H}_{{\tilde{A}}}^{wnf}(H) = \mathcal{H}_{{\tilde{A}}}^{snf}(H),
	\end{split} \]
	where the last equality holds by the assumption that  $\mathcal{H}_{{\tilde{A}}}^{wnf}(H) =\mathcal{H}_{{\tilde{A}}}^{snf}(H)$. The above equality implies Condition 3 in Lemma \ref{lem: charaterization of A set with scon} holds for $\mathcal{H}_{{\tilde{A}}}^{snf}(H)$. As a result, $\mathcal{H}_{{\tilde{A}}}^{wnf}(\mathcal{H}_{{\tilde{A}}}^{snf}(H))=\mathcal{H}_{{\tilde{A}}}^{scon}(\mathcal{H}_{{\tilde{A}}}^{snf}(H))$. So $\mathcal{H}_{{\tilde{A}}}^{snf}(H)$ is strongly binary decidable by Lemma \ref{lem: charaterization of strong binary detectable}.

	Let $H^{sub}$ be any non-trivial strongly binary decidable subset set ($H^{ext}\ne \emptyset$). By Lemma \ref{lem: charaterization of strong binary detectable} and \ref{lem: charaterization of A set with scon}, it implies 
	\[H^{sub}=\mathcal{H}^{scon}_{{\tilde{A}}}(H^{sub})=\mathcal{H}^{wcon}_{{\tilde{A}}}(H^{sub}).\]
	Since $H^{sub}\subseteq H$, we have 
	\[
	H^{sub}=\mathcal{H}^{scon}_{{\tilde{A}}}(H^{sub})\subseteq \mathcal{H}^{scon}_{{\tilde{A}}}(H).
	\]
	The above inclusion says any strongly binary decidable shrinkage must be included in  $\mathcal{H}^{scon}_{{\tilde{A}}}(H)$. If we can show that  $\mathcal{H}^{scon}_{{\tilde{A}}}(H)$ is a strongly binary decidable subset set, it must be the largest. By the same argument as shown in the proof of non-refutability set, we can show $\mathcal{H}^{scon}_{{\tilde{A}}}(\mathcal{H}^{wcon}_{{\tilde{A}}}(H))=\mathcal{H}^{scon}_{{\tilde{A}}}(H)=\mathcal{H}^{wcon}_{{\tilde{A}}}(H)$, so by Lemma \ref{lem: charaterization of A set with scon}, it is equivalent to 
	\[
	\mathcal{H}^{scon}_{{\tilde{A}}}(\mathcal{H}^{wcon}_{{\tilde{A}}}(H))=\mathcal{H}^{wnf}_{{\tilde{A}}}(\mathcal{H}^{wcon}_{{\tilde{A}}}(H)).
	\]
	As a result, $\mathcal{H}^{wcon}_{{\tilde{A}}}(H)$ is strongly binary decidable by Lemma \ref{lem: charaterization of strong binary detectable}.
\end{proof}

\subsection{Proof of Proposition \ref{prop: equivalent decision}}
\begin{proof}
	I first prove the first part of the proposition. 
	
	"$1\Rightarrow 2$": If $\exists s\in H$ and $F\in M^s(G^s)$, and since $F$ is what we observe, then $F\in M^{s^0}(G^{s^0})$ holds. By definition $s^0\in \mathcal{H}_{{\tilde{A}}}^{wnf}(H)$. By assumption $\left(\mathcal{H}_{{\tilde{A}}}^{wnf}(H)\backslash \mathcal{H}_{{\tilde{A}}}^{snf}(H)\right) =\varnothing$, we have $s^0\in \mathcal{H}_{{\tilde{A}}}^{snf}(H)$
	
	"$1\Leftarrow 2$": If $s^0\in \mathcal{H}_{{\tilde{A}}}^{snf}(H)$ and $F\in M^{s^0}(G^{s^0})$, by definition of $\mathcal{H}_{{\tilde{A}}}^{snf}(H)$, there exists an $s\in H$ such that $F\in M^s(G^s)$.
	
	I now show the second part of the proposition.
	
	"$1\Rightarrow 2$": If $\{s\in \tilde{A}: F\in M^s(G^s)\}\subset H$, and since $F$ is what we observe, then $F\in M^{s^0}(G^{s^0})$ holds. By definition $s^0\in \mathcal{H}_{{\tilde{A}}}^{wcon}(H)$. By assumption $\left(\mathcal{H}_{{\tilde{A}}}^{wcon}(H)\backslash \mathcal{H}_{{\tilde{A}}}^{scon}(H)\right) =\varnothing$, it implies that $s^0\in \mathcal{H}_{{\tilde{A}}}^{scon}(H)$.
	
	"$2\Rightarrow 1$":  Suppose 1 does not hold. So we can find a structure $s^*\in H^c$ such that $F\in M^{s^*}(G^{s^*})$. Since  $F\in M^{s_0}(G^{s_0})$, $F\in \left(\cup_{s\in [H^c\cap \tilde{A}]} M^{s}(G^{s})\right)\cap M^{s_0}(G^{s_0})$ holds. As a result, $M^{s_0}(G^{s_0})\nsubseteq \left(\cap_{s\in [H^c\cap \tilde{A}]} M^{s}(G^{s})^c\right)$, so 2 does not hold.	
\end{proof}

\section{Proof in Section 3}\label{section: Proofs in section 3}
\subsection{Proof of Theorem \ref{thm: testable implication in density form}}

\begin{proof}
	By the definition of $P(B,d)$ and $Q(B,d)$ in (\ref{eq: potential outcome, link between F and G}), they are measures since they are generated by probability measures. Moreover, they are finite measures, bounded above by 1. 
	
	Now, suppose the testable implication in the Radon-Nikodym form (\ref{eq: testable implication in density form}) holds, then for any Borel measurable set $B$,
	\[
	\begin{split}
	P(B,1)-Q(B,1)=\int_{B} p(y,1)-q(y,1) d\mu_F \ge 0\\
	Q(B,0)-P(B,0)=\int_{B} q(y,0)-p(y,0) d\mu_F \ge 0
	\end{split}
	\] 
	holds.
	
	Conversely, suppose testable implication in Radon-Nikodym form (\ref{eq: testable implication in density form}) fails. Without loss of generality, let $B_1$ be the set that $\mu_F(B_1)>0$ and $p(y,1)-q(y,1)<0$ for all  $y\in B_1$. 
	By Lemma \ref{lem: existence of positive measure set },  there exist a measurable set $B_1'\subseteq B_1$ with $\mu_F(B_1')>0$ such that 
	\[\int_{B_1'} p(y,1)-q(y,1) d\mu_F<0.\]
	As a result, 
	\[
	P(B_1',1)-Q(B_1',1)=\int_{B_1'} p(y,1)-q(y,1) d\mu_F <0
	\]
	So the testable implication (\ref{eq: testable implication of LATE}) fails. 
\end{proof}

\subsection{Proofs in Section 3.1}

\subsubsection{Proof of Lemma \ref{lem: alternative representation of IA-M assumption}}
\begin{proof}
	Let $A'=A^{ER}\cap A^{TI}\cap A^{EM-NTAT}\cap A^{ND}$ be the alternative representation, and let $A$ be the representation in (\ref{eq: IA instrument assumption }). Note that if $Z_i$ is an instrument independent of the potential outcomes, then: (1). $Z_i$ is also type independent; (2). The measure of always takers and never takers is independent of $Z_i$. Therefore, we have $A\subseteq A'$.
	
	Conversely, let $s\in A'$. It suffices to show the  condition $\{Y_i(d,z),D_i(z)\}_{d,z\in\{0,1\}}\perp Z_i$ holds for $s$. For any $B_1,B_0$ set, by the exclusion restriction of $s$ we have 
	\begin{equation}\label{eq: full independence under tilde A type C}
	\begin{split}
	&\quad Pr_{G^s}(Y_i(1,0)=Y_i(1,1)\in B_1, Y_i(0,0)=Y_i(0,1)\in B_0,D_i(1)=1,D_i(0)=0|Z_i=1)\\
	&=_{(1)} Pr_{G^s}(Y_i(1,0)=Y_i(1,1)\in B_1, Y_i(0,0)=Y_i(0,1)\in B_0|D_i(1)=1,D_i(0)=0,Z_i=1)\\
	&\quad \times Pr_{G^s}(D_i(1)=1,D_i(0)=0|Z_i=1)\\
	&=_{(2)} Pr_{G^s}(Y_i(1,0)=Y_i(1,1)\in B_1, Y_i(0,0)=Y_i(0,1)\in B_0|D_i(1)=1,D_i(0)=0,Z_i=1)\\
	&\times (1-Pr_{G^s}(D_i(1)=1,D_i(0)=1|Z_i=1)-Pr_{G^s}(D_i(1)=0,D_i(0)=0|Z_i=1))\\
	&=_{(3)} Pr_{G^s}(Y_i(1,0)=Y_i(1,1)\in B_1, Y_i(0,0)=Y_i(0,1)\in B_0|D_i(1)=1,D_i(0)=0,Z_i=0)\\
	&\times (1-Pr_{G^s}(D_i(1)=1,D_i(0)=1|Z_i=1)-Pr_{G^s}(D_i(1)=0,D_i(0)=0|Z_i=1))\\
	&=_{(4)} Pr_{G^s}(Y_i(1,0)=Y_i(1,1)\in B_1, Y_i(0,0)=Y_i(0,1)\in B_0|D_i(1)=1,D_i(0)=0,Z_i=0)\\
	&\times (1-Pr_{G^s}(D_i(1)=1,D_i(0)=1|Z_i=0)-Pr_{G^p}(D_i(1)=0,D_i(0)=0|Z_i=0))\\
	&=_{(5)}Pr_{G^s}(Y_i(1,0)=Y_i(1,1)\in B_1, Y_i(0,0)=Y_i(0,1)\in B_0,D_i(1)=1,D_i(0)=0|Z_i=0),
	\end{split}
	\end{equation}
	where (1) and (5) follow by the formula of conditional probability, (2) follows by $D_i(1)\ge D_i(0)$ almost surely under $s\in A'$, (3) follows by the type independence instrument assumption, (4) follows by the measure of always takers and never takers is independent of $Z_i$. We can similarly show for $j=0,1$:
	\begin{equation}\label{eq: full independence under tilde A type AT NT}
	\begin{split}
	&\quad Pr_{G^s}(Y_i(1,0)=Y_i(1,1)\in B_1, Y_i(0,0)=Y_i(0,1)\in B_0,D_i(1)=D_i(0)=j|Z_i=1)\\
	&=Pr_{G^s}(Y_i(1,0)=Y_i(1,1)\in B_1, Y_i(0,0)=Y_i(0,1)\in B_0,D_i(1)=D_i(0)=j|Z_i=0).
	\end{split}
	\end{equation}
	We then use (\ref{eq: full independence under tilde A type C}) and (\ref{eq: full independence under tilde A type AT NT}) to conclude the independence condition $\{Y_i(d,z),D_i(z)\}_{d,z\in\{0,1\}}\perp Z_i$ holds when there are no defiers. So $s\in A$ and $A'\subseteq A$ holds.
\end{proof} 

\subsubsection{Proof of Proposition \ref{prop: appli, maximal extension arbitrary instrument}}

\begin{proof}
	Let $\tilde{A}'$ be the extension constructed in Assumption \ref{assump: minimal dist to marg ind inst}, and let $\tilde{A}^{max}$ be the extension constructed in Proposition \ref{prop: appli, maximal extension arbitrary instrument}. By Proposition \ref{prop: minimal marg ind as consistent extension}, $\tilde{A}'$ is a well-defined extension, so $\mathcal{H}_{\mathcal{S}}^{snf}(\tilde{A}')=\mathcal{S}$.  By construction, $\tilde{A}'\subseteq \tilde{A}^{max}$, so  $\mathcal{S}=\mathcal{H}_{\mathcal{S}}^{snf}(\tilde{A}')\subseteq\mathcal{H}_{\mathcal{S}}^{snf}(\tilde{A}^{max})\subseteq\mathcal{S}$. Therefore, $\tilde{A}$ is a well-defined extension. Moreover, $\tilde{A}^{max}\cap \mathcal{H}_{\mathcal{S}}^{snf}(A)= A\cap A^{ER}\cap A^{ND}=A$, so $\tilde{A}^{max}$ is a strong extension. When (\ref{eq: testable implication in density form}) holds for $F$, the identification result follows by equation (\ref{eq: appli, ID set}) and Definition  \ref{def: consistent extension, theta and strong}. 
	
	We now show that the identified set in Proposition \ref{prop: appli, maximal extension arbitrary instrument} when the constraint (\ref{eq: testable implication in density form}) fails. First, conditional on the compliers group $D_i(1)=1,D_i(0)=0$ and $Z_i=1$, we know $Y_i(1,z)\in (\underline{\mathcal{Y}}_{P(B,1)},\bar{\mathcal{Y}}_{P(B,1)})$. Similarly, conditional on the compliers group and $Z_i=0$  and $Y_i(0,z)\in (\underline{\mathcal{Y}}_{Q(B,0)},\bar{\mathcal{Y}}_{Q(B,0)})$. Therefore, the identified set is valid. We now show the sharpness of this identified set by explicitly construct a sequence of structures that achieves the bound. 
	
	\paragraph{Step 1. Construction of a $G^s$.}To find the identified set $LATE_{\tilde{A}^{max}}^{ID}(F)$ when $F$ fails (\ref{eq: testable implication in density form}), we consider the following $G^s$ for all measurable sets $B_{dz}$:
	\begin{equation}
	\begin{split}
	&Pr_{G^s}(Y_i(d,z)\in B_{dz}\quad \forall d,z\in\{0,1\}, D_i(1)=1,D_i(0)=1|Z_i=z)\\
	&=\begin{cases}
	G^a(Y_i(0,0)\in B_{00}\cap B_{01})\times \int_{B_{10}\cap B_{11}} p(y,1)-g_c^{1s}(y) d\mu_F(y)\quad &if \quad z=1,\\
	G^a(Y_i(0,0)\in B_{00}\cap B_{01})\times \int_{B_{10}\cap B_{11}} q(y,1)d\mu_F(y)\quad &if \quad z=0,
	\end{cases}
	\end{split}
	\end{equation}
	where $G^a$ is any probability measure, and
	\begin{equation}
	\begin{split}
	&Pr_{G^s}(Y_i(d,z)\in B_{dz}\quad \forall d,z\in\{0,1\}, D_i(1)=0,D_i(0)=0|Z_i=z)\\
	&=\begin{cases}
	G^n(Y_i(1,1)\in B_{11}\cap B_{10})\times \int_{B_{01}\cap B_{00}} q(y,0)-g_c^{0s}(y) d\mu_F(y)\quad &if \quad z=0,\\
	G^n(Y_i(1,1)\in B_{11}\cap B_{10})\times \int_{B_{01}\cap B_{00}} p(y,0)d\mu_F(y)\quad &if \quad z=1,
	\end{cases}
	\end{split}
	\end{equation}
	where $G^n$ is any probability measure. We can consider the $g_c^{1s}$ and $g_{c}^{0s}$ to take the form 
	\[
	\begin{split}
	g_c^{1s}(y)= \alpha \mathbbm{1}(y\in [k_1-\epsilon,k_1])\times   p({y},1)/2 +(1-\alpha) \mathbbm{1}(y\in [k_1',k_1'+\epsilon])\times   p({y},1)/2,\\
	g_c^{0s}(y)=\beta \mathbbm{1}(y\in [k_0,k_0+\epsilon])\times   q({y},0)/2+(1-\beta)\mathbbm{1}(y\in [k_0'-\epsilon,k_0'])\times  q({y},0)/2,\\
	\end{split}
	\]
	where $\alpha,\beta,\epsilon$ are parameters that we can choose and $k_0=\underline{\mathcal{Y}}_{Q(B,0)}$, $k_1=\bar{\mathcal{Y}}_{P(B,1)}$, $k_0'=\bar{\mathcal{Y}}_{P(B,0)}$ and	 $k_1'=\underline{\mathcal{Y}}_{P(B,1)}$. Let 
	\begin{equation}
	\begin{split}
	C^{Z=1}=\int_{y} g_c^{1s}(y)d\mu_F(y)\quad and \quad	C^{Z=0}=\int_{y} g_c^{0s}(y)d\mu_F(y),\\
	\end{split}
	\end{equation}
	and define:
	\begin{equation}\label{eq: append, LATE, max extension, complier probability}
	\begin{split}
	&Pr_{G^s}(Y_i(d,z)\in B_{dz}\quad \forall d,z\in\{0,1\}, D_i(1)=1,D_i(0)=0|Z_i=z)\\
	&=\frac{\int_{B_{00}\cap B_{01}}g_c^{1s}(y) d\mu_F(y) \times \int_{B_{10}\cap B_{11}}g_c^{0s}(y) d\mu_F(y) }{C^{Z=1-z}},
	\end{split}
	\end{equation}
	\begin{equation}
	Pr_{G^s}(Y_i(d,z)\in B_{dz}\quad \forall d,z\in\{0,1\}, D_i(1)=0,D_i(0)=1|Z_i=z)\equiv 0.
	\end{equation}
	 We shall keep in mind that $g_c^{1s}$ and $g_c^{0s}$ are the densities of the potential outcomes conditional on the compliers. We will use $g_c^{1s}$ and $g_c^{0s}$ to calculate LATE.
	 
	 In this construction of $G^s$, we only allow the compliers to exist near the boundary of the support $(\underline{\mathcal{Y}}_{P(B,1)},\bar{\mathcal{Y}}_{P(B,1)})$ when $Z_i=1$ and  $Y_i(0,z)\in (\underline{\mathcal{Y}}_{Q(B,0)},\bar{\mathcal{Y}}_{Q(B,0)})$ when $Z_i=0$. The parameter $\epsilon$ control the closeness to the boundaries, and $\alpha,\beta$ control the share of mixture at the two boundaries. 
	
	\paragraph{Step 2. Check the constructed $G^s$ is in $\tilde{A}^{max}$.} We first show the constructed $G^s$ is a probability measure, since 
	\[
	\begin{split}
	&\quad \sum_{d_1,d_0\in\{0,1\}}Pr_{G^s}(Y_i(d,z)\in \mathcal{Y} \quad \forall d,z\in\{0,1\}, D_i(1)=d_1,D_i(0)=d_0|Z_i=1)\\
	&= P(\mathcal{Y},1)+P(\mathcal{Y},0)=1\\
	&\quad \sum_{d_1,d_0\in\{0,1\}}Pr_{G^s}(Y_i(d,z)\in \mathcal{Y} \quad \forall d,z\in\{0,1\}, D_i(1)=d_1,D_i(0)=d_0|Z_i=0)\\
	&= Q(\mathcal{Y},1)+Q(\mathcal{Y},0)=1.
	\end{split}
	\]
	By our construction of $G^s$, the `No Defiers' assumption holds. The exclusion restriction holds by the construction of $G^s$ and Lemma \ref{lem: auxiliary, criteria for exclusion restriction}. To keep the proof short, we omit the proof of $F\in M^s(G^s)$. The procedure is the same as equation (\ref{eq: append, LATE, mini defier, check F in Ms(Gs)}) in the proof of Lemma \ref{lem: minimal defiers construction for type ind instru}.
	
	\paragraph{Step 3. Find the LATE as a function of $(\alpha,\beta,\epsilon)$.}The LATE for this $G^s$ is 
	\begin{equation}\label{eq: appli, unbounded late, LATE(s)}
	\begin{split}
	LATE(s)&=\frac{\int yg_c^{1s}(y)d\mu_F(y)}{\int g_c^{1s}(y)d\mu_F(y)}- \frac{\int yg_c^{0s}(y)d\mu_F(y)}{\int g_c^{0s}(y)d\mu_F(y)}\\
	&=\frac{\alpha \int_{k_1-\epsilon}^{k_1}yp(y,1)/2 d\mu_F(y)+(1-\alpha) \int_{k^{\prime}_1}^{k^{\prime}_1+\epsilon}yp(y,1)/2 d\mu_F(y)}{\alpha \int_{k_1-\epsilon}^{k_1}p(y,1)/2 d\mu_F(y)+(1-\alpha) \int_{k^{\prime}_1}^{{k^{\prime}_1+\epsilon}}p(y,1)/2 d\mu_F(y)}\\
	&- \frac{\beta \int_{k_0}^{k_0+\epsilon}yq(y,0)/2 d\mu_F(y)+(1-\beta) \int_{k_0^\prime-\epsilon}^{k_0^\prime}yq(y,0)/2 d\mu_F(y)}{\beta \int_{k_0}^{k_0+\epsilon}q(y,0)/2 d\mu_F(y)+(1-\beta) \int_{k_0^\prime-\epsilon}^{k_0^\prime}q(y,0)/2 d\mu_F(y)}\\
	&\equiv \kappa(\alpha,\beta,\epsilon).
	\end{split}
	\end{equation}
	Since $F\in M^s(G^s)$, $\kappa(\alpha,\beta,\epsilon)\in LATE_{\tilde{A}^{max}}^{ID}(F)$ must hold. It is easy to see that $\kappa(1,0,\epsilon)\ge k_1-\epsilon-k_0'$ and $\kappa(0,1,\epsilon)\le k_1^\prime+\epsilon-k_0$. Since $\kappa$ is a continuous function ofn $(\alpha,\beta)$, by the intermediate value theorem, by continuously varying the values of $\alpha,\beta$, we can achieve all values of the identified LATE in the interval $[k_1^\prime+\epsilon-k_0,k_1-\epsilon-k_0']$. By taking $\epsilon\rightarrow 0$ we can see that by taking different combination of $(\alpha,\beta,\epsilon)$, $LATE(s)$ can achieve all value in $(k_1^\prime-k_0,k_1-k_0')$. This proves the form of the identified set in the Proposition. 

\end{proof}

\subsubsection{ Lemmas for Proposition \ref{prop: IA strong extension}}

\begin{lem}\label{lem: minimal defiers construction for type ind instru}
	Let $F$ be any distribution of outcome, and let $p(y,d),q(y,d)$ be the Radon-Nikodym derivatives with respect to $\mu_F$. Consider the following $G^s$ for all $d,z\in\{0,1\}$: let $Pr_{G^s}(Z_i=z)=Pr_F(Z_i=z)$, and
	\begin{equation}\label{eq: append, mini defier, AT construction}
	\begin{split}
	&Pr_{G^s}(Y_i(d,z)\in B_{dz},D_i(1)=1,D_i(0)=1|Z_i=1)= Pr_{G^s}(Y_i(d,z)\in B_{dz},D_i(1)=1,D_i(0)=1|Z_i=0)\\
	&= G^a(Y_i(0,0)\in B_{00}\cap B_{01})\times \int_{B_{11}\cap B_{10}}  (\min\{p(y,1),q(y,1)\})  d\mu_F(y),
	\end{split}
	\end{equation}
	where $G^a$ is any probability measure, and
	\begin{equation}\label{eq: append, mini defier, NT construction}
	\begin{split}
	&Pr_{G^s}(Y_i(d,z)\in B_{dz},D_i(1)=0,D_i(0)=0|Z_i=1)= Pr_{G^s}(Y_i(d,z)\in B_{dz},D_i(1)=0,D_i(0)=0|Z_i=0)\\
	&=G^n(Y_i(1,1)\in B_{11}\cap B_{10})\times \int_{B_{00}\cap B_{01}}	(\min\{p(y,0),q(y,0)\}) d\mu_F(y),
	\end{split}
	\end{equation}
	where $G^n$ is any probability measure. Let 
	\begin{equation}
	\begin{split}
	Pr_{G^s}(D_i(1)=1,D_i(0)=0|Z_i=1)=P(\mathcal{Y}_1,1)-Q(\mathcal{Y}_1,1),\\
	Pr_{G^s}(D_i(1)=1,D_i(0)=0|Z_i=0)=Q(\mathcal{Y}_0,0)-P(\mathcal{Y}_0,0),\\
	Pr_{G^s}(D_i(1)=0,D_i(0)=1|Z_i=0)=Q(\mathcal{Y}_1^c,1)-P(\mathcal{Y}_1^c,1),\\
	Pr_{G^s}(D_i(1)=0,D_i(0)=1|Z_i=1)=P(\mathcal{Y}_0^c,0)-Q(\mathcal{Y}_0^c,0),
	\end{split}
	\end{equation}
	be the $Z_i$-conditional probabilities of the compliers and defiers, and construct 
	\begin{equation}\label{eq: append, mini defier, CP construction}
	\begin{split}
	&Pr_{G^s}(Y_i(d,z)\in B_{dz}|D_i(1)=1,D_i(0)=0,Z_i=1)=Pr_{G^s}(Y_i(d,z)\in B_{dz}|D_i(1)=1,D_i(0)=0,Z_i=0)\\
	&=\frac{\int_{B_{00}\cap B{01}} \max\{q(y,0)-p(y,0),0\} d\mu_F(y)\times \int_{B_{11}\cap B_{10}} \max\{p(y,1)-q(y),0\} d\mu_F(y)}{(P(\mathcal{Y}_1,1)-Q(\mathcal{Y}_1,1))\times (Q(\mathcal{Y}_0,0)-P(\mathcal{Y}_0,0))} ,
	\end{split}
	\end{equation}
	and 
	\begin{equation}\label{eq: append, mini defier, DF construction}
	\begin{split}
	&Pr_{G^s}(Y_i(d,z)\in B_{dz}|D_i(1)=0,D_i(0)=1,Z_i=1)=Pr_{G^s}(Y_i(d,z)\in B_{dz}|D_i(1)=0,D_i(0)=1,Z_i=0)\\
	&= \frac{\int_{B_{00}\cap B{01}} \max\{p(y,0)-q(y,0),0\} d\mu_F(y)\times \int_{B_{11}\cap B_{10}} \max\{q(y,1)-p(y,1),0\} d\mu_F(y)}{(Q(\mathcal{Y}_1^c,1)-P(\mathcal{Y}_1^c,1))\times(P(\mathcal{Y}_0^c,0)-Q(\mathcal{Y}_0^c,0))} .
	\end{split}
	\end{equation}
	Then the constructed $G^s$ satisfies: (1). $G^s$ is a probability measure; (2). $G^s\in A^{TI}\cap A^{ER}\cap A^{EM-ATNT}$; (3). $F\in M^s(G^s)$; (4). $m^d(s)=m^{min}(F)$, with $m^{min}(F)$ defined in Assumption \ref{assumption:type ind instru}.
\end{lem}

\begin{proof}
	I first check that $G^s$ is a probability measure. Since the marginal distribution of $Z_i$ under $G^s$ coincide with the distribution of outcome $F$, it suffices to check the measure of $Y_{i}(d,z),D_i(1),D_i(0)$ is a probability measure conditional on $Z_i=1$ and $Z_i=0$. To do this, we have 
	\begin{equation*}
	\begin{split}
	&\quad \sum_{d_1,d_0\in\{0,1\}} Pr_{G^s}(Y_{dz}\in \mathcal{Y},D_i(1)=d_1,D_i(0)=d_0|Z_i=1)\\
	&=_{(1)} \int_{\mathcal{Y}}  \min\{p(y,1),q(y,1)\}  d\mu_F(y)+ \int_{\mathcal{Y}}	\min\{p(y,0),q(y,0)\} d\mu_F(y)\\
	&+ (P(\mathcal{Y}_1,1)-Q(\mathcal{Y}_1,1))+(P(\mathcal{Y}_0^c,0)-Q(\mathcal{Y}_0^c,0))\\
	&=_{(2)} (P(\mathcal{Y}_1^c,1)+Q(\mathcal{Y}_1,1))+ (P(\mathcal{Y}_0,0)+Q(\mathcal{Y}_0^c,0))+(P(\mathcal{Y}_1,1)-Q(\mathcal{Y}_1,1))+(P(\mathcal{Y}_0^c,0)-Q(\mathcal{Y}_0^c,0))\\
	&=P(\mathcal{Y},1)+P(\mathcal{Y},0)=1,
	\end{split}
	\end{equation*} 
	where equality (1) holds by construction of $G^s$, and equality (2) holds by the definition of $\mathcal{Y}_1$ and $\mathcal{Y}_0$. Similarly, for the measure conditional on $Z_i=0$, we have 
	\begin{equation*}
	\begin{split}
	&\quad \sum_{d_1,d_0\in\{0,1\}} Pr_{G^s}(Y_{dz}\in \mathcal{Y},D_i(1)=d_1,D_i(0)=d_0|Z_i=0)\\
	&= \int_{\mathcal{Y}}  \min\{p(y,1),q(y,1)\}  d\mu_F(y)+ \int_{\mathcal{Y}}	\min\{p(y,0),q(y,0)\} d\mu_F(y)\\
	&+ (Q(\mathcal{Y}_0,0)-P(\mathcal{Y}_0,0))+(Q(\mathcal{Y}_1^c,1)-Q(\mathcal{Y}_1^c,1))\\
	&=Q(\mathcal{Y},0)+Q(\mathcal{Y},1)=1.
	\end{split}
	\end{equation*} 
	This checks that $G^s$ is a probability measure.
	
	The type independence condition $s\in A^{TI}$ follows directly by the construction of $G^s$ in (\ref{eq: append, mini defier, AT construction})-(\ref{eq: append, mini defier, DF construction}). The property that measure of always taker and never taker is independent of $Z$ also follow directly from the construction of $G^s$ in (\ref{eq: append, mini defier, AT construction}) and (\ref{eq: append, mini defier, NT construction}).
	
	To show that exclusion restriction ($Y_i(1,1)=Y_i(0,1)$ and $Y_i(0,1)=Y_i(0,0)$ a.s.), I check the conditions in Lemma \ref{lem: auxiliary, criteria for exclusion restriction}. By construction,
	\begin{equation}\label{eq: LATE append, exclusion restriction}
	\begin{split}
	&Pr_{G^s}(Y_i(1,1)\in B_{11},Y_i(1,0)\in B_{10}|Z_i=1)=Pr_{G^s}(Y_i(1,1)\in B_{11},Y_i(1,0)\in B_{10}|Z_i=0)\\
	&= \int_{B_{11}\cap B_{10}}  (\min\{p(y,1),q(y,1)\})  d\mu_F(y)+ G^n(Y_i(1,1)\in B_{11}\cap B_{10})\\
	&+\int_{B_{11}\cap B_{10}} \max\{p(y,1)-q(y),0\} d\mu_F(y)+\int_{B_{11}\cap B_{10}} \max\{q(y,1)-p(y,1)\}d\mu_F(y),
	\end{split}
	\end{equation}
	where the right hand side of (\ref{eq: LATE append, exclusion restriction}) depends only on the set $B_{11}\cap B_{10}$. Therefore, by Lemma \ref{lem: auxiliary, criteria for exclusion restriction}, $Y_i(1,1)=Y_i(1,0)$ almost surely. Similarly, we can use Lemma \ref{lem: auxiliary, criteria for exclusion restriction} to check $Y_i(0,1)=Y_i(0,0)$ almost surely. As a result, the exclusion restriction holds.
	
	Then I check $G^s$ can generate the data distribution $F$, i.e. $F\in M^s(G^s)$. To do this, I check that the model-predicted observable distribution coincides with the observed data distribution. 
	\begin{equation}\label{eq: append, LATE, mini defier, check F in Ms(Gs)}
	\begin{split}
	\underbrace{Pr_{M^s(G^s)}(Y_i\in B,D_i=1|Z_i=1)}_{Model\,\, Predicted\,\, Outcome\,\, Distribution}&=_{(3)} \sum_{j=0}^1Pr_{G^s}(Y_i(1,1)\in B,Y_i(0,1)\in\mathcal{Y}, D_i(1)=1, D_i(0)=j|Z_i=1 )\\
	&=_{(4)}  [Q(B\cap \mathcal{Y}_1,1)+P(B\cap \mathcal{Y}_1^c,1)]+\left[P(B\cap\mathcal{Y}_1,1)-Q(B\cap\mathcal{Y}_1,1)\right]\\
	&= P(B,1)\\
	&=_{(5)}\underbrace{Pr_F(Y_i\in B, D_i=1| Z_i=1)}_{Observed\,\, Outcome\,\, Distribution},
	\end{split}
	\end{equation}
	where  equality (3) holds by the potential outcome framework (\ref{eq: potential outcome}), (4) holds by the construction of $G^s$, and  (5) holds by the definition of $P(B,1)$. Similarly, 
	\begin{equation*}
	\begin{split}
	\underbrace{Pr_{M^s(G^s)}(Y_i\in B,D_i=0|Z_i=1)}_{Model\,\, Predicted\,\, Outcome\,\, Distribution}&= \sum_{j=0}^1Pr_{G^s}(Y_i(0,1)\in B,Y_i(1,1)\in\mathcal{Y}, D_i(1)=0, D_i(0)=j|Z_i=1 )\\
	&=  [P(B\cap \mathcal{Y}_0,0)+Q(B\cap \mathcal{Y}_0^c,0)]+\left[P(B\cap\mathcal{Y}_0^c,0)-Q(B\cap\mathcal{Y}_0^c,0)\right]\\
	&= P(B,0)\\
	&=\underbrace{Pr_F(Y_i\in B, D_i=0| Z_i=1)}_{Observed\,\, Outcome\,\, Distribution}.
	\end{split}
	\end{equation*}
	Similar relations between $G^s$ and $Pr_F$ also hold when $Z_i=0$. This checks $F\in M^s(G^s)$.

	In the last step, I check that $G^s$ achieves the minimal probability of defiers. We first find a lower bound for $m^{min}(F)$, and show that $m^d(s)$ achieves this lower bound. 
	
	Consider any $s^*\in A^{ER}\cap A^{TI}\cap A^{EM-ATNT}$, and $F\in  M^{s^*}(G^{s^*})$. We have
	{ \footnotesize
	\begin{equation} \label{eq: identity for always takers }
	\begin{split}
	&Pr_{G^{s^*}} (Y_i(1,0)=Y_i(1,1)\in B_1, Y_i(0,0)=Y_i(0,1)\in \mathcal{Y},D_i(1)=1,D_i(0)=1|Z_i=1)\\
	&=_{(6)}Pr_{G^{s^*}} (Y_i(1,0)=Y_i(1,1)\in B_1, Y_i(0,0)=Y_i(0,1)\in \mathcal{Y}|D_i(1)=1,D_i(0)=1,Z_i=1)Pr(D_i(1)=1,D_i(0)=1|Z_i=1)\\
	&=_{(7)}Pr_{G^{s^*}} (Y_i(1,0)=Y_i(1,1)\in B_1, Y_i(0,0)=Y_i(0,1)\in \mathcal{Y}|D_i(1)=1,D_i(0)=1,Z_i=0)Pr(D_i(1)=1,D_i(0)=1|Z_i=0)\\
	&=Pr_{G^{s^*}} (Y_i(1,0)=Y_i(1,1)\in B_1, Y_i(0,0)=Y_i(0,1)\in \mathcal{Y},D_i(1)=1,D_i(0)=1|Z_i=0),
	\end{split}
	\end{equation}}
	where equality (6) holds by the Bayes' rule, (7) holds by $s^*\in A^{ER}\cap A^{TI}\cap A^{EM-ATNT}$. Now we consider the following decompositions:
	\begin{equation}
	\begin{split}
	P( B_1,1)&=Pr_{G^{s^*}} (Y_i(1,0)=Y_i(1,1)\in B_1, Y_i(0,0)=Y_i(0,1)\in \mathcal{Y},D_i(1)=1,D_i(0)=1|Z_i=1)\\
	&+Pr_{G^{s^*}} (Y_i(1,0)=Y_i(1,1)\in B_1, Y_i(0,0)=Y_i(0,1)\in \mathcal{Y},D_i(1)=1,D_i(0)=0|Z_i=1),\\
	Q( B_1,1)&=Pr_{G^{s^*}} (Y_i(1,0)=Y_i(1,1)\in B_1, Y_i(0,0)=Y_i(0,1)\in \mathcal{Y},D_i(1)=1,D_i(0)=1|Z_i=0)\\
	&+Pr_{G^{s^*}} (Y_i(1,0)=Y_i(1,1)\in B_1, Y_i(0,0)=Y_i(0,1)\in \mathcal{Y},D_i(1)=0,D_i(0)=1|Z_i=0),
	\end{split}
	\end{equation}  
	and use (\ref{eq: identity for always takers }) to get:
	\begin{equation} \label{eq: lem type ind const, p-q}
	\begin{split}
	P(B_1,1)-Q(B_1,1)&= Pr_{G^{s^*}} (Y_i(1,0)=Y_i(1,1)\in B_1, Y_i(0,0)=Y_i(0,1)\in \mathcal{Y},D_i(1)=1,D_i(0)=0|Z_i=1)\\
	&-Pr_{G^{s^*}} (Y_i(1,0)=Y_i(1,1)\in B_1, Y_i(0,0)=Y_i(0,1)\in \mathcal{Y},D_i(1)=0,D_i(0)=1|Z_i=0).
	\end{split}
	\end{equation}
	Take $B_1=\mathcal{Y}_1^c$, we have 
	{\footnotesize
	\begin{equation}\label{eq: lem type ind identity for defiers}
	\begin{split}
	&Pr(D_i(1)=0,D_i(0)=1|Z_i=0)\\
	&\ge Pr_{G^{s^*}} (Y_i(1,0)=Y_i(1,1)\in \mathcal{Y}_1^c, Y_i(0,0)=Y_i(0,1)\in \mathcal{Y},D_i(1)=0,D_i(0)=1|Z_i=0)\\&=Pr_{G^{s^*}} (Y_i(1,0)=Y_i(1,1)\in \mathcal{Y}_1^c, Y_i(0,0)=Y_i(0,1)\in \mathcal{Y},D_i(1)=1,D_i(0)=0|Z_i=1)+Q(\mathcal{Y}_1^c,1)-P(\mathcal{Y}_1^c,1)\\
	&\ge Q(\mathcal{Y}_1^c,1)-P(\mathcal{Y}_1^c,1),
	\end{split}
	\end{equation}}
	and similarly, we can show $Pr(D_i(1)=0,D_i(0)=1|Z_i=1)\ge P(\mathcal{Y}_0^c,0)-Q(\mathcal{Y}_0^c,0)$. So the total measure of defiers satisfies
	\[
	\begin{split}
	m^d(s^*)&\ge [Q(\mathcal{Y}_1^c,1)-P(\mathcal{Y}_1^c,1)]Pr(Z_i=0)+[P(\mathcal{Y}_0^c,0)-Q(\mathcal{Y}_0^c,0)]Pr(Z_i=1).
	\end{split}
	\]
	Therefore, we get a lower bound $m^{min}(F)\ge [Q(\mathcal{Y}_1^c,1)-P(\mathcal{Y}_1^c,1)]Pr(Z_i=0)+[P(\mathcal{Y}_0^c,0)-Q(\mathcal{Y}_0^c,0)]Pr(Z_i=1)$. 
	
	On the other hand, the by the construction of $G^s$, the measure of defiers under $G^s$ is 
	\[m^d(s)= [Q(\mathcal{Y}_1^c,1)-P(\mathcal{Y}_1^c,1)]Pr(Z_i=0)
	+[P(\mathcal{Y}_0^c,0)-Q(\mathcal{Y}_0^c,0)]Pr(Z_i=1).\]
	So the constructed $s$ achieves the minimal measure of defiers. 
\end{proof}

\subsubsection{Proof of Proposition \ref{prop: IA strong extension}}
\begin{proof}
	By Lemma \ref{lem: minimal defiers construction for type ind instru}, for any $F\in \mathcal{F}$, we can construct a $G^s$ that achieves the minimal measure of defiers. Therefore, $m^d$ is a well-behaved relaxation measure with respect to $A^{ER},A^{TI},A^{EM-NTAT}$. So by Proposition \ref{prop: minimal deviation as strong extension}, the minimal defier assumption with type independent instrument is a strong extension because $m^d(s)=0$ if and only if $D_i(1)\ge D_i(0)$ almost surely. 
\end{proof}

\subsection{Proofs in Section 3.2}

\subsubsection{Proof of Proposition \ref{prop: identified under minimal defiers type ind inst}}
\begin{proof}
	Because $P(B,d)$ and $Q(B,d)$ are absolutely continuous with respect to some denominating measure $\mu_F$, equation (\ref{eq: appli, M^s}) implies that  $Pr_{G^s}(Y_i(d,z)\in B,D_i(z)=d,D_i(1-z)=d'|Z_i=z)$ are also absolutely continuous with respect to $\mu_F$ for all $d,z,d'\in \{0,1\}$. So we will use the Radon-Nikodym density of $G^s$ with respect to $\mu_F$ throughout this proof: We use $g^s_{y_{dz}}(y,d,d'|Z_i)$ to denote the density of $G^s(Y_i(d,z)\in B_{dz},D_i(1)=d,D_i(0)=0|Z_i)$.
	
	Note that the exclusion restriction holds under Assumption \ref{assumption:type ind instru}. Given a structure $s$, the conditional density of $Y_i(1,1)$ given $D_i(1)-D_i(0)=1$ and $Z_i=1$ is 
	\[
	\frac{ g_{y_{11}}^s(y,1,0|Z_i=1)d\mu_F(y)}{\int_{\mathcal{Y}}  g_{y_{11}}^s(y,1,0|Z_i=1)d\mu_F(y)}.
	\]
	Similarly, conditional density of $Y_i(0,0)$ given $D_i(1)-D_i(0)=1$ and $Z_i=0$ is 
	\[
	\frac{ g_{y_{00}}^s(y,1,0|Z_i=0)d\mu_F(y)}{\int_{\mathcal{Y}}  g_{y_{00}}^s(y,1,0|Z_i=0)d\mu_F(y)}.
	\]
	So for any structure $s$, such that $F\in M^s(G^s)$, we have
	\[
	LATE(s)=\frac{\int_{\mathcal{Y}} y g_{y_{11}}^s(y,1,0|Z_i=1)d\mu_F(y)}{\int_{\mathcal{Y}}  g_{y_{11}}^s(y,1,0|Z_i=1)d\mu_F(y)}-\frac{\int_{\mathcal{Y}} y g_{y_{00}}^s(y,1,0|Z_i=0)d\mu_F(y)}{\int_{\mathcal{Y}}  g_{y_{00}}^s(y,1,0|Z_i=0)d\mu_F(y)}\in LATE^{ID}(F),
	\]
	where the denominator $\int_{\mathcal{Y}}  g_{y_{zz}}^s(y,1,0|Z_i=z)d\mu_F(y)$ represent the measure of compliers conditioned on $Z_i=z$. To show the identified LATE satisfies expression (\ref{eq: appli, formular}), it suffices to show 
	\[
	\begin{split}
	&g_{y_{11}}^s(y,1,0|Z_i=1)=\max\{p(y,1)-q(y,1),0\}\\
	and\quad &g_{y_{00}}^s(y,1,0|Z_i=0)=\max\{q(y,0)-p(y,0),0\}
	\end{split}
	\]
 	hold for all $s\in \tilde{A}$.
	
	First, $\tilde{A}$ satisfies Assumption \ref{assumption:type ind instru}, and  Lemma \ref{lem: minimal defiers construction for type ind instru} implies the minimal measure of defiers:
	\[
	\begin{split}
	m^{min}(F)=[Q(\mathcal{Y}_1^c,1)-P(\mathcal{Y}_1^c,1)]Pr(Z_i=0)+[P(\mathcal{Y}_0^c,0)-Q(\mathcal{Y}_0^c,0)]Pr(Z_i=1).
	\end{split}
	\]
	Now we look at the equation (\ref{eq: lem type ind identity for defiers}) from Lemma \ref{lem: minimal defiers construction for type ind instru}:
	\begin{equation*}
	\begin{split}
	&Pr(D_i(1)=0,D_i(0)=1|Z_i=0)\\
	&=Pr_{G^{s}} (Y_i(1,0)\in \mathcal{Y}_1^c, Y_i(0,0)\in \mathcal{Y},D_i(1)=0,D_i(0)=1|Z_i=0)\\&+Pr_{G^{s}} (Y_i(1,0)\in \mathcal{Y}_1, Y_i(0,0)\in \mathcal{Y},D_i(1)=0,D_i(0)=1|Z_i=0)\\
	&\ge_{(1)} Pr_{G^{s}} (Y_i(1,0)\in \mathcal{Y}_1^c, Y_i(0,0)\in \mathcal{Y},D_i(1)=0,D_i(0)=1|Z_i=0)\\
	&=_{(2)}Pr_{G^{s}} (Y_i(1,1)\in \mathcal{Y}_1^c, Y_i(0,1)\in \mathcal{Y},D_i(1)=1,D_i(0)=0|Z_i=1)+Q(\mathcal{Y}_1^c,1)-P(\mathcal{Y}_1^c,1)\\
	&\ge_{(3)} Q(\mathcal{Y}_1^c,1)-P(\mathcal{Y}_1^c,1),
	\end{split}
	\end{equation*}
	where the first inequality (1) holds with equality if and only if $g_{y_{10}}^s(y,0,1|Z=0)=0$ for all $y\in \mathcal{Y}_1$, equality (2) holds by (\ref{eq: lem type ind identity for defiers}), and  inequality (3) holds with equality if and only if $g_{y_{11}}^s(y,1,0|Z_i=1)=0$ for all $y\in \mathcal{Y}_1^c$. Similarly, we can write the condition for $Z_i=1$:
	\begin{equation*}
	\begin{split}
	&Pr(D_i(1)=0,D_i(0)=1|Z_i=1)\\
	&=Pr_{G^{s}} (Y_i(0,1)\in \mathcal{Y}_0^c, Y_i(1,1)\in \mathcal{Y},D_i(1)=0,D_i(0)=1|Z_i=1)\\&+Pr_{G^{s}} (Y_i(0,1)\in \mathcal{Y}_0, Y_i(1,1)\in \mathcal{Y},D_i(1)=0,D_i(0)=1|Z_i=1)\\
	&\ge Pr_{G^{s}} (Y_i(0,1)\in \mathcal{Y}_0^c, Y_i(1,1)\in \mathcal{Y},D_i(1)=0,D_i(0)=1|Z_i=1)\\&=Pr_{G^{s}} (Y_i(01)\in \mathcal{Y}_0^c, Y_i(1,1)\in \mathcal{Y},D_i(1)=1,D_i(0)=0|Z_i=0)+P(\mathcal{Y}_0^c,0)-Q(\mathcal{Y}_0^c,0)\\
	&\ge P(\mathcal{Y}_0^c,0)-Q(\mathcal{Y}_0^c,0),
	\end{split}
	\end{equation*}
	where the first inequality with equality holds if and only if $g_{y_{01}}^s(y,0,1|Z=1)=0$ for all $y\in \mathcal{Y}_0$ and the last inequality holds with equality if and only if $g_{y_{00}}^s(y,1,0|Z_i=0)=0$ for all $y\in \mathcal{Y}_0^c$.
	Therefore, $m^d(s)=m^{min}(F)$ if and only if the density conditions: 
	\begin{equation} \label{eq: appen, density constraint, min defiers, typ ind}
	\begin{split}
	g_{y_{10}}^s(y,0,1|Z_i=0)=0 \quad \forall y\in \mathcal{Y}_1\\
	g_{y_{11}}^s(y,1,0|Z_i=1)=0 \quad \forall y\in \mathcal{Y}_1^c\\
	g_{y_{01}}^s(y,0,1|Z_i=1)=0 \quad \forall y\in \mathcal{Y}_0\\
	g_{y_{00}}^s(y,1,0|Z_i=0)=0 \quad \forall y\in \mathcal{Y}_0^c\\
	\end{split}
	\end{equation}
	hold. Now, take Radon-Nikodym derivatives of (\ref{eq: lem type ind const, p-q}) with respect to $\mu_F$, we have 
	\[
	p(y,1)-q(y,1)= g_{y_{11}}^s(y,1,0|Z_i=1)-g_{y_{10}}^s(y,0,1|Z=0).
	\] 
	Combine the expression of $p(y,1)-q(y,1)$ and equation (\ref{eq: appen, density constraint, min defiers, typ ind}),  we have  $g_{y_{11}}^s(y,1,0|Z=1)=\max\{p(y,1)-q(y,1),0\}$ must hold for all $s\in \tilde{A}$. We can symmetrically get $g_{y_{00}}(y,0,1|Z=0)=\max\{q(y,0)-p(y,0),0\}$ must hold for all $s\in \tilde{A}$. \end{proof}

\subsubsection{Proof of Proposition \ref{prop: continuity of estimator in different cases}}
\begin{proof}
	To show that ${LATE}^{ID}_{\tilde{A}_1}(F)$ is not upper hemicontinuous with respect to $||\cdot ||_{1,\infty}$, let's consider an $F_0$ such that the corresponding densities $p^0(y,d)$ and $q^0(y,d)$ satisfy $(-1)^{1-d}(p^0(y,d)-q^0(y,d))\ge 0$ for all $y\in\mathcal{Y},d\in\{0,1\}$; Moreover, suppose (1). there exists a set $\mathcal{Y}^{zero}_1$ set such that $p^0(y,1)-q^0(y,1)=0$ for all $y\in \mathcal{Y}^{zero}_1$ and $\mu_F(\mathcal{Y}^{zero}_1)>0$; (2). there exists a constant $c>0$ such that the density $p^0(y,1)>c$ for all $y\in \mathcal{Y}^{zero}_1$.

	Then $F_0$ satisfies the IA-M assumption by Theorem \ref{thm: testable implication in density form}, but the density condition (\ref{eq: testable implication in density form}) is binding on the positively measured set $ \mathcal{Y}^{zero}_1$. For any $\epsilon>0$, we consider an $F^{\epsilon}$ such that the corresponding $p^{\epsilon}(y,d)$ and $q^{\epsilon}(y,d)$ satisfy:
	\begin{enumerate}
		\item $q^{\epsilon}(y,d)=q^0(y,d)$, $p^{\epsilon}(y,0)=p^0(y,0)$ for all $y\in\mathcal{Y},d\in\{0,1\}$;
		\item We find two sets $\mathcal{Y}^{sub,1}$ and $\mathcal{Y}^{sub,2}$ such that $\mathcal{Y}^{sub,j}\subseteq\mathcal{Y}^{zero}_1$ for $j=1,2$ and $\mu_F( \mathcal{Y}^{sub,1})=\mu_F(\mathcal{Y}^{sub,2})>0$; 
		\item $p^{\epsilon}(y,1)= p^0(y,1)-\min\{\epsilon,c/2\}$ for all $y\in \mathcal{Y}^{sub,1}$, and $p^{\epsilon}(y,1)= p^0(y,1)+\min\{\epsilon,c/2\}$ for all $y\in \mathcal{Y}^{sub,2}$, and $p^{\epsilon}(y,1)= p(y,1)$ for all $y\in \mathcal{Y}\backslash(\mathcal{Y}^{sub,1}\cup \mathcal{Y}^{sub,2} )$. 
	\end{enumerate}
	Then $F^{\epsilon}\in \mathcal{F}$ is a probability measure, and $||F_0-F^{\epsilon}||_{1,\infty}\le \epsilon$. However, on the set  $\mathcal{Y}^{sub,1}$, $F^{\epsilon}$ fails the density constraint (\ref{eq: testable implication in density form}) in Theorem \ref{thm: testable implication in density form}. By Proposition \ref{prop: appli, maximal extension arbitrary instrument}, \[LATE^{ID}_{\tilde{A}_1}(F^\epsilon)=\left(\underline{\mathcal{Y}}_{Q(B,0)}-\bar{\mathcal{Y}}_{P(B,1)},-\underline{\mathcal{Y}}_{Q(B,0)}+\bar{\mathcal{Y}}_{P(B,1)}\right).\] On the other hand, $LATE^{ID}_{\tilde{A}_1}(F_0)$ is a singleton. Since $\epsilon$ is arbitrarily small, $LATE^{ID}_{\tilde{A}_1}(F)$ is not upper hemicontinuous at $F=F_0$. 
	
	On the other hand, recall that under $\tilde{A}_2$, the point-identified LATE \begin{equation}
	{LATE}_{\tilde{A}_2}^{ID}(F)=\frac{\int_{\mathcal{Y}_1}{y (p(y,1)-q(y,1))}d\mu_F(y)}{P(\mathcal{Y}_1,1)-Q(\mathcal{Y}_1,1)}
	-\frac{\int_{\mathcal{Y}_0}{y (q(y,0)-p(y,0))}d\mu_F(y)}{Q(\mathcal{Y}_0,0)-P(\mathcal{Y}_0,0)}
	\end{equation}
	is the difference between two fractions. 	
	To show ${LATE}_{\tilde{A}_2}^{ID}(F)$ is continuous at $F_0$ with respect to $||\cdot||_{1,\infty}$, it suffices to show that all numerators and denominators in the fractions are continuous at $F_0$. Without loss of generality, I assume $\mu_F$ is a finite measure. \footnote{Note that for any measure $\mu(y)$ on $\mathcal{Y}$, we can construct a finite measure $\tilde{\mu}(y)$ from a positive and integrable function $\zeta(y)$ such that $\zeta(y)>0$ holds $\mu(y)$ almost surely, and $\tilde{\mu}(B)=\int_B \zeta(y) d\mu(y)$ for all measurable set $B$. Then $\tilde{\mu}$ is a finite measure. Moreover, $\mu$ is absolute continuous with respect to $\tilde{\mu}$. } Note that we can write 
	\[
	\begin{split}
	\int_{\mathcal{Y}_1}{y (p(y,1)-q(y,1))}d\mu_F(y)&= \int_{\mathcal{Y}} y\max\{p(y,1)-q(y,1),0\}d\mu_F(y),\\
	P(\mathcal{Y}_1,1)-Q(\mathcal{Y}_1,1)=\int_{\mathcal{Y}_1}{(p(y,1)-q(y,1))}d\mu_F(y)&= \int_{\mathcal{Y}} \max\{p(y,1)-q(y,1),0\}d\mu_F(y).
	\end{split}
	\]
	Therefore, for any $F^{\eta}$ such that $||F_0-F^{\eta}||_{1,\infty}<\epsilon$:
	\[
	\begin{split}
	&\quad \int_{\mathcal{Y}} y\left(\max\{p^0(y,1)-q^0(y,1),0\}-\max\{p^\eta(y,1)-q^\eta(y,1),0\}\right) d\mu_F(y)\\
	&\le 2\epsilon \int_{\mathcal{Y}} y d\mu_F(y) \le 2\epsilon\max\{|M_s^u|,|M_s^l|\}\mu_F(\mathcal{Y}),
	\end{split}
	\]
	which shows the numerator $\int_{\mathcal{Y}_1}{y (p(y,1)-q(y,1))}d\mu_F(y)$ is continuous in $F$. We also have 
	\[
	\begin{split}
	&\quad \int_{\mathcal{Y}} \left(\max\{p^0(y,1)-q^0(y,1),0\}-\max\{p^\eta(y,1)-q^\eta(y,1),0\}\right) d\mu_F(y)\\
	&\le 2\epsilon  d\mu_F(y) \le 2\epsilon\mu_F(\mathcal{Y}),
	\end{split}
	\]
	which shows the denominator $P(\mathcal{Y}_1,1)-Q(\mathcal{Y}_1,1)$ is continuous in $F$. We can show the other two terms corresponding to $\mathcal{Y}_0$ are continuous in $F$ by similar arguments. The continuity of $LATE_{\tilde{A}_2}^{ID}(F)$ follows by observing that $f(x_1,x_2,x_3,x_4)=\frac{x_1}{x_2}-\frac{x_3}{x_4}$ is continuous whenever $x_3,x_4\ne 0$. \end{proof}

\subsection{Proofs in Section 3.3}
\subsubsection{Lemmas for Theorem \ref{thm: asymptotic property of LATE}}
We  first define the following objects:
	\begin{equation}
	\begin{split}
		f^{l,m}_n(y)&\equiv \frac{1}{nh_n}\sum_{i=1}^n K\left(\frac{Y_i-y}{h_n}\right)\mathbbm{1}(D_i=l,Z_i=m),\\
		\bar{f}^{l,m}_n(y)&\equiv\frac{1}{h_n} E\left[K\left(\frac{Y_i-y}{h_n}\right)\mathbbm{1}(D_i=l,Z_i=m)\right ].
	\end{split}
\end{equation}
\begin{lem}\label{lem: convergence rate of density estimator 1}
	Let $h_n=n^{-\gamma}$ for some $\gamma\in (0,1)$, such that  $\frac{nh_n}{|\log h_n|}\rightarrow \infty$. Define $a_n=\min\{\sqrt{\frac{nh_n}{\log h_n^{-1}}}, h_n^{-2}\}$. Suppose assumptions \ref{assumption: kernel } and \ref{assumption: density} hold,  then there exists a constant $C$ such that for all $d,z\in\{0,1\}$, we have 
	\begin{equation}
	\begin{split}
	&\lim\sup_{n\rightarrow \infty } a_n \sup_{y}|\frac{1}{nh_n}\sum_{i=1}^n K\left(\frac{Y_i-y}{h_n}\right)\mathbbm{1}(D_i=d,Z_i=z)- p(y,1)Pr(Z_i=z)|\le C \quad a.s..
	\end{split}
	\end{equation}
\end{lem}

\begin{proof}
	I prove the inequality for $d=z=1$ and the rest inequalities follow similarly. Using the notation of $f_n^{1,1}$ and $\bar{f}_n^{1,1}$, by the triangular inequality 
	\begin{equation*}
	\begin{split}
	a_n \sup_y |f_n^{1,1}(y)-p(y,1)Pr(Z_i=1)|&\le a_n \sup_y |\bar{f}_n^{1,1}(y)-p(y,1)Pr(Z_i=1)| +a_n|f_n^{1,1}(y)-\bar{f}_n^{1,1}(y)|\\
	&\le a_n \sup_y |\bar{f}_n^{1,1}(y)-p(y,1)Pr(Z_i=1)|+ \sqrt{\frac{nh_n}{\log h_n^{-1}}}|f_n^{1,1}(y)-\bar{f}_n^{1,1}(y)|.\\
	\end{split}
	\end{equation*}
	The first term  $a_n \sup_y |\bar{f}_n^{1,1}(y)-p(y,1)Pr(Z_i=1)|$ is the bias term and can be bounded as: 
	\begin{equation}
	\begin{split}
	&\quad a_n \sup_y |\bar{f}_n^{1,1}(y)-p(y,1)Pr(Z_i=1)|\\
	&=a_n\left|\frac{1}{h_n}\int_t K(\frac{t-y}{h_n}) f(t|D_i=1,Z_i=1)-f(y|D_i=1,Z_i=1)dt\right| Pr(D_i=1,Z_i=1)\\
	&=a_n\left|\int_u K(u)[f(y+uh_n|D_i=1,Z_i=1)-f(y|D_i=1,Z_i=1)]dt\right| Pr(D_i=1,Z_i=1)\\
	&= a_n\left|\int_u K(u)uh_nf'(y|D_i=1,Z_i=1)+ K(u)u^2h_n^2f^{\prime\prime}(y|D_i=1,Z_i=1)+o(h_n^2)\right| Pr(D_i=1,Z_i=1)\\
	&\le h_n^{-2} \left|\int_u K(u)u^2h_n^2f^{\prime\prime}(y|D_i=1,Z_i=1)+o(h_n^2)\right| Pr(D_i=1,Z_i=1)\\
	&\le \left|\int_u K(u)u^2f^{\prime\prime}(y|D_i=1,Z_i=1)\right| Pr(D_i=1,Z_i=1)+o(1).
	\end{split}
	\end{equation}
	Let $\bar{C}$ be the constant in Lemma \ref{lem: strong convergence from Gine and Guillou}, and set set $C=\bar{C}+\left|\int_u K(u)u^2f^{\prime\prime}(y|D_i=1,Z_i=1)\right|$. The result follows.
\end{proof}

\begin{lem}\label{lem: convergence rate of density estimator 2}
	Let $h_n=n^{-\gamma}$ for some $\gamma\in (0,1)$, such that  $\frac{nh_n}{|\log h_n|}\rightarrow \infty$ and let $a_n=\min\{\sqrt{\frac{nh_n}{\log h_n^{-1}}}, h_n^{-2}\}$. If there exists a constant $c>0$ such that   $Pr(Z_i=1)\in [c,1-c]$, and $\sup_y[\max(p(y,d),q(y,d)]<\infty$, then for any $\epsilon>0$ such that 
	\begin{equation}
	\begin{split}
	n^{-\epsilon}a_n\sup_{y} |f_h(y,1)-(p(y,1)-q(y,1))|= o_p(1)\\
	n^{-\epsilon}a_n\sup_{y} |f_h(y,0)-(q(y,0)-p(y,0))|= o_p(1).
	\end{split}
	\end{equation}
\end{lem}
\begin{proof}
	Note that 
	\begin{equation}\label{eq: decomposition of f_h}
	\begin{split}
	|f_h(y,1)-(p(y,1)-q(y,1))|&\le \sup_y \left|\frac{f^{1,1}_n(y)}{\frac{1}{n}\sum_{i=1}^n \mathbbm{1}(Z_i=1)}-\frac{p(y,1)Pr(Z_i=1)}{Pr(Z_i=1)}\right| \\
	&+ \sup_y \left|\frac{f^{1,0}_n(y)}{\frac{1}{n}\sum_{i=1}^n \mathbbm{1}(Z_i=1)}-\frac{q(y,1)Pr(Z_i=1)}{Pr(Z_i=1)}\right|.
	\end{split}
	\end{equation}
	For the first term in (\ref{eq: decomposition of f_h}), we have 
	\begin{equation*}
	\begin{split}
	&n^{-\epsilon}a_n\sup_y \left|\frac{f^{1,1}_n(y)}{\frac{1}{n}\sum_{i=1}^n \mathbbm{1}(Z_i=1)}-\frac{p(y,1)Pr(Z_i=1)}{Pr(Z_i=1)}\right|\\
	&\le  n^{-\epsilon}a_n\left|\frac{f^{1,1}_n(y)}{\frac{1}{n}\sum_{i=1}^n \mathbbm{1}(Z_i=1)}- \frac{f^{1,1}_n(y)}{Pr(Z_i=1)}\right| +n^{-\epsilon}a_n\left|\frac{f^{1,1}_n(y)}{Pr(Z_i=1)}-\frac{p(y,1)}{Pr(Z_i=1)}\right|\\
	&=_{(*)}a_n n^{-\epsilon} \sup_y [p(y,1)+o_p(1)]\times\left[\frac{1}{1/n\sum_{i=1}^n\mathbbm{1}(Z_i=1)}-\frac{1}{Pr(Z_i=1)}\right]+ \frac{n^{-\epsilon}a_n\sup_y |f_n^{1,1}(y)-p(y,1)Pr(Z_i=1)|}{Pr(Z_i=1)}   \\
	&=_{\star} \sup_y[p(y,1)+o(1)]n^{-\epsilon}a_nO_p(\frac{1}{\sqrt{n}})+ \frac{1}{Pr(Z_i=1)} n^{-\epsilon}a_n\sup_y |f_n^{1,1}(y)-p(y,1)Pr(Z_i=1)|\\
	&\le \sup_y p(y,1) O_p\left(\sqrt{\frac{n^{-2\epsilon}h_n}{\log h_n^{-1}}}\right)+ O(n^{-\epsilon})\\
	&=o_p(1),
	\end{split}
	\end{equation*}
	where equality $(*)$ follows by Lemma \ref{lem: convergence rate of density estimator 1},  equality $\star$ follows by $|\sum_{i=1}^n \mathbbm{1}(Z_i=1)-Pr(Z_i=1)|=O_p(1/\sqrt{n})$ and continuous mapping holds when $Pr(Z_i=1)>0$. By the same argument, the second term in (\ref{eq: decomposition of f_h}) is also $o_p(1)$. The result follows. 
\end{proof}


\begin{lem}\label{lem: Limit Distribution of Infeasible Components}
	(Limit Distribution of Infeasible Components) Recall $f(y,1)=p(y,1)-q(y,1)$ and $f(y,0)=q(y,0)-p(y,0)$. Suppose $E[||Y_i||^{2+\delta}]<\infty$ for some $\delta>0$. Define the infeasible trimming set 
	\[
	\mathcal{Y}_d^{infsb}(b_n)=\{y\in\mathcal{Y}: f(y,d)\ge b_n\quad y\in[M_l,M_u]\}\cup \mathcal{Y}_d^{ut}\cup \mathcal{Y}_d^{lt}.
	\]
	Let $\bm{X}_i(b_n)$ and $\Sigma$ be 
	{\footnotesize
	\begin{equation}
	\bm{X}_i(b_n) = \begin{pmatrix}
	\mathbbm{1}(Z_i=0)\\
	\mathbbm{1}(Z_i=1)\\
	Y_i \mathbbm{1}(D_i=1,Z_i=1)\mathbbm{1}(Y_i\in \mathcal{Y}_1^{infsb}(b_n))\\
	Y_i \mathbbm{1}(D_i=1,Z_i=0)\mathbbm{1}(Y_i\in \mathcal{Y}_1^{infsb}(b_n))\\
	Y_i \mathbbm{1}(D_i=0,Z_i=0)\mathbbm{1}(Y_i\in \mathcal{Y}_0^{infsb}(b_n))\\
	Y_i \mathbbm{1}(D_i=0,Z_i=1)\mathbbm{1}(Y_i\in \mathcal{Y}_0^{infsb}(b_n))\\
	\mathbbm{1}(D_i=1,Z_i=1)\mathbbm{1}(Y_i\in \mathcal{Y}_1^{infsb}(b_n))\\
	\mathbbm{1}(D_i=1,Z_i=0)\mathbbm{1}(Y_i\in \mathcal{Y}_1^{infsb}(b_n))\\
	\mathbbm{1}(D_i=0,Z_i=0)\mathbbm{1}(Y_i\in \mathcal{Y}_0^{infsb}(b_n))\\
	\mathbbm{1}(D_i=0,Z_i=1)\mathbbm{1}(Y_i\in \mathcal{Y}_0^{infsb}(b_n))\\
	\end{pmatrix}\quad \quad and\quad\quad 
	\Sigma=Var\begin{pmatrix}
	\mathbbm{1}(Z_i=0)\\
	\mathbbm{1}(Z_i=1)\\
	Y_i \mathbbm{1}(D_i=1,Z_i=1)\mathbbm{1}(Y_i\in \mathcal{Y}_1)\\
	Y_i \mathbbm{1}(D_i=1,Z_i=0)\mathbbm{1}(Y_i\in \mathcal{Y}_1)\\
	Y_i \mathbbm{1}(D_i=0,Z_i=0)\mathbbm{1}(Y_i\in \mathcal{Y}_0)\\
	Y_i \mathbbm{1}(D_i=0,Z_i=1)\mathbbm{1}(Y_i\in \mathcal{Y}_0)\\
	\mathbbm{1}(D_i=1,Z_i=1)\mathbbm{1}(Y_i\in \mathcal{Y}_1)\\
	\mathbbm{1}(D_i=1,Z_i=0)\mathbbm{1}(Y_i\in \mathcal{Y}_1)\\
	\mathbbm{1}(D_i=0,Z_i=0)\mathbbm{1}(Y_i\in \mathcal{Y}_0)\\
	\mathbbm{1}(D_i=0,Z_i=1)\mathbbm{1}(Y_i\in \mathcal{Y}_0)\end{pmatrix}.
	\end{equation}	
	}
	Then for any $b_n\rightarrow 0$, $\frac{1}{\sqrt{n}}\sum_{i=1}^n
(\bm{X}_i(b_n)-E[\bm{X}_i(b_n)])\rightarrow_d N(0,\Sigma)$ where
\end{lem}
\begin{proof}
	 The sequence $\{\bm{X}_i(b_n)-E[\bm{X}_i(b_n)]\}_{i=1}^n$ forms a triangular array and since we assume $E[||Y_i||^{2+\delta}]<\infty$ holds, we have 
	 \[
	 E[||\bm{X}_{i}(b_n)||^{2+\delta}_2]\le \max\{E[Y_i^{2+\delta}],1\}<\infty.
	 \] and the variance of $\bm{X}_{i}(b_n)$ is 
	\[
	\begin{split}
	Var(\bm{X}_{i}(b_n))&=Var\begin{pmatrix}
\mathbbm{1}(Z_i=0)\\
	\mathbbm{1}(Z_i=1)\\
	Y_i \mathbbm{1}(D_i=1,Z_i=1)\mathbbm{1}(Y_i\in \mathcal{Y}_1^{infsb}(b_n))\\
	Y_i \mathbbm{1}(D_i=1,Z_i=0)\mathbbm{1}(Y_i\in \mathcal{Y}_1^{infsb}(b_n))\\
	Y_i \mathbbm{1}(D_i=0,Z_i=0)\mathbbm{1}(Y_i\in \mathcal{Y}_0^{infsb}(b_n))\\
	Y_i \mathbbm{1}(D_i=0,Z_i=1)\mathbbm{1}(Y_i\in \mathcal{Y}_0^{infsb}(b_n))\\
	\mathbbm{1}(D_i=1,Z_i=1)\mathbbm{1}(Y_i\in \mathcal{Y}_1^{infsb}(b_n))\\
	\mathbbm{1}(D_i=1,Z_i=0)\mathbbm{1}(Y_i\in \mathcal{Y}_1^{infsb}(b_n))\\
	\mathbbm{1}(D_i=0,Z_i=0)\mathbbm{1}(Y_i\in \mathcal{Y}_0^{infsb}(b_n))\\
	\mathbbm{1}(D_i=0,Z_i=1)\mathbbm{1}(Y_i\in \mathcal{Y}_0^{infsb}(b_n))\end{pmatrix}\rightarrow_p \Sigma
	\end{split}
	\]
	where the convergence holds by dominated convergence theorem: since $\mathbbm{1}(f(y,d)\ge b_n)\rightarrow \mathbbm{1}(f(y,d)\ge 0)$ pointwisely, and $X_i(b_n)$ is bounded by the integrable variable $Y_i^2$, so we can alternatively write:
	\[
	\mathcal{Y}_d= \{y\in\mathcal{Y}: f(y,d)\ge 0\quad y\in[M_l,M_u]\}\cup \mathcal{Y}_d^{ut}\cup \mathcal{Y}_d^{lt},
	\]
	therefore $\mathbbm{1}(y\in\mathcal{Y}_d^{infsb}(b_n))\rightarrow \mathbbm{1}(y\in\mathcal{Y}_d)$ pointwisely. Last, by Lyapunov CLT, the triangular array converges in distribution to $N(0,\Sigma)$.
\end{proof}

\begin{lem} \label{lem: trimming does not matter 1}
	Let $h_n=n^{-1/5}$ and $b_n=n^{-1/4}/\log n$, and Assumption \ref{assumption: well-defined LATE} - \ref{assumption: trimming bias 1} holds, then  for all $d,k\in\{0,1\}$ and $\underline{t}_F,\bar{t}_F\in \mathbb{R}_{-\infty,+\infty}$ on the extended real line, {\footnotesize
		\begin{equation}
			\begin{split}
				&\bigg|\frac{1}{n} \sum_{i=1}^n Y_i\left[\frac{\mathbbm{1}(D_i=d,Z_i=d)}{Pr(Z_i=d)}-\frac{\mathbbm{1}(D_i=d,Z_i=1-d)}{Pr(Z_i=1-d)}\right](\mathbbm{1}(Y_i\in\hat{\mathcal{Y}}_1(b_n))-\mathbbm{1}(Y_i\in\mathcal{Y}_1^{infsb}(b_n))\bigg|\equiv Term\,\, 1=o_p(1/\sqrt{n}),\\
				&\bigg|\frac{1}{n} \sum_{i=1}^n \left[\frac{\mathbbm{1}(D_i=d,Z_i=d)}{Pr(Z_i=d)}-\frac{\mathbbm{1}(D_i=d,Z_i=1-d)}{Pr(Z_i=1-d)}\right](\mathbbm{1}(Y_i\in\hat{\mathcal{Y}}_1(b_n))-\mathbbm{1}(Y_i\in\mathcal{Y}_1^{infsb}(b_n))\bigg|\equiv Term\,\, 2=o_p(1/\sqrt{n}),\\
				&\bigg|\frac{1}{n} \sum_{i=1}^n Y_i\mathbbm{1}(D_i=d,Z_i=k)(\mathbbm{1}(Y_i\in\hat{\mathcal{Y}}_1(b_n))-\mathbbm{1}(Y_i\in\mathcal{Y}_1^{infsb}(b_n))\bigg|\equiv Term\,\, 3=o_p(1),\\
				&\bigg|\frac{1}{n} \sum_{i=1}^n \mathbbm{1}(D_i=d,Z_i=k)(\mathbbm{1}(Y_i\in\hat{\mathcal{Y}}_1(b_n))-\mathbbm{1}(Y_i\in\mathcal{Y}_1^{infsb}(b_n))\bigg|\equiv Term\,\, 4=o_p(1).
			\end{split}
	\end{equation}}
\end{lem}
\begin{proof}
	I prove the case for $d=k=1$, the rest holds similarly. I look at \textit{$\sqrt{n}$Term 1} first. For any $\epsilon>0$,
	\footnotesize{
	\begin{equation}\label{eq: append, kernel does not affact rate}
		\begin{split}
			&Pr\Bigg(\Bigg|\frac{\sqrt{n}}{n} \sum_{i=1}^n Y_i\left[\frac{\mathbbm{1}(D_i=1,Z_i=1)}{Pr(Z_i=1)}-\frac{\mathbbm{1}(D_i=1,Z_i=0)}{Pr(Z_i=0)}\right] \times(\mathbbm{1}(Y_i\in\hat{\mathcal{Y}}_1(b_n))-\mathbbm{1}(Y_i\in\mathcal{Y}_1^{infsb}(b_n))\Bigg|>\epsilon\Bigg)\\
			&\le Pr( \sup_y|f_h(y,1)-f(y,1)|\ge cn^{-2/5+\epsilon} )\\
			&+  Pr\bigg(\frac{\sqrt{n}}{n} \sum_{i=1}^n \bigg|Y_i\left[\frac{\mathbbm{1}(D_i=1,Z_i=1)}{Pr(Z_i=1)}-\frac{\mathbbm{1}(D_i=1,Z_i=0)}{Pr(Z_i=0)}\right] \times \mathbbm{1}( |f(Y_i,1)|<b_n+cn^{-2/5+\epsilon},\,\,\,\,\,Y_i\in[M_l,M_u])\bigg|>\epsilon\bigg),
		\end{split}
	\end{equation}}
	where the inequality hold because on the event $\sup_y|f_h(y,1)-f(y,1)|< cn^{-2/5+\epsilon}$:
	{\footnotesize \[\left|\mathbbm{1}(Y_i\in\hat{\mathcal{Y}}_1(b_n))-\mathbbm{1}(Y_i\in\mathcal{Y}_1^{infsb}(b_n))\right| \le \mathbbm{1}( |f(Y_i,1)|<b_n+cn^{-2/5+\epsilon},\,\,\,\,\,Y_i\in[M_l,M_u]).\]}
	
	Note that 
	\begin{equation*}
		\footnotesize
		\begin{split}
			&Var\left(\left|\frac{\sqrt{n}}{n} \sum_{i=1}^n Y_i\left[\frac{\mathbbm{1}(D_i=1,Z_i=1)}{Pr(Z_i=1)}-\frac{\mathbbm{1}(D_i=1,Z_i=0)}{Pr(Z_i=0)}\right](\mathbbm{1}( |f(Y_i,1)|<b_n+cn^{-2/5+\epsilon},\,\,\,\,\,Y_i\in[M_l,M_u]))\right|\right)\\	
			&\le E  \left(\left|\frac{\sqrt{n}}{n} \sum_{i=1}^n Y_i\left[\frac{\mathbbm{1}(D_i=1,Z_i=1)}{Pr(Z_i=1)}-\frac{\mathbbm{1}(D_i=1,Z_i=0)}{Pr(Z_i=0)}\right](\mathbbm{1}( |f(Y_i,1)|<b_n+cn^{-2/5+\epsilon},Y_i\in[M_l,M_u]))\right|^2\right)\\
			&\le \underbrace{E  \left[\left|\ Y^2_i\left[\frac{\mathbbm{1}(D_i=1,Z_i=1)}{Pr(Z_i=1)}-\frac{\mathbbm{1}(D_i=1,Z_i=0)}{Pr(Z_i=0)}\right]^2(\mathbbm{1}( |f(Y_i,1)|<b_n+cn^{-2/5+\epsilon},\,\,\,\,\,Y_i\in[M_l,M_u]))\right|\right]}_A\\
			&+ \underbrace{(n-1)E\left[|Y_i|\left|\frac{\mathbbm{1}(D_i=1,Z_i=1)}{Pr(Z_i=1)}-\frac{\mathbbm{1}(D_i=1,Z_i=0)}{Pr(Z_i=0)}\right|\mathbbm{1}( |f(Y_i,1)|<b_n+cn^{-2/5+\epsilon},\,\,\,\,\,Y_i\in[M_l,M_u]\right]^2}_B,\\
		\end{split}
	\end{equation*}
	 Term $A=o(1)$ by the dominated convergence theorem since \[\mathbbm{1}( |f(Y_i,1)|<b_n+cn^{-2/5+\epsilon},\,\,\,\,\,Y_i\in[M_l,M_u])\rightarrow 0,\]
	and the second moment of $Y_i$ is bounded by assumption. For term $B$, by Assumption \ref{assumption: trimming bias 1}, {\footnotesize
		\[
		\begin{split}
			&E\left[Y_i\left|\frac{\mathbbm{1}(D_i=1,Z_i=1)}{Pr(Z_i=1)}-\frac{\mathbbm{1}(D_i=1,Z_i=0)}{Pr(Z_i=0)}\right|\mathbbm{1}( |f(Y_i,1)|<b_n+cn^{-2/5+\epsilon},\,\,\,\,\,Y_i\in[M_l,M_u]))\right]\\
			&\le \max\{|M_l|,|M_u|\}E\left[\left|\frac{\mathbbm{1}(D_i=1,Z_i=1)}{Pr(Z_i=1)}-\frac{\mathbbm{1}(D_i=1,Z_i=0)}{Pr(Z_i=0)}\right|\mathbbm{1}( |f(Y_i,1)|<b_n+cn^{-2/5+\epsilon},\,\,\,\,\,Y_i\in[M_l,M_u])\right]\\
			&=O((b_n+cn^{-2/5+\epsilon})^2)=O(\frac{1}{\sqrt{n}\log^2 n}),
		\end{split}\]}
	therefore $B=(n-1)O(\frac{1}{n\log^4n})=o(1)$. Therefore, the last term in (\ref{eq: append, kernel does not affact rate}) is o(1) by mean squared error convergence. Since $ Pr( \sup_y|f_h(y,1)-f(y,1)|\ge cn^{-2/5+\epsilon} )\rightarrow 0$ by Lemma \ref{lem: convergence rate of density estimator 2}, \textit{$\sqrt{n}$Term 1} is $o_p(1)$.  Proof of \textit{Term 2} is similar to \textit{Term 1}.
	
	Then I look at Term 3:
	\begin{equation*}
		\begin{split}
			&Pr\Bigg(\bigg|\frac{1}{n} \sum_{i=1}^n Y_i\mathbbm{1}(D_i=1,Z_i=1)(\mathbbm{1}(Y_i\in\hat{\mathcal{Y}}_1(b_n)))-\mathbbm{1}(Y_i\in\mathcal{Y}_1^{infsb}(b_n)))\bigg|>\epsilon\Bigg)\\	
			&\le Pr( \sup_y|f_h(y,1)-f(y,1)|\ge cn^{-2/5+\epsilon} )+ Pr\Bigg(\frac{1}{n} \sum_{i=1}^n\bigg| Y_i\mathbbm{1}(D_i=1,Z_i=1)\mathbbm{1}( |f(Y_i,1)|<b_n+cn^{-2/5+\epsilon},\,\,\,Y_i\in[M_l,M_u]))\bigg|>\epsilon\Bigg).
		\end{split}
	\end{equation*}
	Note that 
	\[
	\begin{split}
		&Var\left(\bigg|\frac{1}{n} \sum_{i=1}^n Y_i\mathbbm{1}(D_i=1,Z_i=1)\mathbbm{1}( |f(Y_i,1)|<b_n+cn^{-2/5+\epsilon},\,\,\,Y_i\in[M_l,M_u]))\bigg|\right)\\
		&\le E \left(\bigg|\frac{1}{n} \sum_{i=1}^n Y_i\mathbbm{1}(D_i=1,Z_i=1)\mathbbm{1}( |f(Y_i,1)|<b_n+cn^{-2/5+\epsilon},\,\,\,Y_i\in[M_l,M_u]))\bigg|^2\right)\\
		&\le \underbrace{\frac{1}{n}E \left(\bigg|Y_i^2\mathbbm{1}( |f(Y_i,1)|<b_n+cn^{-2/5+\epsilon},\,\,\,Y_i\in[M_l,M_u])\bigg|\right)}_{C}\\
		&+ \underbrace{\frac{n-1}{n} \left(E\bigg| Y_i\mathbbm{1}(D_i=1,Z_i=1)\mathbbm{1}( |f(Y_i,1)|<b_n+cn^{-2/5+\epsilon},\,\,\,Y_i\in[M_l,M_u])\bigg|\right)^2}_{D}.
	\end{split}
	\]
	Term $C\rightarrow 0$  and Term $D\rightarrow 0$ by dominated convergence theorem, since $b_n+cn^{-2/5+\epsilon}\rightarrow0$.
	The result for Term 3 in the lemma holds by  mean squared error convergence. The result for Term 4 holds by similar argument. 
\end{proof}

\begin{lem}\label{lem: asymptotic linear expansion of numerator and denominator}
	(Asymptotic Linear Expansion of numerators and denominators of (\ref{eq: estimator, type indep}) ) Let $h_n\asymp n^{-1/5}$, $b_n\asymp n^{-1/4}/\log n$, $c_n\asymp n^{-2/5+\epsilon}$ and $0<\epsilon<0.15$ as in Lemma \ref{lem: trimming does not matter 1}, And Assumption \ref{assumption: well-defined LATE} -\ref{assumption: trimming bias 1} holds, then:
	\begin{equation}\label{eq: append, LATE, asymptotic linear expansion equation 1}
	\footnotesize
	\begin{split}
	\bigg|&\quad\frac{1}{n} \sum_{i=1}^n Y_i \left[ \frac{\mathbbm{1}(D_i=1,Z_i=1)}{\frac{1}{n}\sum_{j=1}^n \mathbbm{1}(Z_j=1)}- \frac{\mathbbm{1}(D_i=1,Z_i=0)}{\frac{1}{n}\sum_{j=1}^n \mathbbm{1}(Z_j=0)} \right] \mathbbm{1}(Y_i\in\hat{\mathcal{Y}}_1(b_n))-\int_{\mathcal{Y}_1} y(p(y,1)-q(y,1))dy\\
	&-\frac{E[ Y_i \mathbbm{1}(D_i=1,Z_i=1)\mathbbm{1}(Y_i\in\mathcal{Y}_1) ]}{Pr(Z_i=1)Pr(Z_i=0)} \left[ {\frac{1}{n}\sum_{j=1}^n \mathbbm{1}(Z_j=0)}-{Pr(Z_j=0)}\right] \\
	&+\frac{E[Y_i\mathbbm{1}(D_i=1,Z_i=0)\mathbbm{1}(Y_i\in\mathcal{Y}_1)]}{Pr(Z_i=1)Pr(Z_i=0)}\left[ {\frac{1}{n}\sum_{j=1}^n \mathbbm{1}(Z_j=1)}-{Pr(Z_j=1)}\right]\\
	&-\frac{1}{n{Pr(Z_i=1)}}\sum_{i=1}^n \left(Y_i\mathbbm{1}(D_i=1,Z_i=1)\mathbbm{1}(Y_i\in \mathcal{Y}_1^{infsb}(b_n))-E[Y_i\mathbbm{1}(D_i=1,Z_i=1)\mathbbm{1}(Y_i\in \mathcal{Y}_1^{infsb}(b_n))]\right)\\
	&+\frac{1}{n{Pr(Z_i=0)}}\sum_{i=1}^n \left(Y_i\mathbbm{1}(D_i=1,Z_i=0)\mathbbm{1}(Y_i\in \mathcal{Y}_1^{infsb}(b_n))-E[Y_i\mathbbm{1}(D_i=1,Z_i=0)\mathbbm{1}(Y_i\in \mathcal{Y}_1^{infsb}(b_n))]\right)\bigg|=o_p(1/\sqrt{n})
	\end{split},
	\end{equation}

	\begin{equation}
	\footnotesize
	\begin{split}
	\bigg|&\quad\frac{1}{n} \sum_{i=1}^n Y_i \left[ \frac{\mathbbm{1}(D_i=0,Z_i=0)}{\frac{1}{n}\sum_{j=1}^n \mathbbm{1}(Z_j=0)}- \frac{\mathbbm{1}(D_i=0,Z_i=1)}{\frac{1}{n}\sum_{j=1}^n \mathbbm{1}(Z_j=1)} \right] \mathbbm{1}(Y_i\in\hat{\mathcal{Y}}_0(b_n))-\int_{\mathcal{Y}_0} y(q(y,0)-p(y,0))dy\\
	&- \frac{E[ Y_i \mathbbm{1}(D_i=0,Z_i=0)\mathbbm{1}(Y_i\in\mathcal{Y}_0) ]}{Pr(Z_i=1)Pr(Z_i=0)} \left[ {\frac{1}{n}\sum_{j=1}^n \mathbbm{1}(Z_j=1)}-{Pr(Z_j=1)}\right]\\
	&+\frac{E[Y_i\mathbbm{1}(D_i=0,Z_i=1)\mathbbm{1}(Y_i\in\mathcal{Y}_0)]}{Pr(Z_i=1)Pr(Z_i=0)} \left[ {\frac{1}{n}\sum_{j=1}^n \mathbbm{1}(Z_j=0)}-{Pr(Z_j=0)}\right]\\
	&-\frac{1}{n{Pr(Z_i=0)}}\sum_{i=1}^n \left(Y_i\mathbbm{1}(D_i=0,Z_i=0)\mathbbm{1}(Y_i\in \mathcal{Y}_0^{infsb}(b_n))-E[Y_i\mathbbm{1}(D_i=0,Z_i=0)\mathbbm{1}(Y_i\in \mathcal{Y}_0^{infsb}(b_n))]\right)\\
	&+\frac{1}{n{Pr(Z_i=1)}}\sum_{i=1}^n \left(Y_i\mathbbm{1}(D_i=0,Z_i=1)\mathbbm{1}(Y_i\in \mathcal{Y}_0^{infsb}(b_n))-E[Y_i\mathbbm{1}(D_i=0,Z_i=1)\mathbbm{1}(Y_i\in \mathcal{Y}_0^{infsb}(b_n))]\right)\bigg|=o_p(1/\sqrt{n}),
	\end{split}
	\end{equation}
	
	\begin{equation}
	\footnotesize
	\begin{split}
	\bigg|&\quad\frac{1}{n} \sum_{i=1}^n \left[ \frac{\mathbbm{1}(D_i=1,Z_i=1)}{\frac{1}{n}\sum_{j=1}^n \mathbbm{1}(Z_j=1)}- \frac{\mathbbm{1}(D_i=1,Z_i=0)}{\frac{1}{n}\sum_{j=1}^n \mathbbm{1}(Z_j=0)} \right] \mathbbm{1}(Y_i\in\hat{\mathcal{Y}}_1(b_n))-\int_{\mathcal{Y}_1} (p(y,1)-q(y,1))dy\\
	&-\frac{E[  \mathbbm{1}(D_i=1,Z_i=1)\mathbbm{1}(Y_i\in\mathcal{Y}_1) ]}{Pr(Z_i=1)Pr(Z_i=0)} \left[ {\frac{1}{n}\sum_{j=1}^n \mathbbm{1}(Z_j=0)}-{Pr(Z_j=0)}\right] \\
	&+\frac{E[\mathbbm{1}(D_i=1,Z_i=0)\mathbbm{1}(Y_i\in\mathcal{Y}_1)]}{Pr(Z_i=1)Pr(Z_i=0)}\left[ {\frac{1}{n}\sum_{j=1}^n \mathbbm{1}(Z_j=1)}-{Pr(Z_j=1)}\right]\\
	&-\frac{1}{n{Pr(Z_i=1)}}\sum_{i=1}^n \left(\mathbbm{1}(D_i=1,Z_i=1)\mathbbm{1}(Y_i\in \mathcal{Y}_1^{infsb}(b_n))-E[\mathbbm{1}(D_i=1,Z_i=1)\mathbbm{1}(Y_i\in \mathcal{Y}_1^{infsb}(b_n))]\right)\\
	&+\frac{1}{n{Pr(Z_i=0)}}\sum_{i=1}^n \left(\mathbbm{1}(D_i=1,Z_i=0)\mathbbm{1}(Y_i\in \mathcal{Y}_1^{infsb}(b_n))-E[\mathbbm{1}(D_i=1,Z_i=0)\mathbbm{1}(Y_i\in \mathcal{Y}_1^{infsb}(b_n))]\right)\bigg|=o_p(1/\sqrt{n})
	\end{split},
	\end{equation}

	\begin{equation}
	\footnotesize
	\begin{split}
	\bigg|&\quad\frac{1}{n} \sum_{i=1}^n \left[ \frac{\mathbbm{1}(D_i=0,Z_i=0)}{\frac{1}{n}\sum_{j=1}^n \mathbbm{1}(Z_j=0)}- \frac{\mathbbm{1}(D_i=0,Z_i=1)}{\frac{1}{n}\sum_{j=1}^n \mathbbm{1}(Z_j=1)} \right] \mathbbm{1}(\in\hat{\mathcal{Y}}_0(b_n))-\int_{\mathcal{Y}_0} y(q(y,0)-p(y,0))dy\\
	&- \frac{E[  \mathbbm{1}(D_i=0,Z_i=0)\mathbbm{1}(Y_i\in\mathcal{Y}_0) ]}{Pr(Z_i=1)Pr(Z_i=0)} \left[ {\frac{1}{n}\sum_{j=1}^n \mathbbm{1}(Z_j=1)}-{Pr(Z_j=1)}\right]\\
	&+\frac{E[\mathbbm{1}(D_i=0,Z_i=1)\mathbbm{1}(Y_i\in\mathcal{Y}_0)]}{Pr(Z_i=1)Pr(Z_i=0)} \left[ {\frac{1}{n}\sum_{j=1}^n \mathbbm{1}(Z_j=0)}-{Pr(Z_j=0)}\right]\\
	&-\frac{1}{n{Pr(Z_i=0)}}\sum_{i=1}^n \left(\mathbbm{1}(D_i=0,Z_i=0)\mathbbm{1}(Y_i\in \mathcal{Y}_0^{infsb}(b_n))-E[\mathbbm{1}(D_i=0,Z_i=0)\mathbbm{1}(Y_i\in \mathcal{Y}_0^{infsb}(b_n))]\right)\\
	&+\frac{1}{n{Pr(Z_i=1)}}\sum_{i=1}^n \left(\mathbbm{1}(D_i=0,Z_i=1)\mathbbm{1}(Y_i\in \mathcal{Y}_0^{infsb}(b_n))-E[\mathbbm{1}(D_i=0,Z_i=1)\mathbbm{1}(Y_i\in \mathcal{Y}_0^{infsb}(b_n))]\right)\bigg|=o_p(1/\sqrt{n}),
	\end{split}
	\end{equation}
	
\end{lem}
\begin{proof}
	I prove the first statement (\ref{eq: append, LATE, asymptotic linear expansion equation 1}), and the rest of statements 
	hold similarly by changing the value of $D_i$ and $Z_i$. 
	
	We look at the expansion 
	\begin{equation}
\footnotesize
\begin{split}
&\quad\frac{1}{n} \sum_{i=1}^n Y_i \left[ \frac{\mathbbm{1}(D_i=1,Z_i=1)}{\frac{1}{n}\sum_{j=1}^n \mathbbm{1}(Z_j=1)}- \frac{\mathbbm{1}(D_i=1,Z_i=0)}{\frac{1}{n}\sum_{j=1}^n \mathbbm{1}(Z_j=0)} \right] \mathbbm{1}(Y_i \in\hat{\mathcal{Y}}_1(b_n))-\int_{\mathcal{Y}_1} y(p(y,1)-q(y,1))dy\\
&= \underbrace{\frac{1}{n} \sum_{i=1}^n Y_i \left[ \frac{\mathbbm{1}(D_i=1,Z_i=1)}{\frac{1}{n}\sum_{j=1}^n \mathbbm{1}(Z_j=1)}-\frac{\mathbbm{1}(D_i=1,Z_i=1)}{Pr(Z_j=1)}\right]\mathbbm{1}(Y_i \in\hat{\mathcal{Y}}_1(b_n))}_{A_1}\\
&-\underbrace{\frac{1}{n}\sum_{i=1}^nY_i\left[ \frac{\mathbbm{1}(D_i=1,Z_i=0)}{\frac{1}{n}\sum_{j=1}^n \mathbbm{1}(Z_j=0)}  -\frac{\mathbbm{1}(D_i=1,Z_i=0)}{Pr(Z_j=0)}\right]\mathbbm{1}(Y_i \in\hat{\mathcal{Y}}_1(b_n))}_{A_2}\\
&+\underbrace{\frac{1}{n} \sum_{i=1}^n Y_i\left[\frac{\mathbbm{1}(D_i=1,Z_i=1)}{Pr(Z_i=1)}-\frac{\mathbbm{1}(D_i=1,Z_i=0)}{Pr(Z_i=0)}\right](\mathbbm{1}(Y_i \in\hat{\mathcal{Y}}_1(b_n))-\mathbbm{1}(Y_i\in\mathcal{Y}_1^{infsb}(b_n+c_n))) }_{B}\\
&+ \underbrace{\frac{1}{n} \sum_{i=1}^n Y_i\left[\frac{\mathbbm{1}(D_i=1,Z_i=1)}{Pr(Z_i=1)}-\frac{\mathbbm{1}(D_i=1,Z_i=0)}{Pr(Z_i=0)}\right]\mathbbm{1}(Y_i\in\mathcal{Y}_1^{infsb}(b_n+c_n))}_{C_1}\\
&-\underbrace{\int y(p(y,1)-q(y,1))\mathbbm{1}(y\in\mathcal{Y}_1^{infsb}(b_n+c_n))dy}_{C_2}\\
&+\underbrace{\int_{\mathcal{Y}} y(p(y,1)-q(y,1))(\mathbbm{1}(y\in\mathcal{Y}_1^{infsb}(b_n+c_n))-\mathbbm{1}(y\in\mathcal{Y}_1))dy}_{D}.
\end{split}
\end{equation}
	The expansion holds by adding and subtracting the same terms repeatedly. For term $A_1$, we can write it as{\footnotesize
	\[
	\begin{split}
	A_1&=\frac{1}{n} \sum_{i=1}^n Y_i \mathbbm{1}(D_i=1,Z_i=1)\mathbbm{1}(Y_i \in\hat{\mathcal{Y}}_1(b_n))\left[ \frac{1}{\frac{1}{n}\sum_{j=1}^n \mathbbm{1}(Z_j=1)}-\frac{1}{Pr(Z_j=1)}\right]\\
	&=_{(1)}\left[\frac{1}{n} \sum_{i=1}^n Y_i \mathbbm{1}(D_i=1,Z_i=1)\mathbbm{1}(Y_i\in\mathcal{Y}_1^{infsb}(b_n+c_n)+o_p(1)\right]\left[ \frac{1}{\frac{1}{n}\sum_{j=1}^n \mathbbm{1}(Z_j=1)}-\frac{1}{Pr(Z_j=1)}\right]\\
	&=_{(2)}\left[E[ Y_i \mathbbm{1}(D_i=1,Z_i=1)\mathbbm{1}(Y_i\in\mathcal{Y}_1^{infsb}(b_n+c_n)) ]+o_p(1)\right] \left[ \frac{1}{\frac{1}{n}\sum_{j=1}^n \mathbbm{1}(Z_j=1)}-\frac{1}{Pr(Z_j=1)}\right]\\
	&=_{(3)}\frac{E[ Y_i \mathbbm{1}(D_i=1,Z_i=1)\mathbbm{1}(Y_i\in\mathcal{Y}_1^{infsb}(b_n+c_n)) ]+o_p(1)}{Pr(Z_i=1)Pr(Z_i=0)} \left[ {\frac{1}{n}\sum_{j=1}^n \mathbbm{1}(Z_j=0)}-{Pr(Z_j=0)}\right]\\
	&\times \frac{Pr(Z_i=1)Pr(Z_i=0)}{\frac{1}{n}\sum_{j=1}^n \mathbbm{1}(Z_i=1)\frac{1}{n}\sum_{j=1}^n \mathbbm{1}(Z_i=0)} \\
	&=_{(4)}\frac{E[ Y_i \mathbbm{1}(D_i=1,Z_i=1)\mathbbm{1}(Y_i\in\mathcal{Y}_1^{infsb}(b_n+c_n)) ]+o_p(1)}{Pr(Z_i=1)Pr(Z_i=0)} \left[ {\frac{1}{n}\sum_{j=1}^n \mathbbm{1}(Z_j=0)}-{Pr(Z_j=0)}\right]\times(1+o_p(1))\\
	&=_{(5)} \frac{E[ Y_i \mathbbm{1}(D_i=1,Z_i=1)\mathbbm{1}(Y_i\in\mathcal{Y}_1) ]+o_p(1)}{Pr(Z_i=1)Pr(Z_i=0)} \left[ {\frac{1}{n}\sum_{j=1}^n \mathbbm{1}(Z_j=0)}-{Pr(Z_j=0)}\right]\times(1+o_p(1))\\
	&=_{(6)} \frac{E[ Y_i \mathbbm{1}(D_i=1,Z_i=1)\mathbbm{1}(Y_i\in\mathcal{Y}_1) ]}{Pr(Z_i=1)Pr(Z_i=0)} \left[ {\frac{1}{n}\sum_{j=1}^n \mathbbm{1}(Z_j=0)}-{Pr(Z_j=0)}\right]+o_p(1/\sqrt{n})
	\end{split}
	\]}
	where equality (1) holds by Lemma \ref{lem: trimming does not matter 1} \textit{Term 1} with $d=1$;  equality (2) holds by the Glivenko-Cantalli theorem for changing class of set; equality (3) holds because we multiply and divide the same term; equality (4) holds by the continuous mapping theorem; equality (5) holds by dominated convergence theorem since $\mathbbm{1}(y\in\mathcal{Y}_1^{infsb}(b_n+c_n))\rightarrow \mathbbm{1}(y\in\mathcal{Y}_1)$; equality (6) holds by observing that $ {\frac{1}{n}\sum_{j=1}^n \mathbbm{1}(Z_j=0)}-{Pr(Z_j=0)}=o_p(1/\sqrt{n})$ and then we apply Slutsky's theorem to get the equality. 
	
	Similarly, apply Lemma \ref{lem: trimming does not matter 1} \textit{Term 3} with $d=1,k=0$ we have 
	\[A_2= \frac{E[Y_i\mathbbm{1}(D_i=1,Z_i=0)\mathbbm{1}(Y_i\in\mathcal{Y}_1)]}{Pr(Z_i=1)Pr(Z_i=0)}\left[ {\frac{1}{n}\sum_{j=1}^n \mathbbm{1}(Z_j=1)}-{Pr(Z_j=1)}\right]+o_p(1/\sqrt{n}).\]
	By Lemma \ref{lem: trimming does not matter 1}, $B=o_p(1/\sqrt{n})$ \textit{Term 1}. By Assumption \ref{assumption: trimming bias 1} and the choice of $\epsilon$, \[
	D\le \max\{|M_u|,|M_l|\}O((b_n+c_n)^2)=O_p(n^{-0.5}/\log^2n)=o_p(1/\sqrt{n}).
	\]
	The result follows since $C_1-C_2$ term corresponds to the last two terms in (\ref{eq: append, LATE, asymptotic linear expansion equation 1}). 
	 
	The rest of the equalities in Lemma \ref{lem: asymptotic linear expansion of numerator and denominator} hold by applying different values of $d,k$ in Lemma  \ref{lem: trimming does not matter 1}.
\end{proof}

\subsubsection{Proof of Theorem \ref{thm: asymptotic property of LATE}}
\begin{proof}
	 Let $\bm{X}_i(b_n)$ be the vector in Lemma \ref{lem: Limit Distribution of Infeasible Components}. Now let 
	\begin{equation*}
	\hat{\pi}= \begin{pmatrix}
	\frac{1}{n} \sum_{i=1}^N Y_i \left[ \frac{\mathbbm{1}(D_i=1,Z_i=1)}{\frac{1}{N}\sum_{j=1}^N \mathbbm{1}(Z_j=1)}- \frac{\mathbbm{1}(D_i=1,Z_i=0)}{\frac{1}{N}\sum_{j=1}^N \mathbbm{1}(Z_j=0)} \right] \mathbbm{1}(Y_i\in\hat{\mathcal{Y}}_1(b_n))\\
	\frac{1}{n} \sum_{i=1}^n Y_i \left[ \frac{\mathbbm{1}(D_i=0,Z_i=0)}{\frac{1}{n}\sum_{j=1}^n \mathbbm{1}(Z_j=0)}- \frac{\mathbbm{1}(D_i=0,Z_i=1)}{\frac{1}{n}\sum_{j=1}^n \mathbbm{1}(Z_j=1)} \right] \mathbbm{1}(Y_i\in\hat{\mathcal{Y}}_0(b_n))\\
	\frac{1}{n} \sum_{i=1}^N \left[ \frac{\mathbbm{1}(D_i=1,Z_i=1)}{\frac{1}{N}\sum_{j=1}^N \mathbbm{1}(Z_j=1)}- \frac{\mathbbm{1}(D_i=1,Z_i=0)}{\frac{1}{N}\sum_{j=1}^N \mathbbm{1}(Z_j=0)} \right] \mathbbm{1}(Y_i\in\hat{\mathcal{Y}}_1(b_n))\\
	\frac{1}{n} \sum_{i=1}^n \left[ \frac{\mathbbm{1}(D_i=0,Z_i=0)}{\frac{1}{n}\sum_{j=1}^n \mathbbm{1}(Z_j=0)}- \frac{\mathbbm{1}(D_i=0,Z_i=1)}{\frac{1}{n}\sum_{j=1}^n \mathbbm{1}(Z_j=1)} \right] \mathbbm{1}(Y_i\in\hat{\mathcal{Y}}_0(b_n))
	\end{pmatrix}
	\quad\quad
	\pi= \begin{pmatrix}
	\int_{\mathcal{Y}_1} y(p(y,1)-q(y,1))dy\\
	\int_{\mathcal{Y}_0} y(q(y,0)-p(y,0))dy\\
	\int_{\mathcal{Y}_1} (p(y,1)-q(y,1))dy\\
	\int_{\mathcal{Y}_0} (q(y,0)-p(y,0))dy
	\end{pmatrix}
	\end{equation*}
	By Lemma \ref{lem: asymptotic linear expansion of numerator and denominator},
	\[\sqrt{n}(\hat{\pi}-\pi)=o_p(1)+ \Gamma\sqrt{n}(\bm{X}_i(b_n)-E[\bm{X}_i(b_n)]),\]
	where $\Gamma$ matrix is specified in Theorem \ref{thm: asymptotic property of LATE}.
	And we notice that 
	$\widehat{LATE}=\frac{\hat{\pi}_1}{\hat{\pi}_3}-\frac{\hat{\pi}_2}{\hat{\pi}_4}$, and $\widetilde{LATE}^{ID}=\frac{\pi_1}{\pi_3}-\frac{\pi_2}{\pi_4}$, and $\Pi$ in Theorem \ref{thm: asymptotic property of LATE} is the Jacobian matrix of function  $f(\pi)=\frac{\pi_1}{\pi_3}-\frac{\pi_2}{\pi_4}$. The result follows by delta method.

\end{proof}

\section{Proofs in Section 4}
\subsection{Proof of Proposition \ref{prop: roy model, non-refutable and confirmation set}}
\begin{proof}
	To show the non-refutability sets $\mathcal{H}_{\mathcal{S}}^{snf}(A^{Roy})$ and $\mathcal{H}_{\mathcal{S}}^{wnf}(A^{Roy})$, it suffices to show $\cup_{s\in A^{Roy}}M^s(G^s)=\mathcal{F}^{nf}$ by Definition \ref{def: non-refutablity set in incomplete structure}. Note that for any $s\in A^{Roy}$ such that $F\in M^s(G^s)$, 
	\[
	\begin{split}
	\frac{Pr_F(Y_i=0,Z_i=1)}{Pr_F(Z_i=1)}&=_{(1)} \frac{Pr_{G^s}(Y_i(1)=0,Y_i(0)=0,Z_i=1)}{Pr_{G^s}(Z_i=1)}\\
	&\le_{(2)} \frac{Pr_{G^s}(Y_i(1)=0,Y_i(0)=0,Z_i=0)}{Pr_{G^s}(Z_i=0)}\\
	&= \frac{Pr_F(Y_i=0,Z_i=0)}{Pr_F(Z_i=0)},
	\end{split}
	\]
	where (1) holds by the Roy sector selection rule, and inequality (2) holds by (\ref{eq: roy model, dominating instrument at best and worst outcome}). Therefore $\cup_{s\in A^{Roy}}M^s(G^s) \subseteq \mathcal{F}^{nf}$.
	
	Conversely, if $F$ satisfies $\frac{Pr_F(Y_i=0,Z_i=1)}{Pr_F(Z_i=1)}\le \frac{Pr_F(Y_i=0,Z_i=0)}{Pr_F(Z_i=0)}$, we consider the following $s$:
	\[
	\begin{split}
	Pr_{G^{s}}(Y_i(1)=0,Y_i(0)=0,Z_i=0)&=Pr_F(Y_i=0,Z_i=0)\\
	Pr_{G^{s}}(Y_i(1)=0,Y_i(0)=1,Z_i=0)&=Pr_F(Y_i=1,D_i=0,Z_i=0)\\
	Pr_{G^{s}}(Y_i(1)=1,Y_i(0)=0,Z_i=0)&=Pr_F(Y_i=1,D_i=1,Z_i=0)\\
	Pr_{G^{s}}(Y_i(1)=1,Y_i(0)=1,Z_i=0)&=0\\
	Pr_{G^{s}}(Y_i(1)=0,Y_i(0)=0,Z_i=1)&=Pr_F(Y_i=0,Z_i=1)\\
	Pr_{G^{s}}(Y_i(1)=0,Y_i(0)=1,Z_i=0)&=0\\
	Pr_{G^{s}}(Y_i(1)=1,Y_i(0)=0,Z_i=0)&=0\\
	Pr_{G^{s}}(Y_i(1)=1,Y_i(0)=1,Z_i=0)&=Pr_F(Y_i=1,Z_i=1)\\
	\end{split}
	\]
	and $M^s$ is the Roy sector selection rule (\ref{eq: roy model, perfect selection D_i mapping}) and $M^s(G^s)$ is implied by (\ref{eq: Roy appli, M^s}). By the construction of $G^s$, and set
	\[
	\begin{split}
	C^{110}_0=C^{110}_1&=C^{001}_0=C^{001}_1=0,\\
	C^{111}_0=Pr_F(Y_i=1,D_i=0,Z_i=1), \quad &\quad C^{111}_1=Pr_F(Y_i=1,D_i=1,Z_i=1),\\
	C^{000}_0=Pr_F(Y_i=0,D_i=0,Z_i=0), \quad &\quad C^{000}_1=Pr_F(Y_i=0,D_i=1,Z_i=0),
	\end{split}
	\]
	we can show $F\in M^s(G^s)$. As a result, $\mathcal{F}^{nf}\subseteq \cup_{s\in A^{Roy}}M^s(G^s)$. Therefore, $\cup_{s\in A^{Roy}}M^s(G^s)=\mathcal{F}^{nf}$. 
	
	To show $\mathcal{H}_{\mathcal{S}}^{wcon}(A^{Roy})=\mathcal{H}_{\mathcal{S}}^{scon}(A^{Roy})=\varnothing$, it suffices to show $\cup_{s\in (A^{Roy})^c}M^s(G^s)=\mathcal{F}$ by Definition \ref{def: confirmation sets}. 
	
	For any $F\in \mathcal{F}$, we consider the unrestricted sector selection rule $D^{unc}$, i.e. $D^{unc}(y_1,y_0,z)=\{0,1\}$ for all values of  $y_1,y_0,z\in\{0,1\}$.
	A unconstrained structure $s^{unc}$ corresponding to $D^{unc}$ satisfies:
	\begin{equation}\label{eq: roy model, confirmation set is empty}
	M^{s^{unc}}(G^{s^{unc}})=\{F: Pr_F(Z_i=z)=Pr_{G^{s^{unc}}}(Z_i=z)\}.
	\end{equation}
	In other words, the unconstrained selection rule does not specify how the job sector is picked for all cases.  Now, consider a structure $s=(M^s,G^s)$ such that: (1). $M^s$ corresponds to $D^{unc}$, and (2). $G^s$ satisfies $Pr_{G^s}(Z_i=z)=Pr_F(Z_i=z)$. By (\ref{eq: roy model, confirmation set is empty}), $F\in M^s(G^s)$ must hold. This shows that $\cup_{s\in A^c} M^s(G^s)=\mathcal{F}$. By definition of confirmation sets \ref{def: confirmation sets}, $	\mathcal{H}_{\mathcal{S}}^{scon}(A^{Roy})=\mathcal{H}_{\mathcal{S}}^{wcon}(A^{Roy})=\varnothing $ holds.
	
\end{proof}

\subsection{Proof of Proposition \ref{prop: roy model, sharp characterization of identified set} }
 I first set up some notations for the proof. For any completed structure $s^*$, let 
\[
D_i^{s^*}=\begin{cases}
1\quad \quad with\,\,probability \,\,p^{jkz,s^*}_1, \,\,when \,\,Y_i(1)=j,Y_i(0)=k,Z_i=z,\\
0\quad \quad with\,\,probability \,\,p^{jkz,s^*}_0, \,\,when \,\,Y_i(1)=j,Y_i(0)=k,Z_i=z,\\
\end{cases}
\]
where $p^{jkz,s^*}_1+p^{jkz,s^*}_0=1$ and $p^{jkz,s^*}_1\in [0,1]$. Also, let 
\[
C_d^{jkz}(s^*)=Pr(Y_i(1)=j,Y_i(0)=k,Z_i=z)\times p^{jkz,s^*}_d.
\]
The quantity $p^{jkz,s^*}_d$ under structure $s^*$ can be viewed as the probability of selecting sector $d$ when $Y_i(1)=j,Y_i(0)=k,Z_i=z$. Instead of using the selection rule to characterize $s^*$, we use the $C_d^{jkz,s^*}$ to characterize $s^*$: Each structure can then be represented by a $G^{s^*}$ and the numbers $\{C_d^{jkz}(s^*)\}_{d,j,k,z\in\{0,1\}}$. The set of structures that are consistent with observation $F$ and (\ref{eq: roy model, dominating instrument at best and worst outcome}) is given by 
\begin{equation*}
\begin{split}
\Bigg\{ s^*:\,\,
&C_1^{j1z}(s^*)+C_{1}^{j0z}(s^*)=Pr_F(Y_i=j,D_i=1,Z_i=z),\\
&C_0^{1kz}(s^*)+C_{0}^{0kz}(s^*)=Pr_F(Y_i=k,D_i=1,Z_i=z),\\
&C_1^{jkz}+C_0^{jkz}=Pr_{G^{s^*}}(Y_i(1)=j,Y_i(0)=k,Z_i=z),\\
&and\,\, (\ref{eq: roy model, dominating instrument at best and worst outcome})\,\,holds\Bigg\}.
\end{split}
\end{equation*}
Also, to abbreviate the notation of $G^{s^*}$, we let  $q^{y_1y_0z,s*}=Pr_{G^{s^*}}(Y_i(1)=y_1,Y_i(0)=y_0,Z_i=z)$. The proof of Proposition \ref{prop: roy model, sharp characterization of identified set} is based on the following lemmas. 

\subsubsection{Lemmas}
\begin{lem} \label{lem: no efficiency loss for z=0}
	Let $s^*$ be any structure such that $F\in M^{s^*}(G^{s^*})$ and (\ref{eq: roy model, dominating instrument at best and worst outcome}) holds, there exists an $\tilde{s}^*$ such that $F\in M^{\tilde{s}^*}(G^{\tilde{s}^*})$ such that $\tilde{s}^*$ satisfies (\ref{eq: roy model, dominating instrument at best and worst outcome}) and $C_d^{jk1}(s^*)=C_d^{jk1}(\tilde{s}^*)$ for all $d,j,k\in \{0,1\}$, and $C_0^{100}(\tilde{s}^*)=C_1^{010}(\tilde{s}^*)=0$. Moreover, $m^{EL}(s^*)\ge m^{EL}(\tilde{s}^*)$, with equality hold if and only if $C_0^{100}({s}^*)=C_1^{010}({s}^*)=0$.
\end{lem}
\begin{remark}
	This lemma shows that it suffices to consider structures such that the  efficiency loss at $Z_i=0$ is zero ($C_0^{100}(\tilde{s}^*)=C_1^{010}(\tilde{s}^*)=0$). 
\end{remark}

\begin{proof}
	Given $s^*$, we construct $\tilde{s}^*$ such that $q^{jk1,s^*}=q^{jk1,\tilde{s}^*}$ for $j,k\in\{0,1\}$ and $C_d^{jk1}(s^*)=C_d^{jk1}(\tilde{s}^*)$. The definition of ${s}^*$ implies 
	\[\begin{split}
	&C_1^{j1z}(s^*)+C_{1}^{j01}(s^*)=Pr_F(Y_i=j,D_i=1,Z_i=1),\\
	&C_0^{1kz}(s^*)+C_{0}^{0k1}(s^*)=Pr_F(Y_i=k,D_i=1,Z_i=1).\\
	\end{split}
	\]
	Now, let $C_0^{100}(\tilde{s}^*)=C_1^{010}(\tilde{s}^*)=0$, and let 
	\[
	\begin{split}
	C^{000}_0(\tilde{s}^*)&=C^{000}_0(s^*)+C^{100}_0(s^*),\\
	C^{000}_1(\tilde{s}^*)&=C^{000}_1(s^*)+C^{010}_1(s^*),\\
	C^{010}_0(\tilde{s}^*)=C^{010}_0({s}^*), \quad &\quad C^{100}_1(\tilde{s}^*)=C^{100}_1({s}^*), \\
	C^{110}_1(\tilde{s}^*)=C^{110}_1({s}^*), \quad &\quad C^{110}_0(\tilde{s}^*)=C^{110}_0({s}^*),\\
	q^{000,\tilde{s}^*}=q^{000,s^*}+C^{100}_0(s^*)+C^{010}_1(s^*),\quad &\quad 	q^{110,\tilde{s}^*}=q^{110,{s}^*},\\
	q^{100,\tilde{s}^*}=q^{100,s^*}-C^{100}_0(s^*), \quad &\quad
	q^{010,\tilde{s}^*}=q^{010,s^*}-C^{010}_1(s^*).\\
	\end{split}
	\]
	The construction of $\tilde{s}^*$ above essentially moves the original efficiency loss allocation $C^{100}_0(s^*)$ and $C^{010}_1(s^*)$ to the event that $Y_i(1)=Y_i(0)=0,Z_i=0$. By construction of $\tilde{s}^*$, $F\in M^{s^*}(G^{s^*})$: For example to check $Pr_F(Y_i=0,D_i=0,Z_i=0)$ can be generated by our structure $s$, we write:
	\[
	\begin{split}
	&\,C^{000}_0(\tilde{s}^*)+C^{100}_0(\tilde{s}^*)\\
	&= C^{000}_0(\tilde{s}^*)\\
	&=	C^{000}_0(s^*)+C^{100}_0(s^*)\\
	&= Pr_F(Y_i=0,D_i=0,Z_i=0),
	\end{split}
	\]
	where the first equality holds by construction $C^{100}_0(\tilde{s}^*)=0$, and the third equality holds because $F\in M^{s^*}(G^{s^*})$. By the definition of $m^{EL}(s^*)$ in (\ref{eq: Roy model, efficiency loss}), we can show $m^{EL}(s^*)-m^{EL}(\tilde{s}^*)= C^{100}_0(s^*)+C^{010}_1(s^*)$. So $m^{EL}(s^*)\ge m^{EL}(\tilde{s}^*)$ with equality holds iff $C^{100}_0(s^*)=C^{010}_1(s^*)=0$.

	Last, it suffices to check $\tilde{s}^*$ satisfies (\ref{eq: roy model, dominating instrument at best and worst outcome}). Since $Pr_F(Z_i=z)=Pr_{G^s}(Z_i=z)$ for all $F\in M^s(G^s)$, we will use $F$ to denote the marginal distribution of $Z_i$. To check the first inequality in  (\ref{eq: roy model, dominating instrument at best and worst outcome}), we want to show 
	\[
	\frac{q^{110,\tilde{s}^*}}{Pr_F(Z_i=0)} \le \frac{q^{111,\tilde{s}^*}}{Pr_F(Z_i=1)}.
	\]
	Note that since $s^*$ satisfies $\frac{q^{110,{s}^*}}{Pr_F(Z_i=0)} \le \frac{q^{111,{s}^*}}{Pr_F(Z_i=1)}$, and $q^{11z,\tilde{s}^*}=q^{11z,{s}^*}$ for $z=0,1$ by the construction of $\tilde{s}^*$, then $\frac{q^{110,\tilde{s}^*}}{Pr_F(Z_i=0)} \le \frac{q^{111,\tilde{s}^*}}{Pr_F(Z_i=1)}$ holds. For the second inequality in  (\ref{eq: roy model, dominating instrument at best and worst outcome}), we want to show 
	\[
	\frac{q^{001,\tilde{s}^*}}{Pr_F(Z_i=1)}\le \frac{q^{000,\tilde{s}^*}}{Pr_F(Z_i=0)}.
	\]
	Note that since $\frac{q^{001,{s}^*}}{Pr_F(Z_i=1)}\le \frac{q^{000,{s}^*}}{Pr_F(Z_i=0)}$ holds for $s^*$, and $q^{001,\tilde{s}^*}=q^{001,{s}^*}$, $q^{000,\tilde{s}^*}=q^{000,s^*}+C^{100}_0(s^*)+C^{010}_1(s^*)>q^{000,s^*}$, so $\frac{q^{001,\tilde{s}^*}}{Pr_F(Z_i=1)}\le \frac{q^{000,\tilde{s}^*}}{Pr_F(Z_i=0)}$ holds for $\tilde{s}^*$. The result follows.
\end{proof}

\begin{lem}\label{lem: appen, roy model, a particular construction}
	For any $s^*$ satisfying (\ref{eq: roy model, dominating instrument at best and worst outcome}) and $C_0^{100}({s}^*)=C_1^{010}({s}^*)=0$ and $F\in M^{s^*}(G^{s^*})$, consider a structure $\tilde{s}^*$:
	\begin{equation}\label{eq: append, roy model, particular const of s, 1}
	\begin{split}
	q^{000,\tilde{s}^*}=Pr_F(Y_i=0,Z_i=0),\quad &\quad q^{010,\tilde{s}^*}=Pr_F(Y_i=1,D_i=0,Z_i=0),\\
	q^{110,\tilde{s}^*}=0,\quad &\quad q^{100,\tilde{s}^*}=Pr_F(Y_i=1,D_i=1,Z_i=0),\\
	C^{000}_0(\tilde{s}^*)=C^{000}_1(\tilde{s}^*)=q^{000,\tilde{s}^*}/2, \quad &\quad C^{100}_0(\tilde{s}^*)=C^{010}_1(\tilde{s}^*)=C^{110}_0(\tilde{s}^*)=C^{110}_1(\tilde{s}^*)=0,\\
	C^{010}_0(\tilde{s}^*)=q^{010,\tilde{s}^*},\quad &\quad C^{100}_1(\tilde{s}^*)=q^{100,\tilde{s}^*},
	\end{split}	
	\end{equation}
	and 
	\begin{equation}\label{eq: append, roy model, particular const of s, 2}
	\begin{split}
	&C^{111}_0(\tilde{s}^*)=Pr_F(Y_i=1,D_i=0,Z_i=1),\quad \quad C^{111}_1(\tilde{s}^*)=Pr_F(Y_i=1,D_i=1,Z_i=1),\\
	&C^{101}_0(\tilde{s}^*)= C^{101}_0(s^*), \quad \quad C^{011}_1(\tilde{s}^*)=C^{011}_1({s}^*),\\
	&C^{101}_1(\tilde{s}^*)=C^{011}_0(\tilde{s}^*)=0,\\
	&C^{001}_0(\tilde{s}^*)=Pr_F(Y_i=0,D_i=0,Z_i=1)-C^{101}_0(\tilde{s}^*)=C^{001}_0({s}^*),\\ &C^{001}_1(\tilde{s}^*)=Pr_F(Y_i=0,D_i=1,Z_i=1)-C^{011}_1(\tilde{s}^*)=C^{001}_1({s}^*),\\
	&q^{jk1,\tilde{s}^*} = C^{jk1}_1(\tilde{s}^*)+C^{jk1}_0(\tilde{s}^*).
	\end{split}
	\end{equation}
	Then $F\in M^{\tilde{s}^*}(G^{\tilde{s}^*})$ and $G^{\tilde{s}^*}$ satisfies (\ref{eq: roy model, dominating instrument at best and worst outcome}). Moreover, $
	m^{EL}(s^*)=m^{EL}(\tilde{s}^*)$. 
\end{lem}

\begin{remark}
	This Lemma gives a `representative' of the structures $\tilde{s}^*$ in Lemma \ref{lem: no efficiency loss for z=0}. 
\end{remark}

\begin{proof}
	By the construction of $\tilde{s}^*$,
	\[
	\begin{split}
	&C_1^{j1z}(\tilde{s}^*)+C_{1}^{j0z}(\tilde{s}^*)=Pr_F(Y_i=j,D_i=1,Z_i=z)\\
	&C_0^{1kz}(\tilde{s}^*)+C_{0}^{0kz}(\tilde{s}^*)=Pr_F(Y_i=k,D_i=1,Z_i=z),\\
	\end{split}
	\]
	and $C^{jkz}_d\ge 0$ for all $j,k,z,d\in \{0,1\}$. Also, 
	\[
	\begin{split}
	m^{EL}(\tilde{s}^*)&=C^{100}_0(\tilde{s}^*)+C^{010}_1(\tilde{s}^*)+C^{101}_0(\tilde{s}^*)+C^{011}_1(\tilde{s}^*)\\
	&=C^{101}_0(\tilde{s}^*)+C^{011}_1(\tilde{s}^*)\\
	&=C^{101}_0(\tilde{s})+C^{011}_1(\tilde{s})=m^{EL}(s^*),
	\end{split}
	\]
	so it remains to check $G^{\tilde{s}^*}$ satisfies (\ref{eq: roy model, dominating instrument at best and worst outcome}). 
	
	Since by the construction of $\tilde{s}^*$,  $Pr_{G^{\tilde{s}^*}}(Y_i(1)=1,Y_i(0)=1,Z_i=0)=q^{110,\tilde{s}^*}=0$, so the first inequality in (\ref{eq: roy model, dominating instrument at best and worst outcome}) holds automatically. Also, since the second inequality in (\ref{eq: roy model, dominating instrument at best and worst outcome}) holds for $s^*$, we have 
	\[
	\begin{split}
	\frac{Pr_{G^{{s}^*}}(Y_i(1)=0,Y_i(0)=0,Z_i=1)}{Pr_{G^{{s}^*}}(Z_i=1)}
	\le  \frac{Pr_{G^{{s}^*}}(Y_i(1)=0,Y_i(0)=0,Z_i=0)}{Pr_{G^{{s}^*}}(Z_i=0)}.
	\end{split}
	\]
	By construction, $Pr_{G^{{s}^*}}(Y_i(1)=0,Y_i(0)=0,Z_i=0)=Pr_{G^{\tilde{s}^*}}(Y_i(1)=0,Y_i(0)=0,Z_i=0)$, and $Pr_{G^{\tilde{s}^*}}(Z_i=z)=Pr_{G^{{s}^*}}(Z_i=z)$. Also 
	\[
	\begin{split}
	Pr_{G^{\tilde{s}^*}}(Y_i(1)=0,Y_i(0)=0,Z_i=1)&=_{(1)}Pr_F(Y_i=0,Z_i=1)-C^{101}_0(\tilde{s}^*)-C^{011}_1(\tilde{s}^*)\\
	&=_{(2)} Pr_F(Y_i=0,Z_i=1)-C^{101}_0({s}^*)-C^{011}_1({s}^*)\\
	&=_{(3)} Pr_{G^{{s}^*}}(Y_i(1)=0,Y_i(0)=0,Z_i=1),
	\end{split}
	\]
	where (1) follows by the construction of $\tilde{s}^*$, (2) follows by $C^{101}_0(\tilde{s}^*)=C^{101}_0({s}^*)$ and $C^{011}_0(\tilde{s}^*)=C^{011}_0({s}^*)$, (3) follows by $F\in M^{s^*}(G^{s^*})$. Therefore, we have 
	\[
	\begin{split}
	\frac{Pr_{G^{\tilde{s}^*}}(Y_i(1)=0,Y_i(0)=0,Z_i=1)}{Pr_{G^{\tilde{s}^*}}(Z_i=1)}
	\le  \frac{Pr_{G^{\tilde{s}^*}}(Y_i(1)=0,Y_i(0)=0,Z_i=0)}{Pr_{G^{\tilde{s}^*}}(Z_i=0)}
	\end{split},
	\]
	which implies the second inequality in (\ref{eq: roy model, dominating instrument at best and worst outcome}) also holds $\tilde{s}^*$. The result follows.

\end{proof}

\begin{lem}\label{lem: appen, roy model, min eff loss}
	Given any $F\in\mathcal{F}$, for all structure $s^*$ that satisfies (\ref{eq: roy model, dominating instrument at best and worst outcome}) and $F\in M^{s^*}(G^{s^*})$, 
	\[m^{EL}(s^*)\ge \max\left\{0,{Pr_F(Y_i=0,Z_i=1)}-\frac{Pr_F(Y_i=0,Z_i=0)}{Pr_F(Z_i=0)}Pr(Z_i=1)\right\}. \]
	Moreover, there exists a structure $\tilde{s}^*$ achieves this lower bound. 
\end{lem}

\begin{proof}
	By Lemma \ref{lem: no efficiency loss for z=0}, to find the lower bound of efficiency loss, it suffices to look at the class of structures that has no efficiency loss at $Z_i=0$, i.e. $C_0^{100}({s}^*)=C_1^{010}({s}^*)=0$. By Lemma \ref{lem: appen, roy model, a particular construction}, we can focus on the `representative' type of structures that satisfies (\ref{eq: append, roy model, particular const of s, 1}) and (\ref{eq: append, roy model, particular const of s, 2}) with $C^{101}_0(\tilde{s}^*)$ and $C^{011}_1(\tilde{s}^*)$ undetermined. Let $\tilde{s}^*$ be the representative structure in Lemma \ref{lem: appen, roy model, a particular construction}. Note that $\tilde{s}^*$ needs to be a structure, so in (\ref{eq: append, roy model, particular const of s, 2}), we must have 
	\[
	\begin{split}
	0\le C^{101}_0(\tilde{s}^*)\le Pr_F(Y_i=0,D_i=0,Z_i=1)\\
	0\le C^{011}_1(\tilde{s}^*)\le Pr_F(Y_i=0,D_i=1,Z_i=1).
	\end{split}	
	\]
	Also, $\tilde{s}^*$ has to satisfy  (\ref{eq: roy model, dominating instrument at best and worst outcome}). Since $q^{110,\tilde{s}^*}=0$, the first inequality in  (\ref{eq: roy model, dominating instrument at best and worst outcome}) holds automatically. The second inequality in (\ref{eq: roy model, dominating instrument at best and worst outcome}) requires  
	\[
	\frac{C^{000}_1(\tilde{s}^*)+C^{000}_0(\tilde{s}^*)}{Pr_F(Z_i=0)}\ge \frac{C^{001}_1(\tilde{s}^*)+C^{001}_0(\tilde{s}^*)}{Pr_F(Z_i=0)}.
	\]
	By the construction of $\tilde{s}^*$, this is equivalent to 
	\[
	\frac{Pr_F(Y_i=0,Z_i=0)}{Pr_F(Z_i=0)}\ge \frac{Pr_F(Y_i=0,Z_i=1)- C^{011}_1(\tilde{s}^*)-C^{101}_0(\tilde{s}^*)}{Pr_F(Z_i=1)}.
	\]
	The problem of finding the minimal efficiency loss becomes a linear programming problem:
	\begin{tcolorbox}	
		\begin{equation}
		\begin{split}
		&\quad \quad \quad \quad\quad \quad \quad \quad\quad \quad \quad \quad\min C^{011}_1(\tilde{s}^*)+C^{101}_0(\tilde{s}^*)\\
		&\quad \quad \quad \quad \quad \quad \quad \quad\quad \quad \quad \quad \quad \quad \quad \quad s.t. \\
		&0\le C^{101}_0(\tilde{s}^*)\le Pr_F(Y_i=0,D_i=0,Z_i=1),\\
		&0\le C^{011}_1(\tilde{s}^*)\le Pr_F(Y_i=0,D_i=1,Z_i=1),\\
		&C^{011}_1(\tilde{s}^*)+C^{101}_0(\tilde{s}^*)\ge  {Pr_F(Y_i=0,Z_i=1)}-\frac{Pr_F(Y_i=0,Z_i=0)}{Pr_F(Z_i=0)}Pr(Z_i=1).
		\end{split}
		\end{equation}
	\end{tcolorbox}
	Since 
	\[
	\begin{split}
	&\quad Pr_F(Y_i=0,D_i=1,Z_i=1)+Pr_F(Y_i=0,D_i=0,Z_i=1)\\ &-\left[{Pr_F(Y_i=0,Z_i=1)}-\frac{Pr_F(Y_i=0,Z_i=0)}{Pr_F(Z_i=0)}Pr(Z_i=1)\right]>0
	\end{split}\]
	always holds, so the feasible region is non-empty. This is a linear programming problem with bounded feasible set, so the minimal exists and can be achieved by some $C^{011}_1(\tilde{s}^*),C^{101}_0(\tilde{s}^*)$, which corresponds to the structure that achieves the minimal efficiency loss.  In particular, the minimum is achieved at the feasible point:
	\[
	(C^{011}_1(\tilde{s}^*),C^{101}_0(\tilde{s}^*))=\left(0,\max\left\{0,{Pr_F(Y_i=0,Z_i=1)}-\frac{Pr_F(Y_i=0,Z_i=0)}{Pr_F(Z_i=0)}Pr(Z_i=1)\right\}\right).
	\]
	The result follows. 
\end{proof}
\subsubsection{Main Proof}
\begin{proof}
	Lemma \ref{lem: no efficiency loss for z=0} shows that it the identified set for $G^{s^*}$ must satisfy $C^{010}_1(s^*)=C^{100}_0(s^*)=0$, i.e. no efficiency loss at $Z_i=0$. Lemma \ref{lem: appen, roy model, min eff loss} shows the minimal efficiency loss is $\max\{0,{Pr_F(Y_i=0,Z_i=1)}-\frac{Pr_F(Y_i=0,Z_i=0)}{Pr_F(Z_i=0)}Pr(Z_i=1)\}$. The result in Proposition (\ref{prop: roy model, sharp characterization of identified set}) follows by imposing the minimal efficiency loss and condition (\ref{eq: roy model, dominating instrument at best and worst outcome}).
\end{proof}

\subsection{Proof of Corollary \ref{Corollary: Roy model, ID set of $Pr(Y_i(1)=1|Z_i=z)$}}
\begin{proof}
	First note that $Pr(Y_i(1)=1|Z_i=z)=\frac{C^{11z}_1+C^{11z}_0+C^{10z}_1+C^{10z}_0}{Pr_F(Z_i=z)}$, and 
	\[
	C^{11z}_1+C^{10z}_1=Pr_F(Y_i=1,D_i=1,Z_i=z),
	\]
	so the following inequalities holds by $C^{11z}_0\ge 0$ and $C^{10z}_0\ge 0$:
	\[
	\begin{split}
	Pr(Y_i=1,D_i=1|Z_i=0)\le &Pr(Y_i(1)=1|Z_i=0),\\
	Pr(Y_i=1,D_i=1|Z_i=1)\le &Pr(Y_i(1)=1|Z_i=1).\\
	\end{split}
	\]
	Given the identified set for $G^{s^*}$ in Proposition \ref{prop: roy model, sharp characterization of identified set}, we know $C^{100}_0=0$,  $C^{110}_0\le Pr_F(Y_i=1,D_i=1,Z_i=0)$, and $C^{111}_0\le Pr_F(Y_i=1,D_i=1,Z_i=1)$, $C^{101}_0\le m^{EL,min}(F)$, so the following inequalities hold:
	\[
	\begin{split}
	&Pr(Y_i(1)=1|Z_i=0)\le Pr(Y_i=1|Z_i=0),\\
	&Pr(Y_i(1)=1|Z_i=1)\le Pr(Y_i=1|Z_i=1)+ \frac{m^{EL,min}(F)}{Pr(Z_i=1)}.\\
	\end{split}
	\]
	To show the display in the corollary is sharp, it suffices to show that the bounds can be achieved by some $G^{s^*}$ in the identified set. The $G^{s^*}$ that can achieve the bounds are listed in the following table. The column title denotes the bound we try to achieve in Corollary \ref{Corollary: Roy model, ID set of $Pr(Y_i(1)=1|Z_i=z)$}, and column entries specify the value of $C^{y_1,y_0,z}_d$ that will achieve the bound. 
	
	\begin{table}[H]
		\footnotesize
		\begin{tabular}{lll}
			& $Pr_F(Y_i=1,D_i=1|Z_i=0)= Pr_G(Y_i(1)=1|Z_i=0)$ & $Pr_F(Y_i=1|Z_i=0)=Pr_G(Y_i(1)=1|Z_i=0)$ \\ \hline
			$C^{000}_0$ & $Pr_F(Y_i=0,Z_i=0)/2$                       & $Pr_F(Y_i=0,Z_i=0)/2$                \\
			$C^{000}_1$ & $Pr_F(Y_i=0,Z_i=0)/2$                       & $Pr_F(Y_i=0,Z_i=0)/2$                \\
			$C^{010}_0$ & $Pr_F(Y_i=1,D_i=0,Z_i=0)$                   & 0                                    \\
			$C^{010}_1$ & 0                                           & 0                                    \\
			$C^{100}_0$ & 0                                           & 0                                    \\
			$C^{100}_1$ & $Pr_F(Y_i=1,D_i=1,Z_i=0)$                   & $Pr_F(Y_i=1,D_i=1,Z_i=0)$            \\
			$C^{110}_0$ & 0                                           & $Pr_F(Y_i=1,D_i=0,Z_i=0)$            \\
			$C^{110}_1$ & 0                                           & 0                                    \\
			$C^{001}_0$ & $Pr_F(Y_i=0,Z_i=0)/2-m^{EL,min}(F)/2$          & $Pr_F(Y_i=0,Z_i=0)/2-m^{EL,min}(F)/2$   \\
			$C^{001}_1$ & $Pr_F(Y_i=0,Z_i=0)/2-m^{EL,min}(F)/2$          & $Pr_F(Y_i=0,Z_i=0)/2-m^{EL,min}(F)/2$   \\
			$C^{011}_0$ & 0                                           & 0                                    \\
			$C^{011}_1$ & 0                                           & 0                                    \\
			$C^{101}_0$ & $m^{EL,min}(F)$                                & $m^{EL,min}(F)$                         \\
			$C^{101}_1$ & 0                                           & 0                                    \\
			$C^{111}_0$ & $Pr_F(Y_i=1,D_i=0,Z_i=1)$                   & $Pr_F(Y_i=1,D_i=0,Z_i=1)$            \\
			$C^{111}_1$ & $Pr_F(Y_i=1,D_i=1,Z_i=1)$                   & $Pr_F(Y_i=1,D_i=1,Z_i=1)$           
		\end{tabular}
		
	\end{table}
	\begin{table}[H]
		\footnotesize
		\begin{tabular}{lll}
			& $Pr_F(Y_i=1,D_i=1|Z_i=1)= Pr_G(Y_i(1)=1|Z_i=1)$ & $Pr_F(Y_i=1|Z_i=1)+ \frac{m^{EL,min}(F)}{Pr(Z_i=1)}=Pr_G(Y_i(1)=1|Z_i=1)$ \\ \hline
			$C^{000}_0$ & $Pr_F(Y_i=0,Z_i=0)/2$                       & $Pr_F(Y_i=0,Z_i=0)/2$                                              \\
			$C^{000}_1$ & $Pr_F(Y_i=0,Z_i=0)/2$                       & $Pr_F(Y_i=0,Z_i=0)/2$                                              \\
			$C^{010}_0$ & $Pr_F(Y_i=1,D_i=0,Z_i=0)$                   & $Pr_F(Y_i=1,D_i=0,Z_i=0)$                                          \\
			$C^{010}_1$ & 0                                           & 0                                                                  \\
			$C^{100}_0$ & 0                                           & 0                                                                  \\
			$C^{100}_1$ & $Pr_F(Y_i=1,D_i=1,Z_i=0)$                   & $Pr_F(Y_i=1,D_i=1,Z_i=0)$                                          \\
			$C^{110}_0$ & 0                                           & 0                                                                  \\
			$C^{110}_1$ & 0                                           & 0                                                                  \\
			$C^{001}_0$ & $Pr_F(Y_i=0,Z_i=0)/2-m^{EL,min}(F)/2$          & $Pr_F(Y_i=0,Z_i=0)/2-m^{EL,min}(F)/2$                                 \\
			$C^{001}_1$ & $Pr_F(Y_i=0,Z_i=0)/2-m^{EL,min}(F)/2$          & $Pr_F(Y_i=0,Z_i=0)/2-m^{EL,min}(F)/2$                                 \\
			$C^{011}_0$ & $Pr_F(Y_i=1,D_i=0,Z_i=1)$                   & 0                                                                  \\
			$C^{011}_1$ & $m^{EL,min}(F)$                                & 0                                                                  \\
			$C^{101}_0$ & 0                                           & $m^{EL,min}(F)$                                                       \\
			$C^{101}_1$ & 0                                           & 0                                                                  \\
			$C^{111}_0$ & 0                                           & $Pr_F(Y_i=1,D_i=0,Z_i=1)$                                          \\
			$C^{111}_1$ & $Pr_F(Y_i=1,D_i=1,Z_i=1)$                   & $Pr_F(Y_i=1,D_i=1,Z_i=1)$                                         
		\end{tabular}
	\end{table}

\end{proof}

\pagebreak
\section*{Appendices and Auxiliary Results for Online Publication}

\section{Auxiliary Lemmas and Proofs} 

\subsection{Lemmas Used in the Main Proofs}
\begin{lem}\label{lem: existence of positive measure set }
	Let $p(y,1)$ and $q(y,1)$ be the Radon-Nikodym derivatives defined in (\ref{eq: R-D derivatives of P Q}). Suppose there exists a set $B_1$ such that $\mu_F(B_1)>0$ and $q(y,1)-p(y,1)>0$ $\forall y\in B_1$.
	Then there exists a measurable set $B_1'\subseteq B_1$ with $\mu_F(B_1')>0$ such that $\int_{B_1'} q(y,1)-p(y,1) d\mu_F>0 $.
\end{lem}

\begin{proof}
	Since the Radon-Nikodym derivatives are measurable functions, the level set 
	\[B_1^t=\{y\in \mathcal{Y}: q(y,1)-p(y,1)\ge t\}\]
	is measurable. Consider the sequence of nested level set $\{B_1^{1/n}\}_{n=1}^\infty$, and we have $\mathbbm{1}(y\in B_1^{1/n})\rightarrow_{p.w.} \mathbbm{1}(y\in B_1)$. The dominated convergence theorem implies $\mu_F(B_1^{1/n})\rightarrow \mu_F(B_1)$. 
	
	Suppose there is no $\mu_F$-positively measured set $B_1'$ that satisfies the condition $\int_{B_1'} q(y,1)-p(y,1) d\mu_F>0 $, then either of the following conditions hold
	\begin{enumerate}
		\item $\int_{B_1^{1/n}} q(y,1)-p(y,1) d\mu_F=0$ for all $n$;
		\item $\mu_F(B_1^{1/n})=0$ for all $n$.
	\end{enumerate}
	We first show condition 1 above implies condition 2. By condition 1 and the definition of $B_1^{-1/n}$, we have
	\[
	\int_{B_1^{1/n}} q(y,1)-p(y,1) d\mu_F \ge \frac{1}{n}\mu_F(B_1^{1/n}).
	\]
 If  $\int_{B_1^{1/n}} q(y,1)-p(y,1) d\mu_F=0$ holds, then $\mu_F(B_1^{1/n})=0$ must hold for all $n$.
 
 Now, we can take the limit of $n\rightarrow \infty$ and use the dominated convergence theorem to show $\mu_F(B_1)=\lim_{n\rightarrow \infty} \mu_F(B_1^{-1/n})=0 $. This contradicts $\mu_F(B_1)>0$. 
	
\end{proof}

\begin{lem} \label{lem: auxiliary, criteria for exclusion restriction}
	Let $X,Y,Z$ be real random variables. Let $G$ be a probability measure of $X,Y$ and let $F$ the distribution of a real random variable $Z$. Suppose
	\[
	Pr_G(X\in B_x,Y\in B_y)= Pr_F(Z\in B_x\cap B_y)
	\]
	holds for all measurable set  $B_x,B_y$, then $X=Y$ holds $G$-a.s..
\end{lem}  
\begin{proof}
  	Let $\{x_n\}_{n=1}^\infty$ and $\{y_m\}_{m=1}^\infty$ be the exhausting lists of rational numbers. Then 
  	\[
  	\{(x,y):x\ne y\} \subseteq [\cup_{x_n>y_m} (x_n,\infty)\times (-\infty,y_m)]\cup [\cup_{x_n<y_m} (-\infty,x_n)\times (y_n,\infty)].
  	\]
	Indeed, for any pair of $(x,y)$, if $x>y$  and $x-y=\epsilon>0$, we can find rational numbers $x_n$ and $y_m$ such that $
	x>x_n-\epsilon/3>y_m+\epsilon/3>y
	$
	and $(x,y)\in (x_n,\infty)\times (-\infty,y_m)$. Symmetric arguments hold for $x<y$. Then 
	{\footnotesize
	\[
	\begin{split}
		Pr_{G}(X\ne Y)&= Pr_{G}((X,Y)\in \{\{(x,y):x\ne y\}\})\\
		&\le \sum_{x_n>y_m} Pr_{G}(X\in (x_n,\infty),Y_i\in (-\infty,y_m))\\
		&+\sum_{x_n<y_m} Pr_{G}(X\in (-\infty,x_n),Y_i\in (y_m,\infty))\\
		&=_{(1)}\sum_{x_n>y_m}Pr_F(Z\in (x_n,\infty)\cap(-\infty,y_m))\\
		&+ \sum_{x_n<y_m} Pr_F(Z\in (-\infty,x_n)\cap (y_m,\infty))\\
		&=_{(2)}\sum_{x_n<y_m} 0 +\sum_{x_n>y_m} 0=_{(3)}0,\\
	\end{split}
	\]}
	where equality (1) holds by the assumption of the Lemma \ref{lem: auxiliary, criteria for exclusion restriction}, (2) holds because $ (x_n,\infty)\cap(-\infty,y_m)=\varnothing $ when $x_n>y_m$, (3) holds because the summation over  $x_n<y_m$ and $x_n>y_m$ is countable, and a countable summation of zero is zero.

\end{proof}

\begin{lem}
	(Gine and Guillou, 2002) Let $\mathcal{G}$ be a measurable uniformly bounded VC class of functions, such that 
	\[N(\mathcal{G},L_2(P),\tau ||G||_{L_2(P)})\le \left(\frac{A}{\tau }^v\right),\]
	and let $\sigma$ and $U$ be the number such that $\sigma^2\ge \sup_{g\in\mathcal{G}} Var_p g$ and $U\ge \sup_{g\in \mathcal{G}} ||g||_\infty$, and $0<\sigma < U/2$, $\sqrt{n}\sigma\ge U\sqrt{\frac{U}{\log \sigma}}$. Then there exist constant $L$, $C$ that depends on $A$ and $v$ only such that 
	\begin{equation} \label{eq: maximal inequality from Gine and Guillou}
	\begin{split}
	Pr \Bigg( &\sup_{g\in\mathcal{G}}\left| \sum_{i=1}^n g(x_i)-Eg(x_i)\right|>C\sigma \sqrt{n}\sqrt{\log\frac{U}{\sigma}} \Bigg)\\
	&\le L \exp \left\{-\frac{C\log (1+C/(4L))}{L}\log \frac{U}{\sigma }\right\}.
	\end{split}
	\end{equation}
\end{lem}

\begin{lem}
	(Montgomery-Smith's Maximal Inequality)
	\begin{equation}
	Pr\Bigg( \max_{k\le n} sup_{g\in\mathcal{G}}\left| \sum_{i=1}^k g(x_i)-Eg(x_i)\right|>t \Bigg)\le 9 Pr\Bigg( sup_{g\in\mathcal{G}}\left| \sum_{i=1}^n g(x_i)-Eg(x_i)\right|>t/30 \Bigg).
	\end{equation}
\end{lem}

\begin{lem}\label{lem: strong convergence from Gine and Guillou}
	Let $h_n=n^{-\gamma}$ for some $\gamma\in (0,1)$, such that  $\frac{nh_n}{|\log h_n|}\rightarrow \infty$.  Denote the estimator and its expectation of the estimator as 
	\begin{equation}\label{eq: auxiliary lem: f_n and f_nbar}
	\begin{split}
	f^{l,m}_n(y)&\equiv \frac{1}{nh_n}\sum_{i=1}^n K\left(\frac{Y_i-y}{h_n}\right)\mathbbm{1}(D_i=l,Z_i=m)\\
	\bar{f}^{l,m}_n(y)&\equiv\frac{1}{h_n} E\left[K\left(\frac{Y_i-y}{h_n}\right)\mathbbm{1}(D_i=l,Z_i=m)\right ].
	\end{split}
	\end{equation}
	Then, the following uniform bounds holds for some constant $\bar{C}$:
	\begin{equation}
	\lim\sup_{n\rightarrow \infty} \left(\sup_y \sqrt{\frac{nh_n}{\log h_n^{-1}}}|f_n^{l,m}(y)-\bar{f}^{l,m}_n(y)|\right) \le \bar{C}\quad a.s..
	\end{equation}
\end{lem}
\begin{proof}
	We prove the a.s. convergence result for $l=m=1$ and omit the superscript $l,m$ in $f_n$ and $\bar{f}_n$, and the rest inequalities hold similarly. Use Montgomery-Smith's Maximal inequality, we have 
	\begin{equation}
	\begin{split}
	&Pr\Bigg( \max_{2^{k-1}\le n\le 2^k } \sqrt{\frac{nh_n}{\log h_n^{-1}}}\sup_{y}|f_n(y)-\bar{f}_n(y)| >t\Bigg) \\
	&= Pr\Bigg( \max_{2^{k-1}\le n\le 2^k } \sup_{y}\bigg|\sum_{i=1}^n K\left(\frac{Y_i-y}{h_n}\right)\mathbbm{1}(D_i=1,Z_i=1)\\
	\quad \quad &-E\left[K\left(\frac{Y_i-y}{h_n}\right)\mathbbm{1}(D_i=1,Z_i=1)\right ]\bigg|>t\sqrt{nh_n\log  h_n^{-1}} \Bigg)\\
	&\le 9 Pr \Bigg( \sup_{y, h_{2^{k-1}}\le h\le h_{2^k}} \bigg| \sum_{i=1}^{2^k}K\left(\frac{Y_i-y}{h}\right)\mathbbm{1}(D_i=1,Z_i=1)\\
	&- E\left[K\left(\frac{Y_i-y}{h}\right)\mathbbm{1}(D_i=1,Z_i=1)\right ]\bigg|>t\sqrt{2^{k-1} h_{2^{k}}\log  h_{2^{k}}^{-1}}/30 \Bigg) .
	\end{split}
	\end{equation}
	
	By Gine and Guillou (2002), the class of function $\mathcal{K}_k=\{K(\frac{t-y}{h})\big| t\in R, h_{2^k}\le h\le h_{2^{k-1}}  \}$ is a VC class, and we multiply it by a fixed function $\mathbbm{1}(d=1,z=1)$, the class of function $\tilde{\mathcal{K}}_k=\{K(\frac{t-y}{h})\mathbbm{1}(d=1,z=1)\big| t\in R, h_{2^k}\le h\le h_{2^{k-1}}\}$ is still a VC class. So we take $U_k= ||K(y)||_\infty$, $\sigma_k^2=h_{2^{k-1}} \sup_y |f(y|D_i=1,Z_i=1)|\int_t K^2(t)dt$, then we have 
	\begin{equation*}
	\begin{split}
	&\sup_{y, h_{2^k}\le h\le h_{2^{k-1}}} Var\left[ K(\frac{Y_i-y}{h})\mathbbm{1}(D_i=1,Z_i=1)\right]\\
	&\le \sup_{y, h_{2^k}\le h\le h_{2^{k-1}}} E\left[ K^2(\frac{Y_i-y}{h})\mathbbm{1}(D_i=1,Z_i=1)\right]\\
	&\le \sup_{y, h_{2^k}\le h\le h_{2^{k-1}}} h \int_t K^2(t) f(y-th|D_i=1,Z_i=1)dt \times Pr(D_i=1,Z_i=1)\\
	&\le h ||f||_\infty \int_t K^2(t)dt\le \sigma_k^2,
	\end{split}
	\end{equation*}
	and 
	$\sup_{y, h_{2^k}\le h\le h_{2^{k-1}}} \left|  K(\frac{Y_i-y}{h})\mathbbm{1}(D_i=1,Z_i=1)\right|\le U_k$. 
	
	Then since $h_{2^{k-1}}\rightarrow 0$ and $\frac{2^kh_{2^k}}{\log h_{2^k}^{-1}}\rightarrow \infty$ as $k\rightarrow \infty$, we can find $k_0$ such that $\sigma_k^2\le U_k/2$ and $\sqrt{2^k}\sigma_k\ge U_k\sqrt{\frac{U_k}{\log\sigma_k}}$ for all $k\ge k_0$. So, we can apply 
	(\ref{eq: maximal inequality from Gine and Guillou}). 
	
	Take $t=30C \sqrt{2\times 2^{-\gamma}\sup_y |f(y|D_i=1,Z_i=1)|\int_t K^2(t)dt}$, then 
	\begin{equation*}
	\begin{split}
	t\sqrt{2^{k-1} h_{2^{k}}\log  h_{2^{k}}^{-1}}/30&=C\sqrt{2^k (ch_{2^k})\log  h_{2^{k}}^{-1}{\sup_y |f(y|D_i=1,Z_i=1)|\int_t K^2(t)dt}}\\
	&\ge C\sqrt{2^k}\sqrt{ h_{2^{k-1}}\log  h_{2^{k}}^{-1}{\sup_y |f(y|D_i=1,Z_i=1)|\int_t K^2(t)dt}}\\
	&= C\sqrt{2^k} \sigma_{k} \sqrt{\log{h_{2^k}^{-1}}},
	\end{split}
	\end{equation*}
	Since 
	\[\frac{h_{2^k}^{-1}}{U_k/\sigma_k}\ge \frac{h_{2^k}^{-0.5}\sqrt{\sup_y |f(y|D_i=1,Z_i=1)|\int_t K^2(t)dt}}{||K(y)||_\infty }\rightarrow \infty,\]
	where we use the construction of $\sigma_k$ and $h_{2^k}^{-1}\ge h_{2^{k-1}}^{-1}$, and $h_n\rightarrow 0$. So we can find $k_1$ such that $ C\sqrt{2^k} \sigma_{k} \sqrt{\log{h_{2^k}^{-1}}}>  C\sqrt{2^k} \sigma_{k}\sqrt{\frac{U_k}{\sigma_k}}$ holds for all $k>k_1$. Then for $k>\max\{k_0,k_1\}$, 
	
	\begin{equation}
	\begin{split}
	&9 Pr \Bigg( \sup_{y, h_{2^{k-1}}\le h\le h_{2^k}} \bigg| \sum_{i=1}^{2^k}K\left(\frac{Y_i-y}{h}\right)\mathbbm{1}(D_i=1,Z_i=1)\\
	&- E\left[K\left(\frac{Y_i-y}{h}\right)\mathbbm{1}(D_i=1,Z_i=1)\right ]\bigg|>t\sqrt{2^{k-1} h_{2^{k}}\log  h_{2^{k}}^{-1}}/30 \Bigg)\\
	&\le 9 Pr \Bigg( \sup_{y, h_{2^{k-1}}\le h\le h_{2^k}} \bigg| \sum_{i=1}^{2^k}K\left(\frac{Y_i-y}{h}\right)\mathbbm{1}(D_i=1,Z_i=1)\\
	&- E\left[K\left(\frac{Y_i-y}{h}\right)\mathbbm{1}(D_i=1,Z_i=1)\right ]\bigg|>  C\sqrt{2^k} \sigma_{k}\sqrt{\frac{U_k}{\sigma_k}} \Bigg)\\
	&\le L \exp \left\{-\frac{C\log (1+C/(4L))}{L}\log \frac{U_k}{\sigma_k }\right\}\\
	&\le L \exp \left\{-\frac{C\log (1+C/(4L))}{L}\log \frac{||K(y)||_\infty}{\sqrt{\sup_y |f(y|D_i=1,Z_i=1)|\int_t K^2(t)dt} }\right\} h_{2^{k-1}}\\
	&\le Constant\times \left(\frac{1}{2^\gamma}\right)^{k-1}.
	\end{split}
	\end{equation}
	Note that $\sum_{k=\max\{k_0,k_1\}+1}^\infty \left(\frac{1}{2^\gamma}\right)^{k-1} <\infty$ holds, so by Borel-Cantelli lemma, 
	\begin{equation}
	Pr\Bigg(\lim\sup_{n\rightarrow \infty }\sqrt{\frac{nh_n}{\log h_n^{-1}}}\sup_{y}|f_n(y)-\bar{f}_n(y)| >30C \sqrt{2\times 2^{-\gamma}\sup_y |f(y|D_i=1,Z_i=1)|\int_t K^2(t)dt}\Bigg)=0.
	\end{equation}
\end{proof}

\subsection{Proof of Proposition \ref{prop: continuity of the identified set}}\label{Proof of continuous id set condition}
\begin{proof}
	Recall the definition of the identified set from Definition \ref{def: Identification System}:
	\[
	\begin{split}
	\Theta_{\tilde{A}}^{ID}(F)&=\{\theta(s): \quad F\in M^s(G^s),\,\, s\in\tilde{A}\}\\
	&=\{\theta(s):s\in \cap_{l\ne j}A_l,\,\, F\in M^s(G^s),\,\,and \,\,m_j(s)=\min(F;m_j)\}\\
	&= \theta\circ m_j^{-1}\circ \min(F;m_j)
	\end{split}
	\]
	where the second equality holds by the construction of the minimal deviation extension in Definition \ref{def: minimal deviation extension}.
	
	Since $\theta:\cap_{l\ne j}A_l\rightarrow \Theta$ and $\min(F;m_j):\mathcal{F}\rightarrow \mathbb{R}$ are continuous functions, by applying Lemma \ref{lem: appen, composition of hemicontinuous correspondence} twice, we see that $Theta_{\tilde{A}}^{ID}(F)$  is an upper (resp. lower) hemicontinuous correspondence if $m_j^{-1}:\mathbb{R}\rightarrow \cap_{l\ne j}A_l$ is an upper (resp. lower) hemicontinuous correspondence.

\end{proof}
\begin{lem}\label{lem: appen, composition of hemicontinuous correspondence}
	Let $(X,\tau_X)$, $(Y,\tau_Y)$ and $(Z,\tau_Z)$ be three topological spaces. Let $h_{xy}:X\rightarrow Y$ and $h_{zx}:Z\rightarrow X$ be continuous functions. If $h_{yz}: Y \rightrightarrows Z$ is an upper hemicontinuous (resp. lower hemicontinuous) correspondence, then $h_{yz}\circ h_{xy}$ and $h_{yz}\circ h_{zx}$ are both uppe hemicontinuous (resp. lower hemicontinuous) correspondences.
\end{lem}

\begin{proof}	
	\textit{ 1. $h_{yz}\circ h_{xy}$ is upper hemicontinuous.}
	
	Let $O_z$ be any open set covering $h_{yz}\circ h_{xy}(x)$. Since $h_{yz}$ is upper hemicontinuous, then by definition there exists an open set $O_y$ containing $h_{xy}(x)$ such that for all $\tilde{y}\in O_y$, we have $h_{yz}(\tilde{y})\subset  O_z$. By continuity of $h_{xy}$, $O_x\equiv h_{xy}^{-1}(O_y)$ is an open set, and for any $\tilde{x}\in O_x$, $h_{yz}\circ h_{xy}(\tilde{x})\subset O_z$. By definition, $h_{yz}\circ h_{xy}$ is upper hemicontinuous.
	
	\textit{ 2. $h_{zx}\circ h_{yz}$ is upper hemicontinuous.}
	
	Let $O_x$ be any open set covering $h_{zx}\circ h_{yz}(y)$. Since $h_{zx}$ is continuous, $O_z\equiv h_{zx}^{-1}(O_x)$ is an open set containing $h_{yz}(y)$. Since $h_{yz}$ is upper hemicontinuous, we can find an open set $O_y$ containing $y$ such that for any $\tilde{y}\in O_y$, $h_{yz}(\tilde{y})\subset O_z$ holds. We can then conclude \[
	h_{zx}\circ h_{yz}(\tilde{y})\subset h_{zx}(O_z)\subset O_x.
	\]
	By definition, $h_{zx}\circ h_{yz}$ is upper hemicontinuous.
	
	\textit{ 3. $h_{yz}\circ h_{xy}$ is upper hemicontinuous.}
	
	Let $O_z$ be any open set such that  $O_z\cap h_{yz}\circ h_{xy}(x)\ne \varnothing$. Since $h_{yz}$ is lower hemicontinuous, then by definition there exists an open set $O_y$ containing $h_{xy}(x)$ such that for all $\tilde{y}\in O_y$, we have $O_z\cap h_{yz}(\tilde{y})\ne \varnothing$. By continuity of $h_{xy}$, $O_x\equiv h_{xy}^{-1}(O_y)$ is an open set, and for any $\tilde{x}\in O_x$, we have $ h_{xy}(\tilde{x})\subset O_y$. Therefore $h_{yz}\circ h_{xy}(\tilde{x})\cap O_z\ne \varnothing$. By definition, $h_{yz}\circ h_{xy}$ is lower hemicontinuous.
	
	\textit{ 4. $h_{zx}\circ h_{yz}$ is lower hemicontinuous.}
	
	Let $O_x$ be any open set such that  $O_x\cap h_{zx}\circ h_{yz}(y)\ne \varnothing$. Since $h_{zx}$ is continuous, $O_z\equiv h_{zx}^{-1}(O_x)$ is an open set containing $h_{yz}(y)$. Since $h_{yz}$ is lower hemicontinuous, we can find an open set $O_y$ containing $y$ such that for any $\tilde{y}\in O_y$, $h_{yz}(\tilde{y})\cap O_z\ne \varnothing$ holds. We can then conclude \[
	\varnothing \ne h_{zx}\circ h_{yz}(\tilde{y})\cap h_{zx}(O_z)\subset h_{zx}\circ h_{yz}(\tilde{y})\cap O_x.
	\]
	By definition, $h_{zx}\circ h_{yz}$ is upper hemicontinuous.
\end{proof}

\subsection{Proof of Propositions \ref{prop: minimal marg ind as consistent extension} and \ref{prop: minimal marg diff as LATE-consistent}}

\subsubsection{Lemmas}
\begin{lem}\label{lem: minimal marginal independence construction}
	Let $F$ be any distribution of outcome, and let $p(y,d),q(y,d)$ be the Radon-Nikodym derivatives with respect to $\mu_F$. Consider the following $G^s$:
	
	\begin{equation}
	\begin{split}
	&Pr_{G^s}(Y_i(d,z)\in B_{dz}\quad \forall d,z\in\{0,1\}, D_i(1)=1,D_i(0)=1|Z_i=z)\\
	&=\begin{cases}
	G^a(Y_i(0,0)\in B_{00}\cap B_{01})\times \int_{B_{10}\cap B_{11}} \min\{p(y,1),q(y,1)\} d\mu_F(y)\quad &if \quad z=1,\\
	G^a(Y_i(0,0)\in B_{00}\cap B_{01})\times \int_{B_{10}\cap B_{11}} q(y,1)d\mu_F(y)\quad &if \quad z=0,
	\end{cases}
	\end{split}
	\end{equation}
	where $G^a$ is any probability measure, and
	
	\begin{equation}
	\begin{split}
	&Pr_{G^s}(Y_i(d,z)\in B_{dz}\quad \forall d,z\in\{0,1\}, D_i(1)=0,D_i(0)=0|Z_i=z)\\
	&=\begin{cases}
	G^n(Y_i(0,0)\in B_{00}\cap B_{01})\times \int_{B_{10}\cap B_{11}} \min\{p(y,0),q(y,0)\} d\mu_F(y)\quad &if \quad z=0,\\
	G^n(Y_i(0,0)\in B_{00}\cap B_{01})\times \int_{B_{10}\cap B_{11}} p(y,0)d\mu_F(y)\quad &if \quad z=1,
	\end{cases}
	\end{split}
	\end{equation}
	where $G^n$ is any probability measure. Let 
	\begin{equation}
	\begin{split}
	Pr(D_i(1)=1,D_i(0)=0|Z_i=1)=P(\mathcal{Y}_1,1)-Q(\mathcal{Y}_1,1)\\
	Pr(D_i(1)=1,D_i(0)=0|Z_i=0)=Q(\mathcal{Y}_0,0)-P(\mathcal{Y}_0,0),\\
	\end{split}
	\end{equation}
	and let:
	\begin{equation}
	\begin{split}
	&Pr_{G^s}(Y_i(d,z)\in B_{dz}\quad \forall d,z\in\{0,1\}| D_i(1)=1,D_i(0)=0,Z_i=z),\\
	&=\frac{\int_{B_{00}\cap B_{01}}\min\{q(y,0)-p(y,0),0\} d\mu_F(y) \times \int_{B_{10}\cap B_{11}}\min\{p(y,1)-q(y,1),0\} d\mu_F(y) }{(P(\mathcal{Y}_1,1)-Q(\mathcal{Y}_1,1))(Q(\mathcal{Y}_0,0)-P(\mathcal{Y}_0,0))},
	\end{split}
	\end{equation}
	
	\begin{equation}
	Pr_{G^s}(Y_i(d,z)\in B_{dz}\quad \forall d,z\in\{0,1\}, D_i(1)=0,D_i(0)=1|Z_i=z)\equiv 0.
	\end{equation}
	Then the following results hold: (1). $G$ is a probability measure
 (2). $G^s\in A^{ER}\cap A^{ND}\cap A^{TI-CP}$; (3). $F\in M^s(G^s)$; (4). $m^{MD}(s)=m^{min}(F)$, where $m^{min}(F)$ is defined in Assumption \ref{assump: minimal dist to marg ind inst}.
\end{lem}
\begin{proof}
	I first check that $G$ is a probability measure. \begin{equation*}
	\begin{split}
	&\quad \sum_{d_1,d_0\in\{0,1\}} Pr_{G^s}(Y_{i}(d,1)\in \mathcal{Y},D_i(1)=d_1,D_i(0)=d_0|Z_i=1)\\
	&= \int_{\mathcal{Y}}  \min\{p(y,1),q(y,1)\}  d\mu_F(y)+ \int_{\mathcal{Y}}	p(y,0) d\mu_F(y)+ (P(\mathcal{Y}_1,1)-Q(\mathcal{Y}_1,1))\\
	&= (P(\mathcal{Y}_1^c,1)+Q(\mathcal{Y}_1,1))+P(\mathcal{Y},0)+(P(\mathcal{Y}_1,1)-Q(\mathcal{Y}_1,1))\\
	&=P(\mathcal{Y},1)+P(\mathcal{Y},0)=1.
	\end{split}
	\end{equation*} 
	We can check that the measure sum up to one for $Z_i=0$. This checks $G$ is a probability measure.
	
	Checking $F\in M^s(G^s)$ is similar to the proofs in Lemma
	\ref{lem: minimal defiers construction for type ind instru}. The type independence for compliers, `No Defiers' assumptions hold for $G^s$ by construction. Exclusion restriction holds by Lemma \ref{lem: auxiliary, criteria for exclusion restriction}.\footnote{See Lemma
		\ref{lem: minimal defiers construction for type ind instru} for the procedures for proof of this statement.}
	
	We now show a lower bound for the $m^{min}(F)$. Let $s^*$ be any structure in  $A^{ER}\cap A^{ND}\cap A^{TI-CP}$ and $F\in M^{s^*}(G^{s^*})$ . We use the following decomposition:
	\begin{equation}
	\begin{split}
	P( B_1,1)&=Pr_{G^{s^*}} (Y_i(1,1)\in B_1,D_i(1)=1,D_i(0)=1|Z_i=1)\\
	&+Pr_{G^{s^*}} (Y_i(1,1)\in B_1,D_i(1)=1,D_i(0)=0|Z_i=1),\\
	Q( B_1,1)&=Pr_{G^{s^*}} (Y_i(1,0)\in B_1, D_i(1)=1,D_i(0)=1|Z_i=0)\\
	&+Pr_{G^{s^*}} (Y_i(1,0)\in B_1, D_i(1)=0,D_i(0)=1|Z_i=0)\\
	&=_{(1)}Pr_{G^{s^*}} (Y_i(1,0)\in B_1,D_i(1)=1,D_i(0)=1|Z_i=0),
	\end{split}
	\end{equation}
	where equality (1) follows by the `No Defiers' condition.  Take the Radon-Nikodym derivatives with respect to $\mu_F$ on both sides to get 
	\[
	\begin{split}
	p(y,1)&= g^{s^*}_{y_{11}}(y,1,1|Z_i=1)+ g^{s^*}_{y_{10}}(y,1,0|Z_i=1)\\
	q(y,1)&= g^{s^*}_{y_{10}}(y,1,1|Z_i=0).
	\end{split}
	\]
	Take the difference between $p(y,1)$ and $q(y,1)$ to get 
	\begin{equation}\label{eq: lem marginal ind cont, marg diff for at}
	g^{s^*}_{y_{11}}(y,1,1|Z_i=1)-g^{s^*}_{y_{10}}(y,1,1|Z_i=0)=p(y,1)-q(y,1)-g^{s^*}_{y_{10}}(y,1,0|Z_i=1).
	\end{equation}
	Since $g^{s^*}_{y_{10}}(y,1,0|Z_i=1)\ge 0$ \footnote{ Note that $(x-t)^2\ge (\max\{-x,0\})^2$ when $t\ge 0$ holds} , we have 
	\[\left[g^{s^*}_{y_{11}}(y,1,1|Z_i=1)-g^{s^*}_{y_{10}}(y,1,1|Z_i=0)\right]^2\ge\max\{-(p(y,1)-q(y,1)),0\}^2. \]
	Similarly, using the decomposition of $Q(B,0)$ and $P(B,0)$, we have
	\[\left[g^{s^*}_{y_{00}}(y,0,0|Z_i=0)-g^{s^*}_{y_{01}}(y,0,0|Z_i=1)\right]^2\ge\max\{-(q(y,0)-p(y,0)),0\}^2. \]
	So the measure of deviation from marginal independence equals:
	\[
	\begin{split}
	m^{MI}(s^*)&= \int \left[g^{s^*}_{y_{11}}(y,1,1|Z_i=1)-g^{s^*}_{y_{10}}(y,1,1|Z_i=0)\right]^2\\
	&+\left[g^{s^*}_{y_{00}}(y,0,0|Z_i=0)-g^{s^*}_{y_{01}}(y,0,0|Z_i=1)\right]^2 d\mu_F(y)\\
	&\ge \int \max\{-(p(y,1)-q(y,1)),0\}^2+\max\{-(q(y,0)-p(y,0)),0\}^2 d\mu_F(y)\\
	&= m^{MI}(s),
	\end{split}
	\]
	where the last equality holds by construction of $G^s$. So this shows that $s$ achieves $m^{min}(F)$.
\end{proof}

\begin{lem}\label{lem: minimal marginal difference construction}
	Let $F$ be any distribution of outcome, and let $p(y,d),q(y,d)$ be the Radon-Nikodym derivatives with respect to $\mu_F$. Let 
	\begin{equation}
	\begin{split}
	&Pr_{G^s}(D_i(1)-D_i(0)=1|Z_i=1)=P(\mathcal{Y}_1,1)-Q(\mathcal{Y}_1,1),\\
	&Pr_{G^s}(D_i(1)-D_i(0)=1|Z_i=0)=Q(\mathcal{Y}_0,0)-P(\mathcal{Y}_0,0),\\
	&Pr_{G^s}(D_i(1)=D_i(0)=1|Z_i=0)=Q(\mathcal{Y},1),\\
	&Pr_{G^s}(D_i(1)=D_i(0)=1|Z_i=1)=P(\mathcal{Y},1)-Pr_{G^s}(D_i(1)-D_i(0)=1|Z_i=1),\\
	&Pr_{G^s}(D_i(1)=D_i(0)=0|Z_i=1)=P(\mathcal{Y},0)\\
	&Pr_{G^s}(D_i(1)=D_i(0)=1|Z_i=0)=Q(\mathcal{Y},0)-Pr_{G^s}(D_i(1)-D_i(0)=1|Z_i=0).\\
	\end{split}
	\end{equation}
	Consider the following  $G^s$:
	\begin{equation}
	\begin{split}
	&Pr_{G^s}(Y_{dz}\in B_{dz} \quad \forall d,z\in\{0,1\}|D_i(1)=1,D_i(0)=1,Z_i=z)\\
	&=\frac{G^a(Y_i(0,0)\in B_{00})\times G^a(Y_i(0,1)\in B_{01})\times \int_{B_{10}} q(y,1) d\mu_F(y) \times \int_{B_{11}}\min\{p(y,1),q(y,1)\}d\mu_F(y) }{ Pr_{G^s}(D_i(1)=D_i(0)=1|Z_i=0)\times Pr_{G^s}(D_i(1)=D_i(0)=1|Z_i=1) },
	\end{split}
	\end{equation}
	where $G^a$ is any probability distribution, and
	\begin{equation}
	\begin{split}
	&Pr_{G^s}(Y_{dz}\in B_{dz} \quad \forall d,z\in\{0,1\}|D_i(1)=0,D_i(0)=0,Z_i=z)\\
	&=\frac{\int_{B_{01}} p(y,0) d\mu_F(y) \times \int_{B_{00}}\min\{p(y,0),q(y,0)\}d\mu_F(y)\times G^n(Y_i(1,0)\in B_{10})\times G^n(Y_i(1,1)\in B_{11}) }{ Pr_{G^s}(D_i(1)=D_i(0)=0|Z_i=0)\times Pr_{G^s}(D_i(1)=D_i(0)=0|Z_i=1) },
	\end{split}
	\end{equation}
	where $G^n$ is any probability distribution, and
	\begin{equation}
	\begin{split}
	&Pr_{G^s}(Y_{dz}\in B_{dz} \quad \forall d,z\in\{0,1\}|D_i(1)=1,D_i(0)=0,Z_i=z)\\
	&=\frac{\int_{B_{01}\cap B_{00}} \max\{q(y,0)-p(y,0),0\} d\mu_F(y) \times \int_{B_{10}\cap B_{11}}\min\{p(y,1),q(y,1)\}d\mu_F(y) }{ Pr_{G^s}(D_i(1)-D_i(0)=1|Z_i=0)\times Pr_{G^s}(D_i(1)-D_i(0)=1|Z_i=1) },
	\end{split}
	\end{equation}
	and 
	\[Pr_{G^s}(Y_{dz}\in B_{dz} \quad \forall d,z\in\{0,1\},D_i(1)=0,D_i(0)=1|Z_i=z)=0.\]
	Then the following results hold: (1). $G$ is a probability measure
	(2). $G^s\in A^{ER-CP}\cap A^{ND}\cap A^{TI}$; (3). $F\in M^s(G^s)$; (4). $m^{MD}(s)=m^{min}(F)$, where $m^{min}(F)$ is defined in Assumption \ref{assump: minimal dist to zero marg diff }.
\end{lem}

\begin{proof}
 Checking conditions (1)-(3) in this Lemma is similar to the proof in Lemma \ref{lem: appen, composition of hemicontinuous correspondence}. 
 
 Now I check the minimal deviation condition. Let $s^*$ be any structure that satisfies $F\in M^{s^*}(G^{s^*})$ and $s^*\in A^{ER-CP}\cap A^{ND}\cap A^{TI}$. Note that
	\begin{equation} \label{eq: P,Q identity for AT}
	\begin{split}
	P( B_1,1)&=Pr_{G^{s^*}} (Y_i(1,1)\in B_1,D_i(1)=1,D_i(0)=1|Z_i=1)\\
	&+Pr_{G^{s^*}} (Y_i(1,1)\in B_1,D_i(1)=1,D_i(0)=0|Z_i=1),\\
	Q( B_1,1)&=Pr_{G^{s^*}} (Y_i(1,0)\in B_1, D_i(1)=1,D_i(0)=1|Z_i=0)\\
	&+Pr_{G^{s^*}} (Y_i(1,0)\in B_1, D_i(1)=0,D_i(0)=1|Z_i=0)\\
	&=_{(1)}Pr_{G^{s^*}} (Y_i(1,0)\in B_1,,D_i(1)=1,D_i(0)=1|Z_i=0),
	\end{split}
	\end{equation} 
	where equality (1) follows by the `No Defiers' assumption. Taking Radon-Nikodym derivatives on both side, and take difference between $p(y,1)$ and $q(y,1)$ to get 
	\begin{equation}\label{eq: lem marg diff const, p-q}
	\begin{split}
	p(y,1)-q(y,1)-g^{s^*}_{y_{11}}(y,1,0|Z_i=1)= g^{s^*}_{y_{11}}(y,1,1|Z_i=1)-g^{s^*}_{y_{10}}(y,1,0|Z_i=0).
	\end{split}
	\end{equation}
	Since $g^{s^*}_{y_{11}}(y,1,1|Z_i=1)\ge 0$, 
	\begin{equation}\label{eq: append, for derivation of p-q in minimal marginal diff}
	\begin{split}
	&\quad\left[g^{s^*}_{y_{11}}(y,1,1|Z_i=1)-g^{s^*}_{y_{10}}(y,1,1|Z_i=0)\right]^2\\
	&\ge \max\{-(p(y,1)-q(y,1)),0\}^2.
	\end{split}
	\end{equation}
	Similarly, we have 
	\[
	\left[g^{s^*}_{y_{00}}(y,0,0|Z_i=0)-g^{s^*}_{y_{01}}(y,0,0|Z_i=1)\right]^2\ge \max\{-(q(y,1)-p(y,1)),0\}^2.
	\]
	Therefore 
	\[
	\begin{split}
	m^{MI}(s^*)&= \int 	\left[g^{s^*}_{y_{11}}(y,1,1|Z_i=1)-g^{s^*}_{y_{10}}(y,1,1|Z_i=0)\right]^2d\mu_F(y)\\
	&+\int\left[g^{s^*}_{y_{00}}(y,0,0|Z_i=0)-g^{s^*}_{y_{01}}(y,0,0|Z_i=1)\right]^2 d\mu_F(y)\\
	&\ge \int \max\{-(p(y,1)-q(y,1)),0\}^2d\mu_F(y)\\
	&+\int\max\{-(q(y,0)-p(y,0)),0\}^2 d\mu_F(y)\\
	&= m^{MD}(s),
	\end{split}
	\]
	where the last equality holds by construction of $s$, so $s$ achieves $m^{min}(F)$. 
\end{proof}

\subsubsection{Main Proof of Propositions \ref{prop: minimal marg ind as consistent extension} and \ref{prop: minimal marg diff as LATE-consistent}}
\begin{proof}
	Lemma \ref{lem: minimal marginal independence construction} and \ref{lem: minimal marginal difference construction}, $m^{MD}$ and $m^{MI}$ are well defined extensions. So it suffices to check the condition of $\theta$-consistency in Proposition \ref{prop: minimal deviation as strong extension}. We break the proof into two steps.
	
	\paragraph{Step 1. Closed form expressions for LATE under $\tilde{A}$ in Assumption \ref{assump: minimal dist to marg ind inst} or \ref{assump: minimal dist to zero marg diff }}. We claim the identified LATE under $\tilde{A}$ is 
	\begin{equation}\label{eq: appen, expression using p and q}
	{LATE}^{ID}_{\tilde{A}}(F)=\frac{\int_{\mathcal{Y}_1}{y (p(y,1)-q(y,1))}d\mu_F(y)}{P(\mathcal{Y}_1,1)-Q(\mathcal{Y}_1,1)}
	-\frac{\int_{\mathcal{Y}_0}{y (q(y,0)-p(y,0))}d\mu_F(y)}{Q(\mathcal{Y}_0,0)-P(\mathcal{Y}_0,0)}.
	\end{equation}
	We should note that (\ref{eq: appen, expression using p and q}) is the same  as the identified LATE in Proposition \ref{prop: identified under minimal defiers type ind inst}. 
	
	Suppose $\tilde{A}$ satisfies Assumption \ref{assump: minimal dist to marg ind inst}. By Lemma \ref{lem: minimal marginal independence construction}, the minimal marginal independence deviation is:
	\[
	\begin{split}
		m^{min}(F)= \int \left(\max\{(p(y,1)-q(y,1))^2,0\}+\max\{(q(y,0)-p(y,0))^2,0\}\right) d\mu_F(y).
	\end{split}
	\]
	For any $s^*\in \tilde{A}$ and $F\in M^{s^*}(G^{s^*})$, use equation (\ref{eq: lem marginal ind cont, marg diff for at}), we have
	\begin{equation*}
		\begin{split}
			(g^{s^*}_{y_{11}}(y,1,1|Z_i=1)-g^{s^*}_{y_{10}}(y,1,1|Z_i=0))^2&=[p(y,1)-q(y,1)-g^{s^*}_{y_{11}}(y,1,0|Z_i=1)]^2\\
			&\ge_{(*)} \max\{-p(y,1)+q(y,1),0\}^2,
		\end{split}
	\end{equation*}
	where the inequality $(*)$ holds with equality if and only if
	\[
	g^{s^*}_{y_{11}}(y,1,0|Z_i=1)=\begin{cases}
		0\quad &if\quad  y\in \mathcal{Y}_1^c,\\
		p(y,1)-q(y,1)\quad  &if \quad y\in \mathcal{Y}_1.
	\end{cases}
	\] 
	Similarly, 
	\begin{equation*}
		\begin{split}
			(g^{s^*}_{y_{01}}(y,0,0|Z_i=1)-g^{s^*}_{y_{00}}(y,0,0|Z_i=0))^2&=(q(y,0)-p(y,0)-g^{s^*}_{y_{00}}(y,1,0|Z_i=0))^2\\
			&\ge_{(**)} \max\{-q(y,0)+p(y,0),0\}^2,
		\end{split}
	\end{equation*}
	where the inequality $(**)$ holds with equality if and only if $g^{s^*}_{y_{00}}(y,1,0|Z_i=0)=\max\{q(y,0)-p(y,0),0\}$. Since $s^*$ achieves the $m^{min}(F)$, the two density conditions holds. Given the expression  of $g^{s^*}_{y_{11}}$ and $g^{s^*}_{y_{00}}$, the expression of LATE follows from the proof of Proposition \ref{prop: identified under minimal defiers type ind inst}.
	
	Suppose $\tilde{A}$ satisfies Assumption \ref{assump: minimal dist to zero marg diff }. By Lemma \ref{lem: minimal marginal difference construction}, the minimal marginal difference measure: 
	\[
	\begin{split}
		m^{min}(F)= \int \left(\max\{p(y,1)-q(y,1),0\}^2+\max\{q(y,0)-p(y,0),0\}^2\right) d\mu_F(y)
	\end{split}
	\] 
	For any $s^*\in \tilde{A}$  and $F\in M^{s^*}(G^{s^*})$, by equation (\ref{eq: lem marg diff const, p-q}), we have
	\begin{equation}
		\begin{split}
			(g^{s^*}_{y_{11}}(y,1,1|Z_i=1)-g^{s^*}_{y_{10}}(y,1,0|Z_i=0))^2&=(p(y,1)-q(y,1)-g^{s^*}_{y_{11}}(y,1,0|Z_i=1))^2\\
			&\ge \max\{-(p(y,1)+q(y,1),0\}^2,
		\end{split}
	\end{equation}
	where the inequality holds with equality if and only if
	\[
	g^{s^*}_{y_{11}}(y,1,0|Z_i=1)=\begin{cases}
		0\quad &if\quad  y\in \mathcal{Y}_1^c,\\
		p(y,1)-q(y,1)\quad  &if \quad y\in \mathcal{Y}_1.
	\end{cases}
	\] 
	Similarly, 
	\begin{equation*}
		\begin{split}
			(g^{s^*}_{y_{01}}(y,0,0|Z_i=1)-g^{s^*}_{y_{00}}(y,0,0|Z_i=0))^2&=(q(y,0)-p(y,0)-g^{s^*}_{y_{00}}(y,1,0|Z_i=0))^2\\
			&\ge \max\{-q(y,0)+p(y,0),0\}^2,
		\end{split}
	\end{equation*}
	where the inequality holds with equality if and only if $g^{s^*}_{y_{00}}(y,1,0|Z_i=0)=\max\{q(y,0)-p(y,0),0\}$. Since $s^*$ achieves the $m^{min}(F)$, the two density conditions hold. Given the expression  of $g^{s^*}_{y_{11}}$ and $g^{s^*}_{y_{00}}$. Last, we use the exclusion restrictions for the compliers to get the expression of LATE from the proof of Proposition \ref{prop: identified under minimal defiers type ind inst}.

	\paragraph{Step 2. The LATE-consistent result.} The expression (\ref{eq: appen, expression using p and q}) is the same as the identified LATE as in  Proposition \ref{prop: identified under minimal defiers type ind inst}. Since the minimal defiers extensions in  Proposition \ref{prop: identified under minimal defiers type ind inst} is a strong extension, (\ref{eq: appen, expression using p and q}) must equal the identified LATE under the IA-M assumption, whenever the IA-M assumption is not rejected by $F$. 
\end{proof}

\subsection{Proof of Proposition \ref{prop: m^{MD} is not well behaved if use full independence}}

\begin{proof}
	
 \textbf{Statement (I).}	 It suffices to show $A^{ND}\cap A^{FI}$ is refutable. Note that by  the `No Defiers' assumption, we have $
	Pr_F(D_i=1|Z_i=0)=Pr_{G^s}(D_i(1)=1,D_i(0)=1|Z_i=0)$ and 
	\[
	Pr_F(D_i=1|Z_i=1)=Pr_{G^s}(D_i(1)=1,D_i(0)=1|Z_i=1)+Pr_{G^s}(D_i(1)=1,D_i(0)=0|Z_i=1).
	\]
	By the independent instrument assumption $A^{FI}$, $Pr_{G^s}(D_i(1)=1,D_i(0)=1|Z_i=1)=Pr_{G^s}(D_i(1)=1,D_i(0)=1|Z_i=0)$, therefore the following must hold:
	\[
	0\le Pr_{G^s}(D_i(1)=1,D_i(0)=0|Z_i=0)= Pr_F(D_i=1|Z_i=1)-Pr_F(D_i=1|Z_i=0).
	\]
	So  $A^{ND}\cap A^{FI}$ implies $Pr_F(D_i=1|Z_i=1)-Pr_F(D_i=1|Z_i=0)\ge 0$ must hold for $F\in \cap_{s\in A^{ND}\cap A^{FI}} M^s(G^s)$. So there is no structure in $A^{ND}\cap A^{FI}$ that can rationalize $F_0$ such that $Pr_F(D_i=1|Z_i=1)-Pr_F(D_i=1|Z_i=0)< 0$. In particular, for this $F_0$, $\inf\{m_j(s): F\in M^s(G^s)\quad and \quad s\in A^{FI}\cap A^{ND}\}=\infty$.

	 \textbf{Statement (II).} Let $s^*\in A^{TI}\cap A^{ND}\cap A^{EM-C}$ and $F\in M^{s^*}(G^{s^*})$. Note that the derivation of (\ref{eq: append, for derivation of p-q in minimal marginal diff}) holds for all $s\in A^{TI}\cap A^{ND}$, therefore  $s^*$ must satisfy
	\begin{equation}\label{eq:  B.2.4 marginal constraint}
	\begin{split}
	\quad\left[g^{s^*}_{y_{11}}(y,1,1|Z_i=1)-g^{s^*}_{y_{10}}(y,1,1|Z_i=0)\right]^1
	\ge \max\{-(p(y,1)-q(y,1)),0\}^2,\\
	\left[g^{s^*}_{y_{00}}(y,0,0|Z_i=0)-g^{s^*}_{y_{01}}(y,0,0|Z_i=1)\right]^2\ge \max\{-(q(y,1)-p(y,1)),0\}^2.
	\end{split}
	\end{equation}
	So we have \begin{equation}\label{eq: append, not well behaved measure proof, minimal m^{MD}}
	m^{MD}(s^*)\ge \int \max\{-(p(y,1)-q(y,1)),0\}^2d\mu_F(y)
	+\int\max\{-(q(y,0)-p(y,0)),0\}^2 d\mu_F(y).
	\end{equation}
	Without loss of generality, we look at an observed distribution $F$ such that the corresponding $P$ and $Q$ satisfy: $P(\mathcal{Y}_1,1)-Q(\mathcal{Y}_1,1)\ge Q(\mathcal{Y}_0,0)-P(\mathcal{Y}_0,0)$, and $\mathcal{Y}_1=[\underline{y},\infty)$ for some constant $\underline{y}$. Let 
	\begin{equation}
	\begin{split}
	A^{Z=0}&=Q(\mathcal{Y},1),\quad\quad  
	A^{Z=1}=P(\mathcal{Y},1)-[Q(\mathcal{Y}_0,0)-P(\mathcal{Y}_0,0)],\\
	N^{Z=0}&=P(\mathcal{Y},0),\quad\quad  
	N^{Z=1}=Q(\mathcal{Y},0)-[Q(\mathcal{Y}_0,0)-P(\mathcal{Y}_0,0)].\\
	\end{split}
	\end{equation}
	
	Consider the following sequence of structures ${s^n}$:
	{\footnotesize
		\begin{equation}
		\begin{split}
		&Pr_{G^{s^n}}(Y_i(d,z)\in B_{dz} \quad \forall d,z\in\{0,1\},D_i(1)=1,D_i(0)=1|Z_i=z)\\
		&=\frac{G^a(Y_i(0,0)\in B_{00})G^a(Y_i(0,1)\in B_{01})\times \int_{B_{10}}p(y,1) d\mu_F(y)\times \int_{B_{11}}\min \{p(y,1),q(y,1)\}+g_c^n(y) d\mu_F(y)}{A^{Z=1-z}},
		\end{split}
		\end{equation}}
	where $G^a$ is any probability distribution, and $g_c^n(y)$ satisfies 
	\[
	g_c^n(y)=\begin{cases}
	& \min\left\{\frac{[P(\mathcal{Y}_1,1)-Q(\mathcal{Y}_1,1)]-[ Q(\mathcal{Y}_0,0)-P(\mathcal{Y}_0,0)]}{n}, p(y,1)-q(y,1) \right\}\quad if \quad y\in[\underline{y},\underline{y}+n],\\
	&0\quad otherwise.
	\end{cases}
	\]
	So $g_c^n$ is uniformly distributed over $[\underline{y},\underline{y}+n]$ to correct for the difference between $[P(\mathcal{Y}_1,1)-Q(\mathcal{Y}_1,1)]$ and $[ Q(\mathcal{Y}_0,0)-P(\mathcal{Y}_0,0)]$. Note that 
	\[
	\begin{split}
	&\int_{\mathcal{Y}}\min\{p(y,1),q(y,1)\}+g_c^n(y)d\mu_F(y)=A^{Z=1},\\
	&\int_{\mathcal{Y}} p(y,1) d\mu_F(y)=A^{Z=0}
	\end{split}
	\]
	hold, so we have the following marginal densities of $Y_i(1,0)$ and $Y_i(1,1)$ with respect to $\mu_F$:
	\[\begin{split}
	&g^{s^n}_{y_{10}}(y,1,1|Z_i=0)=p(y_{10},1),\\
	&g^{s^n}_{y_{11}}(y,1,1|Z_i=1)=\min\{p(y,1),q(y,1)\}+g_c^n(y).
	\end{split}\] 
	For never takers, consider
	{\footnotesize
		\begin{equation}
		\begin{split}
		&Pr_{G^{s^n}}(Y_{dz}\in B_{dz} \quad \forall d,z\in\{0,1\},D_i(1)=0,D_i(0)=0|Z_i=z)\\
		&=\frac{\int_{B_{01}} p(y,0) d\mu_F(y) \times \int_{B_{00}}\min\{p(y,0),q(y,0)\}d\mu_F(y)\times G^n(Y_i(1,0)\in B_{10})\times G^n(Y_i(1,1)\in B_{11}) }{N^{Z=1-z}},
		\end{split}
		\end{equation}}
	where $G^n$ is any probability distribution. This construction gives
	\[
	\begin{split}
	&g^{s^n}_{y_{01}}(y,0,0|Z_i=1)=q(y,0),\\
	&g^{s^n}_{y_{00}}(y,0,0|Z_i=0)=\min\{q(y,0),p(y,0).\}
	\end{split}
	\]
	For compliers, consider:
	\begin{equation}
	\begin{split}
	&Pr_{G^{s^n}}(Y_{dz}\in B_{dz} \quad \forall d,z\in\{0,1\},D_i(1)=1,D_i(0)=0|Z_i=z)\\
	&=\frac{\int_{B_{00}\cap B_{01}} \max\{q(y,0)-p(y,0),0\}d\mu_F(y)\times \int_{B_{10}\cap B_{11}} \max\{p(y,1)-q(y,1),0\}-g_c^n(y,1) d\mu_F(y) }{Q(\mathcal{Y}_0,0)-P(\mathcal{Y}_0,0)}.
	\end{split}
	\end{equation}
	Since $g_c^n(y,1)$ is zero on $\mathcal{Y}_1^c$, the construction of $g_c^n(y,1)$ implies $\max\{p(y,1)-q(y,1),0\}-g_c^n(y,1)\ge 0$ holds for all $y$, so $Pr_{G^{s^n}}(Y_{dz}\in B_{dz} \quad \forall d,z\in\{0,1\},D_i(1)=0,D_i(0)=0|Z_i=z)\ge 0$ holds for all $B_{dz}$ sets.

	The construction of $s^n$ ensures that $Z_i$ is a type independent instrument under $s^n$, and there are no defiers. Moreover, the measure of compliers is independent of instrument $Z_i$:
	\[
	Pr_{G^{s^n}}(D_i(1)=1,D_i(0)=0|Z_i=1)=Pr_{G^{s^n}}(D_i(1)=1,D_i(0)=0|Z_i=0)=Q(\mathcal{Y}_0,0)-P(\mathcal{Y}_0,0).
	\] 
	
	For this sequence of $s^n$, by construction, we have {\footnotesize
		\[
		\begin{split}	
		&m^{MD}(s^n)= \int 	\left[g^{s^*}_{y_{11}}(y,1,1|Z_i=1)-g^{s^*}_{y_{10}}(y,1,1|Z_i=0)\right]^2d\mu_F +\int\left[g^{s^*}_{y_{00}}(y,0,0|Z_i=0)-g^{s^*}_{y_{01}}(y,0,0|Z_i=1)\right]^2 d\mu_F\\
		&\le \int \max\{-(p(y,1)-q(y,1)),0\}^2d\mu_F+\int\max\{-(q(y,0)-p(y,0)),0\}^2 d\mu_F+ \int_{\underline{y}}^{\underline{y}+n} \left(\frac{1}{n}\right)^2 dy\\
		&\rightarrow  \int \max\{-(p(y_{11},1)-q(y_{11},1)),0\}^2d\mu_F+\int\max\{-(q(y_{01},0)-p(y_{01},0)),0\}^2 d\mu_F\equiv m^{inf}(F)
		\end{split}.
		\]}
	By (\ref{eq: append, not well behaved measure proof, minimal m^{MD}}), we see the infimum is indeed $m^{inf}(F)$. 
	
	Now, suppose there exists some $s$ that achieves this infimum, by (\ref{eq:  B.2.4 marginal constraint}) it must be the case that $\mu_F(y)$ almost surely:  
	\[
	\begin{split}
	g^s_{y_{10}}(y,1,1|Z_i=0)=q(y,1),\quad and \quad  g^s_{y_{11}}(y,1,1|Z_i=1)=\min\{p(y,1),q(y,1)\},\\
	g^s_{y_{01}}(y,0,0|Z_i=1)=p(y,0),\quad and \quad  g^s_{y_{00}}(y,0,0|Z_i=0)=\min\{p(y,0),q(y,0)\}.
	\end{split}
	\]
	Then by equation (\ref{eq: P,Q identity for AT}), the marginal distribution for compliers is pinned down by the following equations:
	\[
	\begin{split}
	g^s_{y_{11}}(y,1,0|Z_i=1)=\max\{p(y,1)-q(y,1),0\},\\
	g^s_{y_{00}}(y,1,0|Z_i=0)=\max \{q(y,0)-p(y,0),0\}.
	\end{split}
	\]
	Integrate the densities above, then we have 
	\[
	\begin{split}
	Pr_{G^s}(D_i(1)=1,D_i(0)=0|Z_i=0)= Q(\mathcal{Y}_0,0)-P(\mathcal{Y}_0,0),\\
	Pr_{G^s}(D_i(1)=1,D_i(0)=0|Z_i=1)= P(\mathcal{Y}_1,1)-Q(\mathcal{Y}_1,1).
	\end{split}
	\]
	So the measure of compliers is dependent of $Z_i$ for $s$. So $s$ does not satisfy the assumption $A^{EM-C}$. 
\end{proof}

\section{Details of LATE Empirical Illustration } \label{section: Details of Implementation of Empirical Illustration}
I use the following bandwidth and trimming sequences:
\begin{equation*}
\begin{split}
h&=\frac{s.d(Y_i)\times \log(n)}{2 n^{1/5}},\\
b&= n^{-1/4}\times \frac{\sum_{i=1}^n f_h(Y_i,1)+f_h(Y_i,0)}{n},\\
M_u&= \mathbb{F}_Y^{-1}(0.99),\\
M_l&= \mathbb{F}_Y^{-1}(0.01),
\end{split}
\end{equation*}
so $M_u$ and $M_l$ are the empirical 99-th and 1-th quantile of $\{Y_i\}_{i=1}^n$. Note that we can write 
\[
\begin{split}
&p(y,d)={f(y|D_i=d,Z_i=1)}\times Pr(D_i=d|Z_i=1),\\
and\quad &q(y,d)={f(y|D_i=d,Z_i=0)}\times Pr(D_i=d|Z_i=0).
\end{split}
\]
For the tail conditions in  Assumption \ref{assump: tail density sign}, I set
{\footnotesize
	\[
	\mathcal{Y}^{ut}_d=\begin{cases}
	[M_u,+\infty) \quad &if\quad  \mathbb{V}_{0.9}(Y_i|D_i=d,Z_i=d)\mathbb{P}(D_i=d|Z_i=d)^2 \ge \mathbb{V}_{0.9}(Y_i|D_i=d,Z_i=1-d)\mathbb{P}(D_i=d|Z_i=1-d)^2, \\
	\varnothing \quad &otherwise,
	\end{cases}
	\]}
where $\mathbb{V}_{0.9}(Y_i|D_i=d,Z_i=z)$ is the conditional empirical variance of $Y_i$, conditioned on $D_i=d,Z_i=z$ and $Y_i$ being on the 10th upper quantile, i.e. $Y_i\ge \mathbb{F}_{Y|D_i=d,Z_i=z}^{-1}(0.9)$.  $\mathbb{P}(D_i=d|Z_i=z)$ is the empirical conditional probability. Intuitively, this is a selection of tail sign in Assumption \ref{assump: tail density sign} based on how fat the tail is. If the true density $p(y,d)$ and $q(y,d)$ has sub-Gaussian tails, larger conditional variances imply fatter tails. For example, $p(y,1)$ and $q(y,1)$ are Gaussian above the 90-th quantile, and the following variance 
\[
{V}_{0.9}(Y_i|D_i=1,Z_i=1){P}(D_i=1|Z_i=1)^2 \ge {V}_{0.9}(Y_i|D_i=1,Z_i=0){P}(D_i=1|Z_i=0)^2 \]
holds for the true conditional variance $V$, the tail of $p(y,1)$ is fatter than $q(y,1)$, and thus $p(y,1)>q(y,1)$ on $[M_u,+\infty)$. Similarly, we can set the lower-end set $\mathcal{Y}^{lt}_d$ as
{\footnotesize
	\[
	\mathcal{Y}^{lt}_d=\begin{cases}
	(-\infty,M_l] \quad &if\quad  \mathbb{V}_{0.1}(Y_i|D_i=d,Z_i=d)\mathbb{P}(D_i=d|Z_i=d)^2 \ge \mathbb{V}_{0.1}(Y_i|D_i=d,Z_i=1-d)\mathbb{P}(D_i=d|Z_i=1-d)^2, \\
	\varnothing \quad &otherwise,
	\end{cases}
	\]
}
where $\mathbb{V}_{0.1}(Y_i|D_i=d,Z_i=z)$ is the conditional empirical variance of $Y_i$, conditioned on $D_i=d,Z_i=z$ and $Y_i$ being on the 10th lower quantile, i.e. $Y_i\le \mathbb{F}_{Y|D_i=d,Z_i=z}^{-1}(0.1)$.  $\mathbb{P}(D_i=d|Z_i=z)$ is the empirical conditional probability. It should be noted that my trimming band $M_l,M_u$ and tail sign $\mathcal{Y}_d^{ut}, \mathcal{Y}_d^{lt}$ are data driven in this empirical application, but Theorem \ref{thm: asymptotic property of LATE} requires these quantities to be known. This is one limitation of my results. 
\end{document}